\documentclass[12pt]{amsart}

\usepackage{amsmath,amsxtra,amssymb,amsthm,amsfonts}
\usepackage{mathrsfs}
\usepackage[T1]{fontenc}
\usepackage[utf8]{inputenc}
\usepackage{mathtools}   
\usepackage{tikz}
\usepackage{bbm}
\usetikzlibrary{arrows, automata}
\usepackage{color,hyperref}
\definecolor{dblue}{rgb}{0, 0, 0.72}

\allowdisplaybreaks

\numberwithin{equation}{section}
\newtheorem{lemma}{Lemma}[section]

\newtheorem{theorem}[lemma]{Theorem}

\newtheorem{proposition}{Proposition}[section]

\newtheorem{prob}[lemma]{Problem}
\newtheorem{rem}[lemma]{Remark}
\newtheorem{remark}[lemma]{Remark}

\newtheorem{definition}[lemma]{Definition}
\newtheorem{corollary}[lemma]{Corollary}

\newenvironment{claim}[1]{\par\noindent\underline{Claim:}\space#1}{}
\newenvironment{claimproof}[1]{\par\noindent\underline{Proof:}\space#1}{\hfill $\blacksquare$}

\newcommand{\re}{\begin{rem}\rm}
	\newcommand{\mar}{\end{rem}}

\newcommand{\ee }{\mathrm{I}\!\!1}

\renewcommand{\for}{\begin{eqnarray*}}
	
	\newcommand{\mel}{\end{eqnarray*}}

\newcommand{\lel}{\pl = \pl}

\DeclareMathOperator{\CLSI}{CLSI}
\DeclareMathOperator{\MLSI}{MLSI}

\newcommand{\pl}{\hspace{.1cm}}

\newcommand{\qd}{\end{proof}\vspace{0.5ex}}

\newcommand{\pf}{\begin{proof}}

\newcommand{\be}{\left|{\atop}}

\newcommand{\xspace}{\hbox{\kern-2.5pt}}
\newcommand{\xyspace}{\hbox{\kern-1.1pt}}

\newcommand\bra[1]{\langle  #1|}
\newcommand\ket[1]{| #1\rangle}

\newcommand{\ssubset} {\!\!\subset\! \!}

\allowdisplaybreaks

\definecolor{LightGray}{rgb}{0.94,0.94,0.94}
\definecolor{VeryLightBlue}{rgb}{0.9,0.9,1}
\definecolor{LightBlue}{rgb}{0.8,0.8,1}
\definecolor{DarkBlue}{rgb}{0,0,0.6}
\definecolor{LightGreen}{rgb}{0.88,1,0.88}
\definecolor{MidGreen}{rgb}{0.6,1,0.6}
\definecolor{DarkGreen}{rgb}{0,0.6,0}
\definecolor{DarkGrreen}{rgb}{0,0.8,0}

\definecolor{VeryLightYellow}{rgb}{1,1,0.9}
\definecolor{LightYellow}{rgb}{1,1,0.6}
\definecolor{MidYellow}{rgb}{1,1,0.5}
\definecolor{DarkYellow}{rgb}{0.8,1,0.3}
\definecolor{VeryLightRed}{rgb}{1,0.9,0.9}
\definecolor{LightRed}{rgb}{1,0.8,0.8}
\definecolor{DarkRed}{rgb}{0.8,0.2,0}
\definecolor{DarkRedb}{rgb}{0.6,0.2,0}
\definecolor{DarkLila}{rgb}{0.8,0,1}
\definecolor{Beige}{rgb}{0.96,0.96,0.86}
\definecolor{Gold}{rgb}{1.,0.84,0.}
\definecolor{Goldb}{rgb}{0.7,0.3,0.5}
\definecolor{MyYellow}{rgb}{1.,0.84,0.8}

\def\mod{\,\, {\rm mod}\,\,}

\def\11{\mathbb{I}}

\usepackage{tikz-cd}

\usepackage{graphicx}

\DeclareRobustCommand\openone{\leavevmode\hbox{\small1\normalsize\kern-.33em1}}

\renewcommand{\be}{\begin{equation}}
	\renewcommand{\ee}{\end{equation}}
\newcommand{\bea}{\begin{eqnarray}}
	\newcommand{\eea}{\end{eqnarray}}
\newcommand{\beas}{\begin{eqnarray*}}
	\newcommand{\eeas}{\end{eqnarray*}}

\usepackage{amsmath}
\usepackage{ams math}
\usepackage{amssymb}
\usepackage{amsthm}
\usepackage{dsfont}
\usepackage{upgreek}
\usepackage{enumitem}

\usepackage{array}
\usepackage{ams fonts}
\usepackage{graphics}
\usepackage[utf8]{inputenc}
\newtheorem*{theorem*}{Theorem}
\newtheorem*{remark*}{Remark}
\newtheorem*{lemma*}{Lemma}
\newtheorem*{notation*}{Notation}
\newtheorem*{cor*}{Corollary}
\newtheorem*{note*}{Note}
\newtheorem*{prop*}{Proposition}
\newtheorem*{example*}{Example}
\usepackage{tikz-cd}
\usepackage{amsthm}

\usepackage[margin=1.0 in]{geometry}

\title{Entropy Decay Estimates For Collective Noise Models}
\begin{document}
	\author[Y. Chen]{Yidong Chen}
	\author[M. Junge]{Marius Junge} \thanks{MJ was partially supported by  NSF
		Grant DMS 1800872 and NSF RAISE-TAQS 1839177}
	
	\begin{abstract} 
		One of the challenges in quantum information science is to control open quantum systems with a large number of qubits. An important aspect of many-body systems is the emergence of \textit{collective phenomena}.
		One \textit{collective noise model} is an open atomic system in an electromagnetic environment. This model was considered by Dicke in the 50's \cite{Dicke}. In this paper, we study the entropic decay in Dicke's model and other related collective noise models. Specifically, we develop a general framework to estimate the \textit{spectral gap} and \textit{modified logarithmic Sobolev constant} of these collective noise models. In addition, we study the necessary mixing conditions a general Dicke's model must satisfy in order to have a unique equilibrium state. The combination of representation theory and entropy estimates could be a guideline to study more general collective noise models.
	\end{abstract}

	\maketitle
	\tableofcontents
	\section{\large Introduction}\label{section:intro}
	Understanding and controlling large scale open quantum systems is a challenging problem in quantum computation. To build robust and scalable quantum devices, it is necessary to understand how many-body systems interact with their environment. In the weak-coupling limit, it is well-known that the dissipative dynamics of an open quantum system can be described by a Lindbladian \cite{GKS}\cite{L}\cite{D1}\cite{D2}\cite{D3}\cite{D83}. A complex many-body system typically interacts with the environment in a collective way: for example, an atomic array in an electromagnetic field. Such collective interactions often lead to emergent phenomena such as macroscopic quantum coherence. These emergent phenomena do not exist if each atom interacts with the environment separately. One example of such an emergent phenomenon is Dicke's model \cite{Dicke} of an open atomic system in an electromagnetic environment. Although Dicke's model is well studied in quantum optics \cite{DKW}\cite{KSW}\cite{GW}\cite{MAG}\cite{SMAG}, the spectral gap and decay properties are not well enough understood as of now.
	
	In this paper, we study the entropic decay in Dicke's original model and other related collective noise models. Our noise models differ from the usual Davies generator models. In a typical Davies generator noise model, the Hamiltonian of the many-body system includes complicated interactions while the Lindbladian is given as a sum of local Lindbladians \cite{entropyDavies}\cite{BCLPR}\cite{ACCRDSF}\cite{KB}. As such, these Davies generator models can not capture the collective interaction between the many-body system and the environment. Our Lindbladians do not admit decompositions into sums of local Lindbladians. Hence they are more suitable to study collective phenomena. On the other hand, the Hamiltonian in our model is non-interacting and is simply given by a sum of strictly local terms. The comparison between the collective noise models and models using Davies generators are summarized in the following table:
	\begin{center}
		\begin{tabular}{||m{2.2cm}||m{6cm}|m{7cm}||}
			\hline Features & Collective Noise Models & Davies Generators \\
			\hline\hline
			Hamiltonian & Non-interacting. \newline Can be written as a sum of single-body terms. & Interacting. \newline Can be written as a sum of local interactions.\\ \hline
			Lindbladian & Collective. \newline Necessarily include non-local interactions. & Separable. \newline Can be written as a sum of single-body Lindbladian. \\
			\hline
		\end{tabular}
	\end{center}

	The collective noise models are complementary to the Davies generators. Since collective interactions arise naturally in many-body systems, our work fills in a gap in the current understanding of the dissipative dynamics of open quantum systems. In the remainder of the introduction, we state the main results and provide the necessary technical background to study quantum entropy decay. The paper is organized as follows. Section \ref{section:single} and Section \ref{section:abstract} develop a general framework to study entropic decay in collective noise models. This general framework can be applied to noise models other than the original Dicke's model.  Section \ref{section:su(2)} and Section \ref{section:gap} apply the general framework to collective noise models of Dicke's type. Many of the noise models we study in this paper are naturally non-primitive (see the next part of the introduction for a definition of primitivity). In Section \ref{section:related}, we study collective noise models driven by more than one Lindbladian and show that in general these noise models are primitive. We also study the entropic decay of these more complicated models. In Section \ref{section:Davies}, we present a more in-depth comparison between (non-primitive) collective noise models and Davies generator models. Finally we conclude in Section \ref{section:discussion}.
	\subsection{Main Results} The original Dicke's model describes the collective interaction between a many-body system of a large number of two-level atoms and an electromagnetic environment. It is well-known that the time-evolution in Dicke's model can be described by the following Markovian master equation \cite{GW}\cite{KSW}\cite{SCG}:
	\begin{equation}\label{equation:atomic}
		\frac{d}{dt}\rho = -\frac{i}{\hbar}[\mathcal{H},\rho] + \sum_{i,j = 1}^N \frac{\Gamma_{ij}}{2}(2\sigma_{ge}^i \rho \sigma^j_{eg} - \rho\sigma^i_{eg}\sigma^j_{ge} - \sigma^j_{eg}\sigma^i_{ge}\rho)
	\end{equation}
	where $\mathcal{H}$ is the Hamiltonian of the atomic system, $\sigma^i_{ge} = \ket{g_i}\bra{e_i}$ acts on the $i$-th atom where $\ket{e_i}$ and $\bra{g_i}$ are its excited and ground states. The matrix $(\Gamma_{ij})$ can be diagonalized and its eigenvalues describe the decay rates of the atomic system. The non-Hamiltonian part on the right hand side of the equation \ref{equation:atomic} describes the noise obtained from the system's interaction with the environment. By diagonalizing the matrix $(\Gamma_{ij})$ and assuming a two-way heat exchange between the system and the bath, the corresponding Lindbladian  for each eigenvalue $\Gamma_\nu$ has the following form:
	\begin{equation}\label{equation:dicke}
		\mathcal{L}_{\mathcal{O}_\nu}\rho:=\frac{\Gamma_\nu}{2}(2\mathcal{O}_{\nu}^*\rho \mathcal{O}_\nu - \rho \mathcal{O}_\nu^*\mathcal{O}_\nu - \mathcal{O}_\nu^*\mathcal{O}_\nu\rho) + \frac{\Gamma_\nu^{-1}}{2}(2\mathcal{O}_{\nu}\rho \mathcal{O}_\nu^* - \rho \mathcal{O}_\nu\mathcal{O}_\nu^* - \mathcal{O}_\nu\mathcal{O}_\nu^*\rho) 
	\end{equation}
	where $\mathcal{O}_\nu$'s act on the $N$-qubits and describes the physical process of collective emission \cite{MAG}\cite{SMAG}. When the system is put on a one-dimensional lattice, the collective emission operator can be written as:
	\begin{equation}\label{equation:rotationparam}
		\mathcal{O}_\nu = \sum_{1\leq j \leq N}e^{i\theta_j}\pi_{(j)}(a)
	\end{equation} 
	where $\theta_j\in [0,2\pi]$ are rotation angles, $\pi_{(j)}(a) := 1\otimes...\otimes a \otimes...\otimes 1$ and $a = \begin{pmatrix}
		0 & 1\\ 0 & 0
	\end{pmatrix}\in\mathfrak{su}(2)$. A single operator $\pi_{(j)}(a)$ describes the photon emission of one single atom. After a suitable rotation (c.f. Equation \ref{equation:rot}), the Lindbladian $\mathcal{L}_{\mathcal{O}_\nu}$ (c.f. Equation \ref{equation:dicke}) can be written in the canonical form of a Lindbladian \cite{CM1}\cite{CM2}:
	\begin{equation}
		\mathcal{L}^\beta_N=e^{\beta / 2}L_{\pi_N(a)} + e^{-\beta /2}L_{\pi_N(a^*)}
	\end{equation}
	where $L_{\pi_N(a)}x = 2\pi_N(a^*)x\pi_N(a) - \pi_N(a^*)\pi_N(a)x - x\pi_N(a^*)\pi_N(a)$ and $
		\pi_N(a) = \sum_{1\leq j \leq N}\pi_{(j)}(a) $. 
	The parameter $\beta > 0$ relates to the decay rate $\Gamma_\nu$ by: $\beta = 2\log\frac{\Gamma_\nu}{2}$. Typically $\beta$ is interpreted as the inverse temperature. 
	
	We will take $\pi_N(h):=\pi_N\bigg(\begin{pmatrix} 1 & 0 \\ 0 & -1\end{pmatrix}\bigg)$ as the system's Hamiltonian. This Hamiltonian is non-interacting and is a sum of single-body Hamiltonians. The corresponding Gibbs state is given by \begin{equation}
		d_N := N_\beta \exp(-\frac{\beta}{2}\pi_N(h))
	\end{equation}
	where $N_\beta$ is the normalization constant. The state $d_N$ is an invariant (i.e. equilibrium) state of the open system, i.e. $\mathcal{L}^{\beta ^*}_N(d_N) = 0$. The state can also be completely factorized in terms of a tensor product:
	\begin{equation*}
		d_N := d_\beta^{\otimes N} = \begin{pmatrix}
			\frac{e^{-\beta / 2}}{e^{\beta / 2} + e^{-\beta / 2}} & 0 \\ 0 & \frac{e^{\beta / 2}}{e^{\beta / 2} + e^{-\beta /2}}
		\end{pmatrix}^{\otimes N}
	\end{equation*}
	
	For every inverse temperature $\beta$, the Lindbladian $\mathcal{L}^\beta_N$ satisfies the $d_N$-detailed balance condition. This is a direct consequence of the underlying $\mathfrak{su}(2)$ representation theory, which we will discuss later. We refer to the seminal papers \cite{CM1}\cite{CM2} for the general theory of $\sigma$-detailed balanced Lindbladians and the crucial inner product (KMS-inner product with respect to the reference state $\sigma$):
	\begin{equation*}
		\langle x,y\rangle_\sigma := tr(\sigma^{1/2}x^*\sigma^{1/2}y)
	\end{equation*}
	For the fixed reference state $\sigma$, a Lindbladian $\mathcal{L}$ is KMS-symmetric if 
	\begin{equation*}
		\langle \mathcal{L}x, y \rangle_\sigma = \langle x, \mathcal{L}u\rangle_\sigma
	\end{equation*}
	A Lindbladian $\mathcal{L}$ admits a spectral gap $\lambda_2(\mathcal{L}) \geq c$ if there exists a constant $c > 0$ such that
	\begin{equation*}
		c\langle x - E_{fix}x, x-E_{fix}x\rangle_\sigma \leq \langle \mathcal{L}x, x\rangle_\sigma
	\end{equation*}
	where $E_{fix}$ projects onto the kernel of $\mathcal{L}$. For a KMS-symmetric $\mathcal{L}$, its kernel is a closed $*$-algebra, and $E_{fix}$ becomes the so-called conditional expectation. Beyond the spectral gap, the decay to equilibrium is measured with the help of the relative entropy:
	\begin{equation*}
		D(\rho | \sigma) := tr\big(\rho(\log\rho - \log\sigma)\big)
	\end{equation*}
	The constant $\MLSI(\mathcal{L})$ is the largest constant $c > 0$ such that for all $\rho$
	\begin{equation*}
		cD(\rho|E^*_{fix}\rho) \leq tr\big(\mathcal{L}^*(\rho)(\log\rho - \log E^*_{fix}\rho)\big)
	\end{equation*}
	The right hand side is called the entropy production \cite{spohn}. Here the (trace-)adjoint $E^*_{fix}$ and $\mathcal{L}^*$ acts on the densities. Equivalently, we have exponential decay
	\begin{equation*}
		D(e^{-t\mathcal{L}}\rho | E^*_{fix}\rho)\leq e^{-\MLSI(\mathcal{L})t}D(\rho | E^*_{fix}\rho)
	\end{equation*}
	The complete version $\CLSI(\mathcal{L}) := \inf_{d\geq 0}\MLSI(\mathcal{L}\otimes id_d)$ was introduced in \cite{Fisher} in order to ensure tensor stability. Recently, a wealth of results \cite{Fisher, MBLGMJ, LGCR, GJL, BCLPR, SU(2)} provide lower bounds for the CLSI constant. For $\beta = 0$ (the infinite temperature limit), tools from harmonic analysis, Lie groups and foliation can be used to obtain a uniform spectral gap and CLSI constant. \cite{LGCR, MBLGMJ}
	\begin{theorem}
		The family of Lindbladians $\mathcal{L}^0_N$ admits a uniform lower bound on the spectral gaps and the CLSI constants.
	\end{theorem}
	Since $\mathcal{L}^\beta_N$ depends continuously on $\beta$, one expects a similar result to hold at finite temperature (at least for high temperature). Indeed, the change of measure argument \cite{JLR} provides such a link. However, the constants obtained by the Hooley-Stroock argument in \cite{JLR} are not uniform in dimension. Despite an extensive effort by the second named author, the uniformity in $N$ simply could not be achieved. This leads to a polynomial bound:
	\begin{theorem}\label{theorem:main}
		Suppose $\beta > 0$. There exists constants $C_1(\beta), C_2(\beta) > 0$ such that:\begin{equation}
			\frac{C_1(\beta)}{N^2} \leq \CLSI(\mathcal{L}^\beta_N) \leq 2\lambda_2(\mathcal{L}^\beta_N)\leq \frac{C_2(\beta)}{N}\pl.
		\end{equation}
	\end{theorem} 
	It is an open problem to calculate the precise orders of the spectral gap and the CLSI constant. The same result holds for the Lindbladian $\mathcal{L}^\beta_{\mathcal{O}_\nu}$ (c.f. Equation \ref{equation:dicke}). The proof has two parts. The first part relies on the representation theory of $\mathfrak{su}(2)$. And the second part follows essentially from estimates of modified logarithmic Sobolev constants on finite graphs. In addition, the upper bound of the spectral gaps can be achieved by tuning the off-diagonal coefficients between quantum number $N$ and $N-2$. The new dimension-dependenet bounds show that there are certain densities embedded in a system of $N$ qubits which remain stable on a time scale of $O(N)$. It would be interesting to know how these states can be prepared and used for specific quantum tasks.
	
	In addition, the same argument gives us estimates on the CLSI constants of other Lindbladians as well. In particular, we have
	\begin{theorem}\label{theorem: infinite}
		Let $a_N(\gamma) := \pi_N(a)(\pi_N(a^*)\pi_N(a))^\gamma$ where the parameter $\gamma \geq 1, \gamma\in\mathbb{Z}$. Consider the Lindbladian:
		\begin{equation}
			\mathcal{L}^\beta_{N,\gamma} := e^{\beta / 2}L_{a_N(\gamma)} + e^{-\beta /2}L_{a_N^*(\gamma)}
		\end{equation}
		where $L_{a_N(\gamma)}x := 2a_N(\gamma)^*x a_N(\gamma) -a_N(\gamma)^*a_N(\gamma)x - xa_N(\gamma)^*a_N(\gamma)$. In addition, fix again the tensor product state $d_N = d_\beta^{\otimes N}$ as the reference state. Then there exist constants $C_3(\beta, \gamma) > 0$ such that
		\begin{equation}
			2\lambda_2(\mathcal{L}^\beta_{N,\gamma})\geq\CLSI(\mathcal{L}^\beta_{N,\gamma}) \geq C_3(\beta, \gamma)\pl.
		\end{equation}
	\end{theorem}
	The dimension-free bound is important when we take $N$ to be infinity. In this limit, our result gives an estimate of entropy decay in a type $III_\lambda$ ($\lambda = e^{-\beta}$) factor. The Lindbladian we considered here is generated by two infinitesimal generators. 
	
	In general, an atomic system may have more than one decay mode. Dissipation in such a system can be described by a finite sum of Lindbladians: $\sum_\nu \mathcal{L}^\beta_{\mathcal{O}_\nu}$. When the system is on a one-dimensional lattice, we show that under generic conditions, the system is primitive:
	\begin{theorem}
		Given a  Lindbladian $\mathcal{L}^\beta:=\sum_\nu \mathcal{L}^\beta_{\mathcal{O}_\nu}$. If there exists a pair of summands $\mathcal{L}^\beta_{\mathcal{O}_\nu}$ and $\mathcal{L}^\beta_{\mathcal{O}_\mu}$ whose rotation parameters (c.f. Equation \ref{equation:rotationparam}) $\{\theta_j\}_{1\leq j \leq N}$ and $\{\varphi_j\}_{1\leq j \leq N}$ satisfy the condition:
		\begin{equation*}
			(\theta_i - \theta_j) - (\varphi_i - \varphi_j) \neq 0 \mod 2\pi
		\end{equation*}
		for $i\neq j$, then the Lindbladian $\mathcal{L}^\beta$ is primitive.
	\end{theorem}
	Since the condition in the theorem is satisfied by generic rotation parameters, the theorem shows that generically the generalized Dicke-type system has a unique equilibrium state. This result is proved in Section \ref{section:related}.
	
	Lastly, using quantum expander and quantum Fourier transform, we are able to prove CLSI inequality for a generalized Dicke's model. 
	\begin{theorem}
		There exists a finite number $m(\beta) = O(1)$ such that for all $N$ there exists a finite set of unitaries $\{U_j\}_{1\leq j \leq m(\beta)}$ such that:
		\begin{enumerate}
			\item Each unitary $U_j$ commutes with the density $d_N$;
			\item $D(\rho | E^*_\mathbb{C}\rho) \leq C(N)\sum_{1\leq j \leq m}I_{\mathcal{L}^\beta_{U_j}}(\rho)$ where $\mathcal{L}^\beta_{U_j} := e^{-\beta / 2}L_{U^*_j\pi_N(a)U_j} + e^{\beta / 2}L_{U^*_j\pi_N(a)U_j}$. 
		\end{enumerate}
	\end{theorem}
	To prove this CLSI inequality, it is necessary to use quantum expanders to efficiently reduce the fixed point algebra of $\mathcal{L}^\beta_N$. The proof is rather technical and the unitaries constructed in the proof may not be easy to implement in experiments. However, the techniques used in the proof are novel. It is the authors' hope that the same techniques can be used to obtain CLSI inequalities for more realistic models in the future.
	
	\section{\large CLSI of a Single Self-Adjoint Operator}
	\label{section:single}
	In this section, we prove $\CLSI$ for Lindbladian generated by a single self-adjoint operator:
	\begin{equation}
		\mathcal{L}_A(x)  := 2AxA - A^2 x - xA^2
	\end{equation}
	where $A$ is a self-adjoint operator with discrete spectrum and each eigenvalue has multiplicity 1. 
	For simplicity, we first restrict ourselves to the case where $\mathcal{L}_A$ acts on a matrix algebra $\mathbb{M}_d$.
	\begin{align}
		\mathcal{L}_A: \mathbb{M}_{d}\rightarrow\mathbb{M}_{d}
	\end{align}
	In this case, the CLSI constant will be related to the spectral distribution of  $A$. Later, we will use a transference principle to generalize this result to finite von Neumann algebras.
	
	 The key to prove this result is the connection between complete return time and $\CLSI$ constant established in \cite{MBLGMJ}. More precisely, for a symmetric quantum Markov semigroup, it is proved that a lower Ricci curvature bound along with a complete return time estimate implies the semigroup satisfies $\CLSI$. 
	
	To apply this result, we first make the simple observation that $\mathcal{L}_A$ can be represented as a Schur multiplier. Recall the definition of a Schur multiplier on a matrix algebra:
	\begin{definition}
		Let $\mathbb{M}_n$ be a $n\times n$-matrix algebra and $a:=(a_{ij})_{i,j=1}^n\in\mathbb{M}_n$ be a matrix. Then the Schur multiplier $T_a$ associated with $a$ is\begin{equation*}
			T_a : \mathbb{M}_n\rightarrow\mathbb{M}_n : T_a((x_{ij})) = (a_{ij}x_{ij})
		\end{equation*}
	\end{definition}
	\begin{lemma}\label{lemma:schur} Let $A = \sum_{i}\lambda_i e_i$ be the spectral decomposition of $A$ with distinct eigenvalues. Then the Lindbladian $\mathcal{L}_A: \mathbb{M}_{d} \rightarrow \mathbb{M}_{d}$ is given by a Schur multiplier.
		\begin{equation}
			\mathcal{L}_A(x) = -\sum_{i,j} (\lambda_i-\lambda_j)^2e_i xe_j
		\end{equation}
	The fixed point algebra is the subalgebra of diagonal matrices: $\mathbb{M}_{diag} := span\{e_i: 1\leq i\leq d\}$. In particular, the quantum Markov semigroup generated by $\mathcal{L}_A$ is not primitive. 
	\end{lemma}
	\begin{proof}
		By a direction computation we have:
		\begin{align}
			\begin{split}
				\mathcal{L}_A(x) &= 2AxA - A^2x - xA^2 = -\sum_i \lambda_i^2 e_i x - \sum_i \lambda_i^2 xe_i + 2\sum_{i,j} \lambda_i \lambda_j e_i xe_j
				\\
				& = -\sum_{i,j} (\lambda_i^2 e_i xe_j + \lambda_j^2 e_i xe_j - 2 \lambda_i\lambda_j e_i xe_j) = -\sum_{i,j}(\lambda_i - \lambda_j)^2 e_ixe_j
			\end{split}
		\end{align}
		where in the third equation we used partition of unity $\sum_i e_i = 1$. From the Schur multiplier representation of $\mathcal{L}_A$, it is clear that $\mathcal{L}_A(x) = 0$ if and only if $x\in \mathbb{M}_{diag}$. Hence the fixed point algebra is the subalgebra of diagonal matrices. 
	\end{proof}

	It is shown in \cite{MBLGMJ} that the Ricci curvature of the Lindbladian associated with a Schur multiplier is bounded below by 0. In this case, the paper showed that the $\CLSI$ constant is inversely proportional to the complete return time. Therefore in the remainder of this section, we will calculate the complete return time of $\mathcal{L}_A$ in terms of the spectral distribution of $A$. 
	
	It turns out that the correct definition of return time for a nonprimitive quantum Markov semigroup $T_t:\mathcal{M}\rightarrow\mathcal{M}$ is given by \cite{Fisher}\cite{LGMJNL}\cite{MJJP}
	\begin{definition}
		For a nonprimitive quantum Markov semigroup $T_t: \mathcal{M}\rightarrow\mathcal{M}$ where $\mathcal{M}$ is a finite von Neumann algebra and let $E:\mathcal{M}\rightarrow\mathcal{N}$ be the conditional expectation onto the fixed point algebra $\mathcal{N}\ssubset\mathcal{M}$, the complete bounded return time of $T_t$ is given by:\begin{equation}
			t_{cb}(\epsilon) := \inf\{t \geq 0: ||T_t - E : L_\infty^1(\mathcal{N}\ssubset\mathcal{M})\rightarrow L_\infty(\mathcal{M})||_{cb} \leq \epsilon \}
		\end{equation}
		where $0 < \epsilon < 1$, and $L_\infty^1(\mathcal{N}\ssubset\mathcal{M})$ is the amalgamated $L_\infty$ space equipped with the norm:
		\begin{equation*}
			||x||_{L^1_\infty} := \sup_{\substack{||a||_2 = ||b||_2 = 1 \\ a,b \in L_2(\mathcal{N})}} ||axb||_1
		\end{equation*}
		For $\epsilon = 1/2$, we write $t_{cb} := t_{cb}(1/2)$
	\end{definition}
	For the motivation of this definition, see \cite{Fisher} and \cite{MBLGMJ}. Below we present the key lemma to calculate the complete return time in the finite dimensional case  \cite{Fisher}. It extends the classical result that relates cb-norm of a quantum channel with the norm of its Choi matrix. Fix a pair of unital algebras $\mathcal{N}\ssubset\mathcal{M}$, recall a map $T:\mathcal{M}\rightarrow\mathcal{M}$ is a $\mathcal{N}$-bimodule map if for all $a,b\in\mathcal{N}$ and $x\in\mathcal{M}$ we have $T(axb) = aT(x)b$. 
	\begin{lemma}
		Let $\mathbb{M}_d$ be a matrix algebra, and fix a set of matrix units $\{e_{ij}\}_{1\leq i,j \leq d}$ of $\mathbb{M}_d$. Let $\mathcal{N}\ssubset\mathbb{N}_d$ be a subalgebra.Then for a completely bounded $\mathcal{N}$-bimodule map: $T:L_1(\mathbb{M}_d)\rightarrow L_\infty(\mathbb{M}_d)$ we have\begin{equation}
			||T:L^1_\infty(\mathcal{N}\ssubset\mathbb{M}_d)\rightarrow L_\infty(\mathbb{M}_d)||_{cb} \leq ||\chi_T||_{\mathbb{M}_d\overline{\otimes}\mathbb{M}_d}
		\end{equation}
	where $\chi_T := \sum_{i,j} e_{ij}\otimes T(e_{ij})$ is the Choi matrix of $T$.
	\end{lemma}
	For a proof, see \cite{Fisher}\cite{LGMJNL}. After applying this lemma to the Schur multipliers, we have
	\begin{lemma}
		With the same set-up as the previous lemma, let $T((e_{ij})) = (t_{ij}e_{ij})$ be a Schur multiplier that fixes the diagonal subalgebra $\mathbb{M}_{diag}$. Assume $T$ is a $\mathbb{M}_{diag}$-bimodule map. Then 
		\begin{equation}
			||T||_{1\rightarrow\infty, cb}:=||T:L_1(\mathbb{M}_{diag}\ssubset\mathbb{M}_d)\rightarrow L_\infty(\mathbb{M}_d)||_{cb} \leq ||(t_{ij})||_{\mathbb{M}_d}
		\end{equation} 
	\end{lemma}
	\begin{proof}
		Consider the following map:\begin{align}
			\begin{split}
				\varphi: &\mathbb{M}_d\rightarrow\mathbb{M}_d\otimes\mathbb{M}_d
				\\
				&e_{ij} \mapsto e_{ij}\otimes e_{ij}
			\end{split}
		\end{align}
		$\varphi$ is a non-unital $*$-homomorphism. Since $\varphi(1) = \sum_i e_{ii}\otimes e_{ii}$, $\varphi$ is also an isometry. By the previous lemma, $||T||_{1\rightarrow\infty, cb} 
		\leq ||\chi_T||_{\mathbb{M}_d\overline{\otimes}\mathbb{M}_d}$. Since $\chi_T = \sum_{i,j}e_{ij}\otimes T(e_{ij}) = \sum_{i,j}e_{ij}\otimes t_{ij}e_{ij} = \varphi((t_{ij}))$, we have:
		\begin{equation}
			||T||_{1\rightarrow\infty, cb} \leq ||\chi_T||_{\mathbb{M}_d\otimes\mathbb{M}_d} = ||(t_{ij})||_{\mathbb{M}_d}\qedhere
		\end{equation}\qedhere\end{proof}
	Let's go back to the complete return time of the quantum Markov semigroup $\{T_t := e^{-t\mathcal{L}_A}:\mathbb{M}_d\rightarrow\mathbb{M}_d\}_{t\geq 0}$. Let $E:\mathbb{M}_d\rightarrow \mathbb{M}_{diag}$ be the conditional expectation onto the fixed point algebra, the previous lemma gives the following complete return time estimate.
	\begin{corollary}
		Let $\delta = min_{i\neq j}|\lambda_i - \lambda_j|$ be the smallest gap in the spectral distribution of $A$, then the complete return time satisfies the following bound:
		\begin{equation}
			t_{cb}(\epsilon) \leq \frac{\epsilon^2\pi}{\delta^2}
		\end{equation}
		In particular, $t_{cb} = t_{cb}(1/2) \leq \frac{4\pi}{\delta^2}$.
	\end{corollary}
	\begin{proof}
		Since $(T_t - E)x = \sum_{ij} e^{-t(\lambda_i - \lambda_j)^2}x_{ij} - \sum_i x_{ii} = \sum_{i\neq j}e^{-t(\lambda_i - \lambda_j)^2}x_{ij}$, by the previous lemma
		\begin{align}
			\begin{split}
				||T_t - E||_{1\rightarrow\infty, cb} &= ||(e^{-t(\lambda_i - \lambda_j)^2})_{i\neq j}||_{\mathbb{M}_d}=||(\sum_{k=1}^{d-1} + \sum_{k=-1}^{-d+1})\sum_i e^{-t(\lambda_i - \lambda_{i+k})}e_{i,i+k}||
				\\
				&\leq2\sum_{k=1}^{d-1}||\sum_i e^{-t\delta^2k^2}e_{i,i+k}||\leq 2\sum_{k=1}^{d-1}e^{-t\delta^2k^2}
				\leq \int_{-\infty}^\infty dx e^{-t\delta^2x^2} = \frac{\sqrt{\pi}}{\delta\sqrt{t}}
			\end{split}
		\end{align}
		Hence $t_{cb}(\epsilon) \leq \frac{\epsilon^2\pi}{\delta^2}$ and $t_{cb} \leq \frac{4\pi}{\delta^2}$. 
	\end{proof}
	To state the main theorem of this section, we introduce the following notations.
	\begin{definition}
		For two variables $F,G$ (1) $F \lessapprox G$ if and only if there exists an absolute constant $C >0$ such that $F \leq CG$. (2) $F \gtrapprox G)$ if and only if there exists an absolute constant $c> 0$ such that $F \geq cG$. (3) $F \sim G$ if and only if there exists absolute constants $c, C >0$ such that $cG \leq F \leq CG$.
	\end{definition}
	\begin{theorem}\label{theorem:s.a.CLSI}
		Let $\mathcal{L}_A(x) = 2AxA - A^2x - xA^2:\mathbb{M}_d\rightarrow\mathbb{M}_d$ be the generator of a non-primitive quantum Markov semigroup where $A$ is a self-adjoint operator with discrete spectrum of multiplicity 1. Then the $\CLSI$ constant of $\mathcal{L}_A$ is proportional to the spectral gap of $\mathcal{L}_A$:\begin{equation}
			\CLSI(\mathcal{L}_A) \sim \lambda(\mathcal{L}_A)
		\end{equation}
	\end{theorem}
	\begin{proof}
		Let $\delta = \min_{i\neq j}|\lambda_i - \lambda_j|$. Since $\mathcal{L}_A(x) = -\sum_{i,j} (\lambda_i - \lambda_j)^2e_ixe_j$, the spectral gap $\lambda(\mathcal{L}_A) \sim \delta^2$.
		
		By the CLSI estimate from \cite{MBLGMJ}, since Schur multiplier's geometric Ricci curvature is bounded below by 0 (for a definition of geometric Ricci curvature see \cite{MBLGMJ}) then 
		\begin{eqnarray}
			\CLSI(\mathcal{L}_A) \geq\frac{\delta^2}{16\pi}
		\end{eqnarray} 
		Hence $\CLSI(\mathcal{L}_A) \gtrapprox \lambda(\mathcal{L}_A)$. By a well known result \cite{KT}, we have $\CLSI(\mathcal{L}_A)\leq 2\lambda(\mathcal{L})$. Hence $\CLSI(\mathcal{L}_A) \sim \lambda(\mathcal{L}_A)$.
	\end{proof}
	\begin{corollary}
		Let $\mathcal{L}_A(x):=2AxA - A^2x - xA^2 :\mathcal{M}\rightarrow\mathcal{M}$ be the generator of a quantum Markov semigroup (not necessarily primitive) on a finite von Neumann algebra, and let $A$ be a self-adjoint operator with spectral decompostion $A = \sum_{1\leq i \leq d}\lambda_i e_i$. Then the same conclusion of the last theorem holds.
	\end{corollary}
	\begin{proof}
		By the same calculation as lemma \ref{lemma:schur}, we have $\mathcal{L}_A(x) = -\sum_{1\leq i,j\leq d}(\lambda_i - \lambda_j)^2 e_i xe_j$. Let $\mathcal{L}_d$ be the Lindbladian on $\mathbb{M}_d$ given by: $\mathcal{L}_d(x) := -\sum_{1\leq i,j\leq n}(\lambda_i - \lambda_j)^2 e_i xe_j$. Consider the non-unital trace preserving map:
		\begin{align}
			\begin{split}
				\varphi:& \mathcal{M}\rightarrow \mathbb{M}_d\overline{\otimes}\mathcal{M}
				\\
				&\varphi(e_i xe_j):=e_{ij}\otimes e_ixe_j
			\end{split}
		\end{align}
		Then we have
		\begin{align}
			\begin{split}
				(e^{-t\mathcal{L}_d}\otimes id_\mathcal{M})\circ\varphi(x) &= \sum_{i,j}e^{-t(\lambda_i - \lambda_j)^2}e_{ij}\otimes e_ixe_j
				\\
				&=\sum_{i,j}e_{ij}\otimes e^{-t(\lambda_i - \lambda_j)^2}e_ixe_j = \sum_{i,j}e_{ij}\otimes e^{-t\mathcal{L}_A}x
				\\
				&=\varphi(e^{-t\mathcal{L}_A}x)
			\end{split}
		\end{align}
		In addition, $\varphi$ is compatible with the projection onto the fixed point algebra. Hence by the definition of the CLSI constant, we have $\CLSI(\mathcal{L}^\mathcal{M}_A) \sim \lambda(\mathcal{L}^\mathcal{M}_A)$ 
	\end{proof}
	\begin{remark*} In \cite{LGCR}, Gao and Rouzé have shown that $\frac{\lambda(\mathcal{L}_A)}{C_{cb}(\mathcal{M}:\mathcal{N})} \leq CLSI(\mathcal{L}_A)\leq 2\lambda(\mathcal{L}_A)$ where $C_{cb}(\mathcal{M}:\mathcal{N})$ is the cb index of the pair $\mathcal{N}\ssubset\mathcal{M}$. Our result gives a tighter bound because the cb index of the diagonal subalgebra is the dimension of the underlying Hilbert space.
	\end{remark*}
	
	\section{\large CLSI of KMS-Symmetric QMS with Single Bohr Frequency}\label{section:abstract}
	Following \cite{CM1}\cite{CM2}, we consider a Lindbladian associated with a single Bohr frequency in this section.
	\begin{align}
		\begin{split}
			&\mathcal{L}^\beta:\mathcal{M}\rightarrow\mathcal{M}
			\\
			&\mathcal{L}^\beta x := e^{\beta/2}L_ax + 
			e^{-\beta/2}L_{a^*}x
		\end{split}
	\end{align}
	where $L_ax := 2a^*xa- a^*a x - xa^*a$ and $\mathcal{M}$ is a finite von Neumann algebra. We assume $\beta > 0$. Fix a faithful state $\phi$ on $\mathcal{M}$ such that its density $d_\phi$ is an fixed point state under the predual $\mathcal{L}^\beta_*:L_1(\mathcal{M})\rightarrow L_1(\mathcal{M})$ (i.e. $\mathcal{L}^\beta_*(d_\phi) = 0$). The KMS-inner product on $\mathcal{M}$ with respect to $d_\phi$ is given by:
	\begin{equation}
		\langle x, y\rangle_\phi := \tau(d_\phi^{1/2}x^*d_\phi^{1/2}y)
	\end{equation}
	where $\tau$ is the canonical trace on $\mathcal{M}$. We assume $\mathcal{L}^\beta$ is self-adjoint with respect to this KMS-inner product:
	\begin{equation}
		\langle \mathcal{L}(x), y \rangle_\phi = \langle x, \mathcal{L}(y)\rangle_\phi
	\end{equation}

	We make the crucial assumption that the modular automorphism associated with $\phi$ satisfies the following equation
	\begin{equation}\label{equation:modular}
		\sigma^\phi_t(a) = d_\phi^{it}ad_\phi^{-it} = e^{i \beta t}a
 	\end{equation}
 	The modular automorphism assumption implies that the Lindbladian commutes with modular operator: $[\mathcal{L}^\beta, \Delta_\phi] = 0$. By a result from \cite{CM1}, this implies that the Lindbladian satisfies $\phi$-detailed balance condition. The state $\phi$ induces a nontrivial modular automorphism group. In the following, we will use a $2\times 2$-matrix trick to cancel this modular automorphism group. Using this technical tool, we will provide a direct estimate of the $\CLSI$ constant. 
 	
 	Before stating the next proposition, recall the definition of the difference quotient:
 	\begin{definition}
 		Let $f:\mathbb{R}_+\rightarrow\mathbb{R}$ be a continuously differentiable function, and let $\rho$ be a faithful state with spectral decomposition $\rho = \int\lambda dE_\lambda$. Then the difference quotient with respect to $\rho$ is given by the double operator integral
 		\begin{equation}
 			J^f_\rho(x) := \int\int \frac{f(\lambda) - f(\mu)}{\lambda - \mu}dE_\lambda x dE_\mu
 		\end{equation} 
 		where the ratio is defined to be $f'(\lambda)$ if $\lambda = \mu$.
 	
 		If the spectral decomposition is discreete i.e. $\rho = \sum_\lambda \lambda e_\lambda$, then 
 		\begin{equation}
 			J^f_\rho(x) = \sum_{\lambda, \mu}\frac{f(\lambda) - f(\mu)}{\lambda - \mu}e_\lambda x e_\mu
 		\end{equation}
 	\end{definition}
 	\begin{proposition}\label{prop:cancel modular}
 		Let $\tau$ be the canonical trace on the finite von Neumann algebra $\mathcal{M}$. Fix $\beta$ as above, consider the density matrix
 		\begin{equation}d_\beta :=
 			\begin{pmatrix}
 				\frac{e^{-\beta/2}}{e^{\beta / 2}+e^{-\beta / 2}}& 0 \\ 0 & \frac{e^{\beta / 2}}{e^{\beta / 2}+e^{-\beta / 2}}
 			\end{pmatrix}
 		\end{equation}
 		and consider the generator
 		\begin{equation}
 			A := e_{12} \otimes a \in \mathbb{M}_2\overline{\otimes}\mathcal{M}
 		\end{equation}
 		Then we have
 		i)\begin{equation}
 				[d_\beta\otimes d_\phi, A] = 0
 			\end{equation}
 		ii) For state $\rho_\beta := (d_\beta\otimes d_\phi)^{1/2}(1\otimes x)(d_\beta\otimes d_\phi)^{1/2}$, we have
 			\begin{equation}
 				EP_A(\rho_\beta) \leq C(\beta) EP_a(d_\phi^{1/2}xd_\phi^{1/2}), EP_{A^*}(\rho_\beta) \leq C(\beta)EP_{a^*}(d_\phi^{1/2}xd_\phi^{1/2})
 			\end{equation}
 			where $EP_A(\rho) := \langle[A, \rho], J^{\log}_{\rho}[A,\rho]\rangle$ is the entropy production of state $\rho$ under generator $A$.
 	\end{proposition}
 	\begin{proof}
 		i) For the first claim, from direct calculation we have
 			\begin{align}
 				\begin{split}
 					(d_\beta^{it}\otimes d_\phi^{it})(e_{12}\otimes a)(d_\beta^{-it}\otimes d_\phi^{-it}) &= (d_\beta^{it}e_{12}d_\beta^{-it})\otimes(d_\phi^{it}a d_\phi^{-it})
 					\\
 					&=e^{-i\beta t}e_{12} \otimes e^{i\beta t}a = e_{12}\otimes a = A
 				\end{split}
 			\end{align}
 			Therefore $[d_\beta\otimes d_\phi, A] = 0$. 
 		\\
 		ii) For the second claim, we need the following technical result:
 		\begin{claim}
 			\begin{equation}
 				\frac{\log\lambda - \log\mu}{\lambda - \mu}\geq e^{-\beta/2}\frac{\log(\lambda e^{-\beta/2}) - \log(\mu e^{\beta/2})}{\lambda e^{-\beta/2} - \mu e^{\beta/2}}
 			\end{equation}
 		\end{claim}
 			We postpone the proof of this inequality and proceed with the main proof.
 			
 			Let $d_\phi^{1/2}xd_\phi^{1/2} = \sum_\lambda \lambda e_\lambda$ be the spectral decomposition of the state. We have 
 			\begin{equation*}
 				(d_\beta^{1/2}\otimes d_\phi^{1/2})(e_{12}\otimes[a,x])(d_\beta^{1/2}\otimes d_\phi^{1/2}) =\frac{1}{e^{\frac{\beta}{2}} + e^{-\frac{\beta}{2}}}e_{12}\otimes(d_\phi^{1/2}[a,x]d_\phi^{1/2})\pl.
 			\end{equation*}
 			And the spectral decomposition of $\rho_\beta$ is given by
 			\begin{equation*}
 				\rho_\beta = \sum_\lambda \frac{\lambda e^{-\beta / 2}}{e^{\beta / 2} + e^{-\beta / 2}}e_{11}\otimes e_\lambda + \sum_\mu \frac{\mu e^{\beta / 2}}{e^{\beta / 2} + e^{-\beta / 2}}e_{22}\otimes e_\mu \pl.
 			\end{equation*}
 			Combining these two calculations, we have
 			\begin{align*}
 				\begin{split}
 					J^{\log}_{\rho_\beta}[A, \rho_\beta] &= J^{\log}_{\rho_\beta}\big((d_\beta^{1/2}\otimes d_\phi^{1/2})(e_{12}\otimes [a, x])(d_\beta^{1/2}\otimes d_\phi^{1/2})\big)
 					\\
 					&=\frac{1}{e^{\beta / 2}+ e^{-\beta/2}}\sum_{\lambda, \mu}\frac{\log(\frac{\lambda e^{-\beta / 2}}{e^{\beta / 2}+e^{-\beta/2}}) - \log(\frac{\mu e^{\beta/2}}{e^{\beta / 2}+e^{-\beta/2}})}{\frac{\lambda e^{-\beta / 2}}{e^{\beta / 2}+e^{-\beta/2}} - \frac{\mu e^{\beta/2}}{e^{\beta / 2}+e^{-\beta/2}}}(e_{12}\otimes e_\lambda d_\phi^{1/2}[a,x]d^{1/2}_\phi e_\mu)
 					\\
 					&=\sum_{\lambda,\mu}\frac{\log(\lambda e^{-\beta /2}) - \log(\mu e^{\beta / 2})}{\lambda e^{-\beta / 2} - \mu e^{\beta /2}}(e_{12}\otimes e_\lambda d_\phi^{1/2}[a,x]d_\phi^{1/2}) \pl.
 				\end{split}
 			\end{align*}
 			Hence the entropy production with respect to the trace on $\mathbb{M}_2\otimes\mathcal{M}$ is given by
 			\begingroup
 				\allowdisplaybreaks
 				\begin{align*}
 					E&P_A(\rho_\beta) \lel \langle[A,\rho_\beta], J^{\log}_{\rho_\beta}[A,\rho_\beta] \rangle 
 					\\
 					&=\langle\frac{1}{e^{\beta / 2}+e^{-\beta/2}}e_{12}\otimes d_\phi^{1/2}[a,x]d^{1/2}_\phi, e^{\beta/2}\sum_{\lambda,\mu}\frac{\log(\lambda e^{-\beta / 2}) - \log(\mu e^{\beta/2})}{\lambda e^{-\beta /2}- \mu e^{\beta/2}}e_{12}\otimes e_\lambda d^{1/2}_\phi[a,x]d^{1/2}_\phi e_\mu\rangle
 					\\
 					&=\frac{e^{\beta/2}}{2(e^{\beta / 2}+e^{-\beta /2})}\sum_{\lambda, \mu}\frac{\log(\lambda e^{-\beta / 2}) -\log(\mu e^{\beta/2})}{\lambda e^{-\beta / 2} - \mu e^{\beta/2}}\tau(d_\phi^{1/2}[a,x]^*d_\phi^{1/2}e_\lambda d_\phi^{1/2}[a,x]d^{1/2}_\phi e_\mu)
 					\\
 					&\leq \frac{e^{\beta}}{2(e^{\beta / 2}+e^{-\beta/2})} \sum_{\lambda, \mu}\frac{\log\lambda - \log\mu}{\lambda - \mu}\tau(d_\phi^{1/2}[a,x]^*d_\phi^{1/2}e_\lambda d_\phi^{1/2}[a,x]d^{1/2}_\phi e_\mu)
 					\\
 					&=\frac{e^{\beta}}{2(e^{\beta /2}+e^{-\beta/2})}\langle d^{1/2}_\phi[a,x]d^{1/2}_\phi, J_\phi^{\log} d^{1/2}_\phi [a,x]d^{1/2}_\phi\rangle = \frac{e^{\beta}}{2(e^{\beta/2} + e^{-\beta/2})}EP_a(d_\phi^{1/2}xd_\phi^{1/2})
 			\end{align*}
 			\endgroup
 			Analogously, the same calculation works for $A^*$. Hence we obtain the desired upper bound where the constant $C(\beta)$ can be taken to be  $\frac{e^{\beta}}{2(1e^{\beta/2}+e^{-\beta/2})}$.
 	\end{proof}
 	\begin{lemma}
 		Let $\beta > 0$, then we have
 		\begin{eqnarray}
 			\frac{\log\lambda - \log\mu}{\lambda - \mu} \geq e^{-\beta/2}\frac{\log(\lambda e^{-\beta/2}) - \log(\mu e^{\beta/2})}{\lambda e^{-\beta/2} - \mu e^{\beta/2}}
 		\end{eqnarray} 
 	\end{lemma}
	\begin{proof}
		Consider the following integral representation:\begin{align}
			\begin{split}
				\frac{\lambda e^{-\beta/2} - \mu e^{\beta/2}}{\log(\lambda e^{-\beta/2}) - \log(\mu e^{\beta/2})}  &= \int_0^1 (\lambda e^{-\beta/2})^t(\mu e^{\beta/2})^{1-t} dt
				= \int_0^1 e^{\beta/2 - t\beta}\lambda^t\mu^{1-t}dt
				\\
				&\geq e^{-\beta/2}\int_0^1  \lambda^t\mu^{1-t}dt
				=e^{-\beta/2}\frac{\lambda - \mu}{\log\lambda - \log\mu}
			\end{split}
		\end{align}
		Therefore we have $\frac{\log\lambda - \log\mu}{\lambda - \mu} \geq e^{-\beta/2}\frac{\log(\lambda e^{-\beta/2}) - \log(\mu e^{\beta/2})}{\lambda e^{-\beta/2} - \mu e^{\beta/2}}$
	\end{proof}
	
	The previous proposition relates the entropy production of the augmented system (i.e. $\mathbb{M}_2\overline{\otimes}\mathcal{M}$ with generator $A$ and state $\rho_\beta$) with the entropy production of the original system (i.e. $\mathcal{M}$ with generator $a$ and state $d_\phi^{1/2}xd_\phi^{1/2}$). 
	
	For the next proposition, we first make the following obersation. Let $a = u|a|$ be the polar decomposition of $a$, then $a^* = |a|u^*$. Moreover, in $\mathbb{M}_2\overline{\otimes}\mathcal{M}$, we have
	\begin{equation}\label{equation:polardecomposition1}
		A = \begin{pmatrix}
			0 & a \\ 0&0
		\end{pmatrix} = \begin{pmatrix}
		u&0\\0&1
	\end{pmatrix}\begin{pmatrix}
	0 & |a| \\ 0&0
\end{pmatrix}\begin{pmatrix}
u^*&0 \\ 0&1
\end{pmatrix}\pl.
	\end{equation}Similarly, we have \begin{equation}\label{equation: polardecomposition2}
	A^* = \begin{pmatrix}
		0 & 0 \\ a^* & 0
	\end{pmatrix} = \begin{pmatrix}
	u&0\\0&1
\end{pmatrix}\begin{pmatrix}
0&0\\|a|&0
\end{pmatrix}\begin{pmatrix}
u^*&0 \\ 0&1
\end{pmatrix}\pl.
\end{equation}In the following, we shall denote $U:=\begin{pmatrix}
u&0 \\ 0&1
\end{pmatrix}$. A final observation is that:
\begin{equation*}
	H := A + A^* = U\begin{pmatrix}
		0 & |a| \\ |a| & 0
	\end{pmatrix}U^* = U(X\otimes |a|)U^*
\end{equation*}where $X\in\mathbb{M}_2$ is the Pauli matrix $\begin{pmatrix}
0&1\\1&0
\end{pmatrix}$. Let $\sigma(|a|)$ be the spectrum of $|a|$, then $\{|\lambda \pm \mu|: \lambda\neq \mu, \lambda, \mu\in \sigma(|a|)\}$ is the set of spectral difference of $X\otimes |a|$. 
	\begin{proposition}\label{proposition:uniformgap}
		Let $\delta = \min\{|\lambda \pm \mu|: \lambda\neq\mu, \lambda, \mu \in \sigma(|a|)\}$. Let $C^*(|a|)$ be the $C^*$-algebra generated by $|a|$ and let $E_a:\mathcal{M}\rightarrow C^*(|a|)'$ be the conditional expectation onto its commutant. Similarly define $C^*$-algebra $C^*(|a^*|)$ with the conditional expectation $E_{a^*}$.  Then for any state $\rho = d_\phi^{1/2}xd_\phi^{1/2}$ on $\mathcal{M}$, we have
		 \begin{equation}
			D(\rho| E^*_a\rho) \leq \frac{C(\beta)}{\delta^2}(EP_a(\rho) + EP_{a^*}(\rho))
		\end{equation}
		where the constant $C(\beta)$ only depends on $\beta$. The same holds for $a^*$:
		\begin{equation}
			D(\rho|E^*_{a^*}\rho) \leq \frac{C(\beta)}{\delta^2}(EP_a(\rho) + EP_{a^*}(\rho))
		\end{equation}. 
	\end{proposition}
	\begin{proof}
		Let $N := C^*(X\otimes |a|)$ be the $C^*$-subalgebra of $\mathbb{M}_2\overline{\otimes}\mathcal{M}$ generated by $X\otimes |a|$, and let $E_{N'}$ be the conditional expectation onto the commutant $N'$. Since $X\otimes |a|$ is a self-adjoint operator with discrete spectrum
		\begin{equation*}
			\sigma(X\otimes |a|) = \{|\lambda\pm\mu|:\lambda, \mu \in \sigma(|a|)\}
		\end{equation*}
		Theorem \ref{theorem:s.a.CLSI} shows that there exists an absolute constant $C >0$ such that for any state $\widetilde{\rho}\in S_1(\mathbb{M}_2\overline{\otimes}\mathcal{M})$
		\begin{equation}
			D(\widetilde{\rho}| E^*_{N'}\widetilde{\rho}) \leq \frac{C}{\delta^2} EP_{X\otimes |a|}(\widetilde{\rho})
		\end{equation} 
		where $EP_{X\otimes|a|}(\widetilde{\rho}) = \langle[X\otimes|a|, \widetilde{\rho}], J^{\log}_{\widetilde{\rho}}[X\otimes|a|, \widetilde{\rho}]\rangle$ is the entropy production with generator $X\otimes|a|$. By the triangle inequality
		\begin{align*}
			\begin{split}
				EP_{X\otimes|a|}(\widetilde{\rho}) &= \langle[e_{12}\otimes|a| + e_{21}\otimes |a|, \widetilde{\rho}], J^{\log}_{\widetilde{\rho}}[e_{12}\otimes|a| + e_{21}\otimes|a|, \widetilde{\rho}]\rangle
				\\&\leq 2(\langle[e_{12}\otimes|a|, \widetilde{\rho}], J^{\log}_{\widetilde{\rho}}[e_{12}\otimes|a|, \widetilde{\rho}]\rangle + \langle[e_{21}\otimes|a|, \widetilde{\rho}], J^{\log}_{\widetilde{\rho}}[e_{21}\otimes|a|, \widetilde{\rho}]\rangle)
			\end{split}
		\end{align*} 
	where the $\langle,\rangle$ here is the tracial inner product. 
	
	Using the notation introduced immediately before the propostion, we apply the inequality for the state $Ad_{U^*}(\rho_\beta)$ where \begin{equation*}
		\rho_\beta = (d_\beta\otimes d_\phi)^{1/2}(1\otimes x)(d_\beta\otimes d_\phi)^{1/2}
	\end{equation*} and
	\begin{equation*}
		d_\beta := \begin{pmatrix}
			e^{-\beta/2}/ (e^{\beta/2} + e^{-\beta/2}) & 0 \\ 0 & e^{\beta/2} / (e^{\beta/2} + e^{-\beta/2})
		\end{pmatrix}
	\end{equation*}
	 Then we have
	\begin{align}
		\begin{split}
			EP_{X\otimes |a|}(Ad_{U^*}(\rho_\beta))&\leq 2(\langle[e_{12}\otimes|a|, Ad_{U^*}(\rho_\beta)], J^{\log}_{Ad_{U^*}(\rho_\beta)}[e_{12}\otimes|a|, Ad_{U^*}(\rho_\beta)] \rangle 
			\\
			&+ \langle[e_{21}\otimes|a|, Ad_{U^*}(\rho_\beta)], J^{\log}_{Ad_{U^*}(\rho_\beta)}[e_{21}\otimes|a|, Ad_{U^*}(\rho_\beta)]\rangle)
			\\
			&=2(\langle [A, \rho_\beta], J^{\log}_{\rho_\beta}[A, \rho_\beta] \rangle + \langle[A^*, \rho_\beta], J^{\log}_{\rho_\beta}[A^*, \rho_\beta]\rangle)
			\\
			&=2(EP_A(\rho_\beta) + EP_{A^*}(\rho_\beta))
			\\
			&\leq C(\beta)(EP_a(\rho) + EP_{a^*}(\rho))
		\end{split}
	\end{align}
	where the first equality follows from the (modified-)polar decomposition equations \ref{equation:polardecomposition1}\ref{equation: polardecomposition2} and the last inequality follows from proposition \ref{prop:cancel modular}. The constant $C(\beta)$ depends only on $\beta$. This inequality provides an upper bound on the entropy production term.
	
	On the other hand, since adjoint action $Ad_{U^*}$ preserves the diagonal subalgebra, we have
	\begin{equation*}
		Ad_{U^*}(N') \cap \ell^2_\infty(\mathcal{M}) = \{\begin{pmatrix}
			uxu^* & 0 \\ 0 & y
		\end{pmatrix}: x,y\in C^*(|a|)'\}
	\end{equation*}
	In addition we deduce from the diagonal structure:
	\begin{align}
		\begin{split}
			D(Ad_{U^*}(\rho_\beta)| E^*_{Ad_{U^*}(N')}Ad_{U^*}(\rho_\beta)) &= \frac{e^{\beta}}{2(1+e^{\beta})}D(u^*\rho u | E^*_{a^*}u^*\rho u) 
			+ \frac{1}{2(1+e^{\beta})}D(\rho | E^*_a\rho)
			\\
			&\geq \frac{1}{2(1+e^\beta)}D(\rho | E^*_a\rho)
		\end{split}
	\end{align}
	Combining these two inequalities, we find a constant $C(\beta)$ such that
	\begin{equation}
		D(\rho|E^*_a\rho) \leq \frac{C(\beta)}{\delta^2}(EP_a(\rho) + EP_{a^*}(\rho))
	\end{equation}
	Apply the same argument to $a^*$, we have
	\begin{equation}
		D(\rho|E^*_{a^*}\rho) \leq \frac{C(\beta)}{\delta^2}(EP_a(\rho) + EP_{a^*}(\rho))
	\end{equation}
	\end{proof}
	\begin{corollary}\label{corollary:intermediate}
		Using the same notation as above, suppose $\Omega := C^*(|a|)' \cap C^*(|a^*|)'$ and the conditional expectations form a commuting square, i.e.
		\begin{equation}
			E_{\Omega} = E_a E_{a^*} = E_{a^*} E_a
		\end{equation}
		Then there exists a constant $C(\beta)$ depending only on $\beta$ such that
		\begin{equation}
			D(\rho | E^*_\Omega\rho)\leq \frac{C(\beta)}{\delta^2}(EP_a(\rho) + EP_{a^*}(\rho))
		\end{equation}
	\end{corollary}
	\begin{proof}
		By a well-known approximate tensorization result \cite{Petz95}, if a pair of conditional expectations satisfies the commuting square condition, we have\begin{equation*}
			D(\rho | E^*_\Omega\rho) \leq D(\rho | E^*_a \rho) + D(\rho | E^*_{a^*}\rho)
		\end{equation*}
		Hence by proposition \ref{proposition:uniformgap}, we get
		\begin{equation}
			D(\rho | E^*_\Omega \rho) \leq \frac{C(\beta)}{\delta^2}(EP_a(\rho) + EP_{a^*}(\rho))
		\end{equation}
		for some constant $C(\beta)$.
	\end{proof}

	\section{\large Application to su(2)-Lindbladians}\label{section:su(2)}
	In this section, we apply results from last section to Lindbladians generated by tensor product representation of $\mathfrak{su}(2)$. Such Lindbladians can be used to describe dynamics of open physical systems with $\mathfrak{su}(2)$-symmetry (for example, a chain of qubits). The main goal of this section is to estimate the CLSI constant of these $\mathfrak{su}(2)$-Lindbladians.
	\subsection{Structure and Representation Theory of $\mathfrak{su}(2)$}
	
	Recall that the Lie algebra $\mathfrak{su}(2)$ is a 3-dimensional real Lie algebra. \cite{Fulton} Pauli matrices $\sigma_x = \begin{pmatrix}
		0&1 \\ 1 &0
	\end{pmatrix}, \sigma_y = \begin{pmatrix}
		0& -i \\ i& 0
	\end{pmatrix}, \sigma_z = \begin{pmatrix}
		1& 0 \\ 0& -1
	\end{pmatrix}$ form a basis of $\mathfrak{su}(2)$. For our purpose, we need to consider the complexification of $\mathfrak{su}(2)$. Recall the complexification of $\mathfrak{su}(2)$ is given by:
	\begin{align*}
		\mathfrak{su}(2) \otimes \mathbb{C} \cong \mathfrak{sl}(2,\mathbb {C}) := \{A=\begin{pmatrix}
			\alpha&\beta \\ \gamma&\delta
		\end{pmatrix}: tr(A) = 0, A\in M_2(\mathbb{C})\}
	\end{align*}
	As a complex Lie algebra, $dim_\mathbb{C} (\mathfrak{sl}(2,\mathbb{C})) = 3$. There exists another canonical basis for $
	\mathfrak{sl}(2,\mathbb{C})$ (Cartan-Weyl basis). The transformation from Pauli matrices to Cartan-Weyl basis is given by:
	\begin{align*}
		\begin{split}
			a = \frac{1}{2}(\sigma_x + i\sigma_y) = \begin{pmatrix}
				0&1 \\ 0&0
			\end{pmatrix},\pl
			a^* = \frac{1}{2}(\sigma_x -i\sigma_y) = \begin{pmatrix}
				0&0 \\ 1&0
			\end{pmatrix}
			,\pl
			h = \sigma_z = \begin{pmatrix}
				1&0\\0&-1
			\end{pmatrix}
		\end{split}\pl.
	\end{align*}
	In physics literature, $a$ is more commonly known as creation operator, $a^*$ the annhilation operators and $h$ the Hamiltonian. Cartan-Weyl basis satisfies the following relation:
	\begin{align}
		\begin{split}
			[h,a] = 2a,\pl
			[h,a^*] = -2a^*,\pl
			[a,a^*] = h 
		\end{split}
	\end{align}
	
	Recall the fundamental representation of $\mathfrak{su}(2)$ acts on the two-dimensional vector space $V:=\mathbb{C}^2$  by matrix multiplication. The fundamental representation is the irreducible representation of $\mathfrak{su}(2)$ of the smallest dimension. All irreducible representations of $\mathfrak{su}(2)$ are completely known. For each dimension $n \geq 1$, there exists a unique irreducible representation:
	\begin{equation}
		\pi_{(n)} : \mathfrak{su}(2) \rightarrow \mathbb{M}_{n+1}(\mathbb{C})
	\end{equation}
	Weight theory shows that there exists a unique highest weight vector $\ket{n,0}\in \mathbb{C}^{n+1}$ such that \begin{align}
		\begin{split}
			&a\ket{n,0} = 0 \\ &h\ket{n,0} = n\ket{n,0}
		\end{split}
	\end{align}
	From the highest weight vector, a complete basis for $\mathbb{C}^{n+1}$ is given by $\{\ket{n,j}\}_{0\leq j\leq n}$ such that
	\begin{align}
		\begin{split}
			&h\ket{n,j}= (n - 2j) \ket{n,j}
			\\
			&a^*\ket{n,j} = \sqrt{(j+1)(n-j)}\ket{n,j+1}
			\\
			&a\ket{n,j} = \sqrt{j(n-j+1)}\ket{n,j-1}
		\end{split}
	\end{align}
	Since the coefficients will be used repeatedly later, we denote $\alpha_{n,j} := \sqrt{j(n-j+1)}$.
	
	The $N$-fold tensor product $V^{\otimes N} := (\mathbb{C}^2)^{\otimes N}$ admits a tensor product representation:
	\begin{align}
		\begin{split}
			&\pi_N:=\sum_{1\leq j\leq N}\pi^j = \sum_{1\leq j \leq N}1\otimes...\otimes \pi_{(1)}\otimes...\otimes 1:\mathfrak{su}(2) \rightarrow \mathbb{M}_{2^N}(\mathbb{C})
			\\
			&\pi_N(g)(v_1\otimes...\otimes v_N) = \sum_{1\leq j\leq N}v_1\otimes...\otimes(\pi_{(1)}v_j)\otimes...\otimes v_N
		\end{split}
	\end{align}
	
	The tensor product representation $\pi_N = \pi_{(1)}^{\otimes N}$ is not irreducible. Decomposition of $\pi_N$ is given by the Schur-Weyl duality. The basic observation of the Schur-Weyl duality is that, in addition to left action of $\mathfrak{su}(2)$, the symmetric group $S_N$ acts on $V^{\otimes N}$ from the right: $\forall \sigma \in S_N$ and $v_1\otimes...\otimes v_N \in V^{\otimes N}$
	\begin{equation}
		(v_1\otimes...\otimes v_N)\cdot \sigma = v_{\sigma(1)}\otimes...\otimes v_{\sigma(N)}
	\end{equation}
	Since the left action by $\mathfrak{su}(2)$ commute with the right action by $S_N$, by double centralizer theorem $V^{\otimes N}$ has the following decomposition as a $\mathfrak{su}(2)-S_N$ bimodule:
	\begin{equation}
		V^{\otimes N} \cong \bigotimes_{\lambda \in \mathcal{P}(N)} V_\lambda \otimes W_\lambda
	\end{equation}
	where $\lambda := (\lambda_1\geq \lambda_2\geq ...\geq \lambda _l)$ is an ordered partition of $N$, $W_\lambda$ is the irreducible representation of $S_N$ labeled by the Young's diagram with $l$-rows where the $i$-th row has $\lambda_i$ boxes. In addition, $V_\lambda \cong Hom_{\mathbb{C}[S_N]}(W_\lambda, V^{\otimes N})$ is the $\mathbb{C}[S_N]$-module map from $W_\lambda$ to $V^{\otimes N}$.
	
	Under the left $\mathfrak{su}(2)$ action, $V_\lambda$ is an irreducible representation. It is a fact that $V_\lambda \neq 0$ if the Young diagram $\lambda$ has no more than $dim(V)$-rows. Hence for fundamental representation of $\mathfrak{su}(2)$, $V_\lambda \neq 0$ if $\lambda$ has at most two rows. 
	
	As $\mathfrak{su}(2)$-modules, we obtain an irreducible decomposition:
	\begin{equation}
		V^{\otimes N} \cong \bigoplus_{\lambda\in \mathcal{P}_2(N)} V_\lambda ^{\oplus dim(W_\lambda)}
	\end{equation}
	where $\mathcal{P}_2(N)$ is the set of paritions of $N$ with at most 2 parts. Note that the multiplicity of each irreducible component is given by the dimension of $W_\lambda$. It is well-known that the dimension of $W_\lambda$ is given by the hook length formula of Young diagrams associated with $\lambda$. 
	
	The Young diagram with only one row corresponds to the irreducible component with the largest dimension. We denote this subrepresentation by $V_{(N)}$ since it corresponds to the trivial partition of $N$. The following fact is well known \cite{Fulton}
	\begin{lemma}
		\begin{equation*}
			\dim(V_{(N)}) = N+1
		\end{equation*}
		In addition, let $\{e_0, e_1\}$ be a basis of the fundamental representation where $e_0$ is the highest weight vector and $e_1$ is the lowest weight vector, then the highest weight vector of $V_{(N)}$ is given by 
		\begin{equation*}
			\ket{N, 0} = e_0^{\otimes N}
		\end{equation*}
		Moreover, the weight space decomposition gives $V_{(N)}$ a basis of totally symmetric vectors $(v_k)_{0\leq k \leq N}$, where
		\begin{equation*}
			v_k := \sum_{A\subset [|N|], |A| = k}\frac{1}{\binom{N}{k}}v_A = \sum_{A\subset [|N|], |A| = k}\frac{1}{\binom{N}{k}}(\bigotimes_{j\in A} e_1)\otimes(\bigotimes_{j\notin A} e_0)
		\end{equation*}
		where $A$ is a cardinality $k$ subset of $[|N|] :=\{1,...,N\}$, and $v_A$ is a simple tensor product of $e_0, e_1$ where the $j$-th tensor factor is $e_1$ if and only if $j\in A$
	\end{lemma}
	
	To apply results from last section, we need the following lemma.
	\begin{lemma}\label{lemma:commutant}
		Let $\mathcal{H}:=\bigoplus_{n = 1}^N \mathbb{C}^{n+1}$ and consider the direct sum representation $\pi:=\bigoplus_{n = 1}^N\pi_n$. Then we have:
		
		i) The commutant of $|\pi(a)|$ in $\mathbb{B}(\mathcal{H})$ is given by
		\begin{equation*}
			C^*(|\pi(a)|)' = span\{\ket{n,j}\bra{m,k}: \alpha_{n,j} = \alpha_{m,k}\}
		\end{equation*}
	
		ii) Let $E_{\pi(a)}$ be the conditional expectation onto $C^*(|\pi(a)|)'$ and $E_{\pi(a^*)}$ be the conditional expectation onto $C^*(|\pi(a^*)|)'$. Then 
		\begin{equation}
			E_{\pi(a)}E_{\pi(a^*)} = E_{\pi(a^*)}E_{\pi(a)}
		\end{equation}
	
		iii) The intersection of $C^*(|\pi(a)|)' \cap C^*(|\pi(a^*)|)'$ is contained in $\bigoplus_{n = 1}^N \ell_\infty^{n+1}$.
	\end{lemma}
	\begin{proof}
		i) Since $
			\pi(a) = \sum_{1\leq n\leq N}\sum_{0 < j \leq n}\alpha_{n,j}\ket{n,j-1}\bra{n,j}$, then \begin{equation*}
			|\pi(a)| = \sum_{1\leq n \leq N}\sum_{0 \leq j \leq n}\alpha_{n,j}\ket{n,j}\bra{n,j}
		\end{equation*} (Note $\alpha_{n,j} = 0$ when $j = 0$.)
		Therefore a matrix \begin{equation*}
			x:=\sum_{(n,j),(m,k)}x_{(n,j),(m,k)}\ket{n,j}\bra{m,k}\in\mathbb{B}(\mathcal{H})
		\end{equation*} commutes with $|\pi(a)|$ if and only if $x_{(n,j),(m,k)} = 0$ whenever $\alpha_{n,j}\neq \alpha_{m,k}$. 
		
		Since $\pi(a^*) = \sum_{1\leq n\leq N}\sum_{0\leq j < n}\alpha_{n,j+1}\ket{n,j + 1}\bra{n,j}
		$, then
		\begin{equation*}
			|\pi(a^*)| = \sum_{1\leq n \leq N}\sum_{0\leq j \leq n}\alpha_{n,j+1}\ket{n,j}\bra{n,j}
		\end{equation*}
		(Note $\alpha_{n,j+1} = 0$ when $j = n$.) Therefore a matrix $x$ commutes with $|\pi(a^*)|$ if and only if $x_{(n,j), (m,,k)}  = 0$ whenever $\alpha_{n,j+1} \neq \alpha_{m,k+1}$.
		
		ii) By computation above, both conditional expectations $E_{\pi(a)}$ and $E_{\pi(a^*)}$ are given by Schur multipliers in the basis of $\ket{n,j}$. Therefore they commute. 
		
		iii) Moreover, we can give a concrete description of $\Omega:=C^*(|\pi(a)|)'\cap C^*(|\pi(a^*)|)'$. Consider $E_\Omega(\ket{n,j}\bra{m,k}):=E_{\pi(a)}E_{\pi(a^*)}(\ket{n,j}\bra{m,k})$, then $\ket{n,j}\bra{m,k}$ is in $\Omega$ if and only if $E_\Omega(\ket{n,j}\bra{m,k}) = \ket{n,j}\bra{m,k}$. This is equivalent to the following two conditions:
		\begin{align*}
			\begin{split}
				&\alpha_{n,j} = 
				\alpha_{m,k}\\
				&\alpha_{n,j+1} = \alpha_{m,k+1}
			\end{split}
		\end{align*}
		Then by an explicit computation (see the lemma below), we have $n =m, j = k$. Hence $\ket{n,j}\bra{m,k}\in\Omega$ if and only if $\ket{n,j}\bra{m,k}$ belongs to the diagonal subalgebra $ \bigoplus_{n =1}^N \ell_\infty^{n+1}$.
	\end{proof}
	\begin{lemma}
		$\alpha_{n,j} = \alpha_{m,k}$ and $\alpha_{n,j+1} = \alpha_{m,k+1}$ hold simultaneously if and only if $n = m$ and $j = k$
	\end{lemma}
	\begin{proof}
		One direction is obvious. For the other direction, denote $\delta := n - 2j$. Since $n - 2j = \alpha_{n, j+1} - \alpha_{n,j} = \alpha_{m,k+1} - \alpha_{m,k} = m-2k$, then $m-2k = \delta$. Then we have
		\begin{align*}
			\begin{split}
				\alpha_{n,j}^2 - \alpha_{m,k}^2 &= j(n-j+1) - k(m-k+1) = j(\delta + j +1) - k(\delta + k +1)
				\\
				&=(j - k)(\delta +1 + j+k) = (j-k)(j+1+m-k)
			\end{split}
		\end{align*}
		By assumption, this equals to $0$. If $j\neq k$, then $j+1+m-k = 0$. Since $j\geq 0$ and $k \leq m$, $j+1+m-k \geq 1$. Hence $j = k$. Then $n = \delta + 2j = \delta + 2k = m$. 
	\end{proof}

	We are mainly interested in Lindbladian associated with the representation $\pi_N$. To apply the result of the last section, we need to check whether $|\pi_N(a)|$ has a uniform spectral gap. This is a part of the following lemma.
	
	\begin{lemma}\label{lemma:unigap}\label{corollary:evaporation}
		Using the notation above and fix the weight basis $\big(\ket{n,j}\big)_{0\leq j \leq n}$ for each irreducible component of $\pi_N$. And fix a constant $\gamma \geq 1$. Then for a pair of distinct coefficients $\alpha_{n,j}\neq \alpha_{m,k}$
		\begin{equation}
			|\alpha_{n,j}^\gamma - \alpha_{m,k}^\gamma| \gtrapprox \min\{N^{\gamma - 2}, 1\} \pl.
		\end{equation}
		In particular, the lower bound can be achieved. For $\gamma < 2$, the lower bound is achieved when $n = m = N, j = [\frac{N+1}{2}], k = j - 1$. For $\gamma \geq 2$, the lower bound is achieved when $n, m \approx O(1)$.
	\end{lemma}
	\begin{proof}
		Without loss of generality, assume $\alpha_{n,j} > \alpha_{m,k}$. By mean value theorem, there exists $\theta\in [\alpha_{m,k}^2, \alpha_{n,j}^2]$ such that:
		\begin{equation*}
			\alpha_{n,j}^\gamma - \alpha_{m,k}^\gamma = \frac{\gamma}{2}\theta^{\gamma/2 - 1}(\alpha_{n,j}^2 - \alpha_{m,k}^2) \pl.
		\end{equation*}
		Note since $\alpha_{n,j}^2, \alpha_{m,k}^2\in\mathbb{Z}$, $\alpha_{n,j}^2 - \alpha_{m,k}^2\geq 1$. Since $\theta \geq \alpha_{m,k}^2$ and $1\lessapprox\alpha_{m,k}\lessapprox N$, we have that
		\begin{equation*}
			\alpha_{n,j}^\gamma - \alpha_{m,k}^\gamma \gtrapprox \min\{N^{\gamma-2},1 \} \pl.
		\end{equation*}
		To see that the lower bound can be achieved, we assume $n = m$, we can define afunction $f_\gamma(\epsilon):=\bigg((\frac{n+1}{2} - \epsilon)(\frac{n+1}{2} + \epsilon)\bigg)^{\gamma / 2}$ where $-\frac{n+1}{2} < \epsilon \leq \frac{n - 1}{2}$ and $\epsilon + \frac{1}{2} \in \mathbb{Z}$. By symmetry, we restrict to $\epsilon \geq 0$. Then the discrete difference $f_\gamma(\epsilon) - f_\gamma(\epsilon + 1)$ achieves minimum when $\epsilon$ is approximately 0. More precisely, if $n+1$ is even, the minimum is achieved when $\epsilon = 0$. In this case, we have:
		\begin{align*}
			\begin{split}
				f_\gamma(\epsilon) - f_\gamma(\epsilon + 1) &= (\frac{n+1}{2})^{\gamma} - \big((\frac{n+1}{2})^2 - 1\big)^{\gamma/2} = \big(\frac{n+1}{2}\big)^{\gamma}\big(1 - (1 -(\frac{2}{n+1})^2)^{\gamma/2}\big)\\
				&\approxeq \frac{\gamma}{2} (n+1)^{\gamma -2}\pl.
			\end{split}
		\end{align*}
		If $n+1$ is odd, the minimum is achieved by $\epsilon = \frac{1}{2}$. In this case, we have:
		\begin{align*}
			\begin{split}
				f_\gamma(\epsilon) - f_\gamma(\epsilon + 1) &= (\frac{n(n+2)}{4})^{\gamma / 2} - (\frac{(n-2)(n+4)}{4})^{\gamma / 2}\\& = (\frac{n(n+2)}{4})^{\gamma / 2}\big(1 - (1 - \frac{8}{n(n+2)})^{\gamma / 2}\big) \approxeq\frac{\gamma}{2}(n(n+2))^{\gamma / 2 - 1} \pl.
			\end{split}
		\end{align*}
		Since $1\leq n \leq N$, therefore if $\gamma < 2$, then $f_\gamma(\epsilon) - f_\gamma(\epsilon + 1)\gtrapprox N^{\gamma-2}$ for $n = m = N$. If $\gamma \geq 2$, then $f_\gamma(\epsilon) - f_\gamma(\epsilon + 1) \gtrapprox 1$ for $n = m = 1$. 
	\end{proof}

	\begin{remark}
		Since the dimension of $\pi_N$ is of the order $2^N$, the uniform spectral gap of $|\pi_N(a)|$ is of the order $\Theta(\frac{1}{\sqrt{\log\dim(\pi_N)}})$.
	\end{remark}
	
	\subsection{CLSI Constant of $\mathfrak{su}(2)$-Lindbladian}
	Recall that the Lindbladian we are interested in is given by:
	\begin{equation}
		\mathcal{L}_N^\beta(x) = \beta^{1/2}L_{\pi_N(a)}(x) + \beta^{-1/2}L_{\pi_N(a^*)}(x)\pl.
	\end{equation}The next lemma shows that we may restrict this Lindbladian to any subrepresentation.
	\begin{lemma}
		Each irreducible component $V_\lambda$ of $\mathcal{H}:=V^{\otimes N}$ is invariant under the Lindbladian $\mathcal{L}_N^\beta$. Hence given any subset of partitions $A\ssubset \mathcal{P}_2(N)$, if an operator $x\in \mathbb{B}(\bigoplus_{\lambda\in A}V_\lambda)$, then $\mathcal{L}_N^\beta(x)\in \mathbb{B}(\bigoplus_{\lambda\in A}V_\lambda)$
	\end{lemma}
	\begin{proof}
		The proof is obvious. For any $x\in \mathbb{B}(V_\lambda)$, we can expand it in terms a basis:
		\begin{equation}
			x = \sum_{i,j = 1}^{dim(V_\lambda)} x_{ij}\ket{i}\bra{j}
		\end{equation}
		Since $\pi_N(a)\ket{i}\in V_\lambda$, $\pi_N(a^*)\ket{j}\in V_\lambda$, it is clear that each term in the Lindbladian maps $x$ to an operator on $V_\lambda$. The statement about direct sum of irreducible components is obvious. 
	\end{proof}

	Recall Lie algebra relation $[\pi_N(a^*), \pi_N(a)] = \pi_N(h)$. We consider the density 
	\begin{equation}
		d_N := N_\beta \exp(-\frac{\beta}{2}\pi_N(h))
	\end{equation}
	where $N_\beta$ is a normalization constant depending on dimension only. By simple computation, we have
	\begin{lemma}
		Following two equations hold:\begin{align}
			\begin{split}
				&d^{it}_N \pi_N(a) d^{-it}_N = e^{i\beta t}\pi_N(a)
				\\
				&d^{it}_N \pi_N(a^*)d^{-it}_N = e^{-i\beta t}\pi_N(a^*)
			\end{split}	
		\end{align}
	\end{lemma}
	\begin{proof}
		Since $[\pi_N(h), \pi_N(a)] = -2 \pi_N(a)$, we have $[\pi_N^n(h), \pi_N(a)] = (-2)^n\pi_N(a)$. Hence
		\begin{align}
			\begin{split}
				d^{it}_N \pi_N(a) d^{-it}_N& = \exp(-\frac{i\beta t}{2}\pi_N(h)) \pi_N(a) \exp(\frac{i\beta t}{2}\pi_N(h))
				\\
				&= \sum_{n\geq 0 }\frac{(-i\beta t)^n}{2^n n!}[\pi_N^n(h), \pi_N(a)]
				\\
				&=\sum_{n\geq 0}\frac{(-i \beta t)^n}{2^n n!}(-2)^n \pi_N(a) = e^{i\beta t}\pi_N(a)
			\end{split}
		\end{align}
		Similar calculation holds for the other equation.
	\end{proof}

	Before applying results from section \ref{section:abstract}, we make the following simplification.
	\begin{lemma}
		Let $V^{\otimes N} = \bigoplus_{\lambda \in \mathcal{P}_2(N)}V_\lambda^{\oplus \dim(W_\lambda)}$ be the irreducible decomposition of $\pi_N$-representation. Let $n_0 = \max_{\lambda\in\mathcal{P}_2(N)}\{\dim(W_\lambda)\}$, then the inclusion map:
		\begin{equation*}
			\iota:\mathbb{B}(V^{\otimes N})\rightarrow \mathbb{B}(\bigoplus_{\lambda}V_\lambda) \bigotimes \mathbb{M}_{n_0}
		\end{equation*}
		satisfies the following intertwining relation:
		\begin{equation}
			\iota\circ\mathcal{L}_N^\beta = (\mathcal{L}^\beta_{\pi}\otimes id) \circ\iota 
		\end{equation}
		where $\mathcal{L}^\beta_\pi$ is the Lindbladian associated with the direct sum representation $\pi = \bigoplus_{\lambda}\pi_\lambda$. 
	\end{lemma}
	\begin{proof}
		Since $\mathcal{L}^\beta_N$ preserves irreducible decomposition, the intertwining relation is obvious.
	\end{proof}

	Therefore by tensorization property of the $\CLSI$ constant,  to estimate $\CLSI(\mathcal{L}^\beta_N)$ it suffices to consider $\mathcal{L}^\beta_\pi$. Then by results of the last section and lemma above, we have the following estimate.
	
	Recall the intermediate subalgebra is defined by $\Omega = C^*(|\pi_N(a)|)' \cap C^*(|\pi_N(a^*)|)'$.
	\begin{proposition}
		Let $\rho_N := d_N^{1/2}x d_N^{1/2}$ be a state in $S_1(V^{\otimes N})$ and $x\in\mathbb{B}(V^{\otimes N})$. Then there exists a constant $C(\beta)$ depending only on $\beta$ such that
		\begin{equation}
			D(\rho_N | E^*_\Omega\rho_N) \leq C(\beta)N^2(EP_{\pi_N(a)}(\rho_N) + EP_{\pi_N(a^*)}(\rho_N))
		\end{equation}
		where $E_\Omega$ is the conditional expectation onto $\Omega$.
	\end{proposition}
	\begin{proof}
		By commutativity of the conditional expectations: $E_\Omega = E_{\pi_N(a)}E_{\pi_N(a^*)} = E_{\pi_N(a^*)}E_{\pi_N(a)}$, the relative entropy factorizes 
		\begin{equation*}
			D(\rho_N | E^*_\Omega\rho_N) = D(\rho_N | E^*_{\pi_N(a)}\rho_N) + D(\rho_N | E^*_{\pi_N(a^*)}\rho_N)
		\end{equation*}
		Since uniform spectral gap of $|\pi_N(a)|$ is of the order $\Theta(\frac{1}{N})$, by proposition \ref{proposition:uniformgap} there exists a constant $C(\beta)$ such that
			\begin{align*}
				\begin{split}
					&D(\rho_N | E^*_{\pi_N(a)}\rho_N) \leq C(\beta)N^2(EP_{\pi_N(a)}(\rho_N) + EP_{\pi_N(a^*)}(\rho_N))
					\\
					&D(\rho_N | E^*_{\pi_N(a^*)}\rho_N) \leq C(\beta)N^2(EP_{\pi_N(a)}(\rho_N) + EP_{\pi_N(a^*)}(\rho_N))
				\end{split}
			\end{align*}
		Combing these two inequalities, we have the desired result.
	\end{proof}
	\begin{lemma}\label{lemma:Dirichlet}
		The fixed point algebra $\Omega_{fix}$ of $\mathcal{L}^\beta_N$ in $V^{\otimes N}$ is $\bigoplus_\lambda \mathbb{B}(W_\lambda)$. In particular, $\Omega_{fix} \ssubset \Omega$.
	\end{lemma}
	\begin{proof}
		Given the faithful density $d_N$, we obtain an embedding:
		\begin{equation*}
			\iota_N:\mathbb{B}(V^{\otimes N})\rightarrow S_2(V^{\otimes N}): x\mapsto d^{1/4}_N x d^{1/4}_N
		\end{equation*}
		Recall that on the Schatten 2-space, the Lindbladian $L^\beta_N$ uniquely defines a Dirichlet form. Restricted to the dense subspace $\iota_N(\mathbb{B}(V^{\otimes N}))$, the Dirichlet form is given by the formula:
		\begin{align*}
			\begin{split}
				\langle\xi,\xi\rangle_L &= \langle\iota_N(x), \iota_N(x)\rangle_L
				\\
				&=||d_N^{1/4}(\pi_N(a) x -  x\pi_N(a))d_N^{1/4}||^2_2 + ||d_N^{1/4}(\pi_N(a^*) x -  x\pi_N(a^*) )d_N^{1/4}||^2_2 
			\end{split}
		\end{align*}
		where $||\cdot||^2_2$ is the Hilbert-Schmidt inner product. Therefore $x\in\Omega_{fix}$ if and only if $[\pi_N(a), x] = [\pi_N(a^*), x] = 0$. In particular, $x$ must commute with $|\pi_N(a)|$ and $|\pi_N(a^*)|$, and thus $\Omega_{fix}\in\Omega = \bigoplus_\lambda \ell_\infty^{n_\lambda}\otimes \mathbb{B}(W_\lambda)$, where $n_\lambda := \dim(W_\lambda)$ and the equation follows from the irreducible decomposition and lemma \ref{lemma:commutant}.
		
		For an operator in one irreducible component $x\in\ell_\infty^{n_\lambda}\otimes \mathbb{B}(W_\lambda)$, by irreducibility $x$ must have the form $\mathbb{C}\otimes \mathbb{B}(W_\lambda)$. Therefore $\Omega_{fix} = \bigoplus_\lambda\mathbb{B}(W_\lambda)$.  
	\end{proof}
	\begin{corollary}
		\begin{equation}
		D(\rho_N | E^*_{fix}\rho_N) = D(\rho_N | E^*_\Omega\rho_N) + D(E^*_\Omega\rho_N | E^*_{fix}\rho_N)
	\end{equation}
	\end{corollary}
	\begin{proof}
		This is an immediate consequence of the previous lemma and the chain rule of relative entropy.
	\end{proof}
	\begin{remark}
		Lemma \ref{lemma:Dirichlet} gives a complete description of the fixed point algebra. We stress that the quantum Markov semigroup generated by $\mathcal{L}^\beta_N$ is non-primitive. 
		
		To clarify the relation between the various algebras used in our proof, we present the following commutative diagram:
		\begin{center}
			\begin{tikzcd}
				& C^*(|\pi_N(a)|)' \ar[dr, "E_{\pi_N(a^*)}"]
				&
				&[1.5em] \\
				\mathbb{B}(V^{\otimes N}) \ar[ur, "E_{\pi_N(a)}"] \ar[dr, "E_{\pi_N(a^*)}"']
				&
				& \Omega \ar[r]
				& \Omega_{fix} \\
				& C^*(|\pi_N(a^*)|)' \ar[ur, "E_{\pi_N(a)}"']
				&
				&
			\end{tikzcd}
		\end{center}
		The commuting square of the conditional expectations is a special feature of $\mathfrak{su}(2)$ representation (more generally any representation of simple Lie algebra). It may not be true for general Lindbladian generators.
	\end{remark}
	So far, we have discussed the left half of the diagram. To calculate $\CLSI(\mathcal{L}^\beta_N)$, it remains to bound $D(E^*_\Omega\rho_N | E^*_{fix}\rho_N)$ from above by an entropy production term. This calculation will be done in two parts. We first bound $EP_{\pi_N(a)}(E^*_\Omega\rho_N)$ from above by $EP_{\pi_N(a)}(\rho_N)$. To achieve this, we use a generalization of Lieb concavity theorem \cite{HP} and apply another $2\times2$-matrix trick.
	\begin{lemma}\label{lemma:concavitydouble}
		Let $M$ be a finite dimensional von Neumann algebra with a fixed trace $\tau$ and $N\ssubset M$ be a subalgebra. Let $E:M\rightarrow N$ be the conditional expectation onto $N$. In addition, let $\rho_N\in L_1(N)_+$ be a reference state with density $d_N$ and let $\rho = d_N^{1/2}xd_N^{1/2} \in L_1(M)_+$ be a strictly positive normalized state on $M$. Finally let $a\in M$ (not necessarily self-adjoint) be the generator of a Lindbladian $\mathcal{L}^\beta_a$ such that $\mathcal{L}^\beta_a$ fixes the reference state. Then we have:
		\begin{enumerate}
			\item The following formula for entropy production holds:
			\begin{equation}
				EP_a(\rho) + EP_{a^*}(\rho) = EP_A(\rho_2)
			\end{equation}
			where $A = \begin{pmatrix}
				0 & a \\ a^* & 0
			\end{pmatrix}$. In addition, we denote $\rho_2 := \begin{pmatrix}
				\rho & 0 \\ 0 & \rho
			\end{pmatrix}$ as the doubled-up state of $\rho$ in $\mathbb{M}_2(M)$.
		\item Let $E_2:=E\otimes id:\mathbb{M}_2(M)\rightarrow \mathbb{M}_2(N)$ be the conditional expectation induced by $E$, then we have:
		\begin{equation}
			EP_a(E^*\rho) + EP_{a^*}(E^*\rho)  = EP_A(E_2^*\rho_2) \leq EP_A(\rho_2) = EP_a(\rho) + EP_{a^*}(\rho)\pl.
		\end{equation} 
		\end{enumerate}
	\end{lemma}
	\begin{proof}
		For the first claim, we construct a monotone metric using a $2\times 2$-matrix trick \cite{CM1}\cite{CM2}\cite{JLR}:
		\begin{align}\label{equation:doubleentropy}
			\begin{split}
				EP&_A(\rho_2) := \langle[A, x\otimes id_2],J^{\log}_{\rho_2}[A,\rho]\rangle_{\rho_N\otimes id_2} \\&= \tau\otimes tr_2\bigg(\rho_2^{\frac{1}{2}}[A, x\otimes id_2]^*\rho_2^{\frac{1}{2}}f_\beta^{-1}(\Delta_{\rho_2})R_{\rho_2}^{-1}\rho_2^{\frac{1}{2}}[A, x\otimes id_2]\rho_2^{\frac{1}{2}}\bigg)\\
				&= \tau\otimes tr_2\begin{pmatrix}
					\rho^{1/2}[a^*, x]^*\rho^{\frac{1}{2}}f_\beta^{-1}(\Delta_\rho)R^{-1}_\rho\rho^{\frac{1}{2}}[a^*,x]\rho^{1/2} & 0 \\ 0 & \rho^{1/2}[a,x]^*\rho^{\frac{1}{2}}f_\beta(\Delta_\rho)^{-1}R_\rho^{-1}\rho^{\frac{1}{2}}[a,x]\rho^{\frac{1}{2}}
				\end{pmatrix}
				\\
				&=\tau(\rho^{\frac{1}{2}}[a^*, x]^*\rho^{\frac{1}{2}}f_\beta^{-1}(\Delta_\rho)R^{-1}_\rho\rho^{\frac{1}{2}}[a^*,x]\rho^{\frac{1}{2}} ) + \tau(\rho^{\frac{1}{2}}[a,x]^*\rho^{\frac{1}{2}}f_\beta(\Delta_\rho)^{-1}R_\rho^{-1}\rho^{\frac{1}{2}}[a,x]\rho^{\frac{1}{2}}) \\&= EP_a(\rho) + EP_{a^*}(\rho)
			\end{split}
		\end{align}
		where the first equation should be taken as the definition of $EP_A$ and the second equation follows from the well-known double operator integral representation of the monotone metric \cite{CM1}\cite{HP}\cite{JLR}. Here the function $f_\beta$ is operator monotone and is given by $f_\beta(t):=\int_0^1 e^{s(\beta-\frac{1}{2})}t^sds$ \cite{CM1}. And $\Delta_\rho$ is the modular operator and $R_\rho$ is the right multiplication operator.  \footnote{Although we can define the Lindbladian $\mathcal{L}^\beta_A := e^{-\beta /2}L_A + e^{\beta/ 2}L_{A^*}$ where $L_Ax = 2A^*xA - A^*Ax - xA^*A$, the generator $A$ is not an eigenvector of the modular automorphism group $\sigma_t^{\rho_2}$. Therefore $\mathcal{L}^\beta_A$ is not $\rho_2$-detailed balanced and the theory of Carlen and Maas \cite{CM1}\cite{CM2} does not apply. Nevertheless, we can still define the monotone metric $\langle\cdot,f_\beta(\Delta_{\rho_2}^{-1})R_{\rho_2}^{-1}\cdot\rangle_{\rho_N\otimes id_2}$ as we did in the proof. And the convexity result of Hiai and Petz \cite{HP} still applies.}

		For the second claim, we used a generalization of Lieb concavity theorem due to Hiai and Petz \cite{HP}. Their result applies only to the case were the generator is self-adjoint. Therefore the result cannot be applied directly to $EP_a$ and $EP_{a^*}$.
		
		The equality: $EP_a(E^*\rho) + EP_{a^*}(E^*\rho) = EP_A(E^*_2\rho_2)$ follows directly from equation \ref{equation:doubleentropy} and the formula: $E_2^*\rho_2 = \begin{pmatrix}
			E^*\rho & 0 \\ 0 & E^*\rho
		\end{pmatrix}$. Then applying the convexity result by Hiai and Petz \cite{HP}, we have the inequality. 
	\end{proof}
	\begin{remark}
		In the statement of lemma \ref{lemma:concavitydouble}, we require the reference state to have density in the subalgebra $N$. Hence the conditional expectation with respect to the trace also preserves the reference state:\begin{equation*}
			\rho_N(E x) = \tau(d_N^{1/2}Ex d_N^{1/2}) = \tau(E(d_N^{1/2}xd_N^{1/2})) = \tau(d_N^{1/2}xd_N^{1/2}) = \rho_N(x)\pl.
		\end{equation*}
		This observation is important for the argument of lemma \ref{lemma:concavitydouble} to go through.
	\end{remark}
	\begin{corollary}\label{corollary:omegaFisherDecay}
		 Using the same notation as lemma \ref{lemma:concavitydouble}, for any state $\rho := d_N^{1/2}xd_N^{1/2}$ where $d_N = N_\beta\exp(-\frac{\beta}{2}\pi_N(h))$ is the density of the reference state, then we have:
		 \begin{equation}
		 	EP_{\pi_N(a)}(E^*_\Omega\rho) + EP_{\pi_N(a^*)}(E^*_\Omega\rho) \leq EP_{\pi_N(a)}(\rho) + EP_{\pi_N(a^*)}(\rho)\pl.
		 \end{equation}
	\end{corollary}
	\begin{proof}
		We only need to check all assumptions of lemma \ref{lemma:concavitydouble} hold. Since the density $d_N$ is diagonal in the canonical basis, it is in the subalgebra $\Omega = \bigoplus_\lambda \ell^{n_\lambda}_\infty\otimes \mathbb{B}(W_\lambda)$. All the other assumptions are known to hold by previous discussions. Hence we have the result.
	\end{proof}
	The remaining part of this section will be devoted to prove the following $\CLSI$ inequality
	\begin{proposition}\label{proposition:graph}
		Let $\rho$ be a state in $\Omega$ then there exists a constant $C(\beta)$ depending only on $\beta$ such that
		\begin{equation}
			D(\rho | E^*_{fix}\rho )\leq C(\beta)\big(EP_{\pi_N(a)}(\rho) + EP_{\pi_N(a^*)}(\rho)\big)
		\end{equation}
	\end{proposition}
	Before giving the proof, recall the embedding:
	\begin{equation*}
		\iota:\mathbb{B}(V^{\otimes N})\rightarrow \mathbb{B}(\bigoplus_\lambda V_\lambda)\bigotimes\mathbb{M}_{n_0} 
	\end{equation*}
	where $V^{\otimes N}= \bigoplus_{\lambda\in\mathcal{P}_2(N)}V_\lambda^{\oplus \dim(W_\lambda)}$ is the irreducible decomposition and $n_0 = \max\{\dim(W_\lambda)\}$. The direct sum representation $\pi = \oplus_\lambda \pi_\lambda$ is multiplicity free. Since this embedding intertwines the Lindbladian $\mathcal{L}^\beta_N$ and $\mathcal{L}^\beta_\pi$, we have 
	\begin{equation*}
		\CLSI(\mathcal{L}^\beta_N)\geq \CLSI(\mathcal{L}^\beta_\pi)
	\end{equation*}
	Therefore to prove the proposition, it suffices to prove that there exists a constant $C(\beta)$ such that
	\begin{equation}
		D(\rho | E^*_{fix}\rho) \leq C(\beta)\big(EP_{\pi(a)}(\rho) + EP_{\pi(a^*)}(\rho)\big)
	\end{equation}
	where $\rho$ is a state in $\Omega_\pi := C^*(|\pi(a)|)'\cap C^*(|\pi(a^*)|)'$ and $E_{fix}$ is the conditional expectation onto the fixed point algebra of $\mathcal{L}^\beta_\pi$. 
	
	 Since $\pi$ is multiplicity free, by lemma \ref{lemma:commutant} $\Omega_\pi = \bigoplus_\lambda \ell_\infty^{n_\lambda}$ where $n_\lambda = \dim(V_\lambda)$. Suppose for each irreducible component $\lambda$, there exists a constant $C(\beta,\lambda)$ such that
	 \begin{equation*}
	 	D(\rho_\lambda | E^*_{fix}\rho_\lambda) \leq C(\beta,\lambda)\big(EP_{\pi_\lambda(a)}(\rho_\lambda) + EP_{\pi_\lambda(a^*)}(\rho_\lambda)\big)
	 \end{equation*}
	where $\rho_\lambda$ is a state on $\ell_\infty^{n_\lambda}$ and $E_{fix}$ is the conditional expectation onto $\mathbb{C}$. Then by direct sum decomposition, we can take the constant $C(\beta)$ to be $\max_\lambda\{C(\beta,\lambda)\}$
	
	Therefore, the proposition would be implied by the follwing lemma.
	\begin{lemma}\label{lemma:irrep}
		Let $V_\lambda$ be the irreducible component of highest dimension in $V^{\otimes N}$ (tensor product of fundamental representation), let $\ell_\infty^{n_\lambda}$ be the diagonal subalgerba of $\mathbb{B}(V_\lambda)$ and let $E_{fix}: \ell_\infty^{n_\lambda} \rightarrow \mathbb{C}$ be the projection onto $\mathbb{C}$. Then for any state $\rho$ on $\ell^{n_\lambda}_\infty$, there exists a constant $C(\beta)$ depending only on $\beta$ such that
		\begin{equation}
			D(\rho | E^*_{fix}\rho) \leq C(\beta)\big(EP_{\pi_\lambda(a)}(\rho)+ EP_{\pi_\lambda(a^*)}(\rho)\big)
		\end{equation}
		In particular, $C(\beta)$ is independent of dimension $n_\lambda$.
	\end{lemma}
	The key idea of this proof is to study the relation between $\mathcal{L}^\beta_{\pi_\lambda}$ and the tensor product of $\mathcal{L}^\beta_{\pi_{(1)}}$ (the Lindbladian associated with a single fundamental representation $\pi_{(1)}$). When restricted to $\ell^{n_\lambda}_\infty$, both Lindbladians have the structure of nearest-neighbor Lindbladian, which we formally define as follows.
	\begin{definition}\label{definition:NN}
		Given a finite sequence of real scalars $(\gamma_j)_{0\leq j\leq n}$, consider the following operator in $\mathbb{M}_{n+1}$:\begin{equation*}
			S_\gamma := \sum_j \gamma_j \ket{j-1}\bra{j}
		\end{equation*}
		Then the nearest-neighbor Lindbladian associated with $(\gamma_j)$ is given by:
		\begin{equation}
			\mathcal{L}^\beta_\gamma x:= \beta^{1/2}L_{S_\gamma} x + \beta^{-1/2}L_{S^*_\gamma}x
		\end{equation}
		where
		\begin{equation*}
			L_{S_\gamma}x := 2 S^*_\gamma xS_\gamma - S^*_\gamma S_\gamma x - xS^*_\gamma S_\gamma
		\end{equation*}
	\end{definition}

	For $x  = \sum_{0\leq j \leq n} x_j \ket{j}\bra{j}\in \ell^{n+1}_\infty$, we have
	\begin{equation*}
		L_{S_\gamma}x = \sum_{0\leq j \leq n} 2(\gamma_{j-1}^2x_j\ket{j-1}\bra{j-1} - \gamma_{j}^2x_j\ket{j}\bra{j})
	\end{equation*}
	As an example, for $\mathcal{L}^\beta_{\pi_\lambda}$ restricted to $\ell^{n_\lambda}_\infty$, the corresponding coefficients are given by $\gamma_k^2 = \alpha_{n_\lambda - 1, k}$ for $0\leq k \leq n_\lambda - 1$. 
	
	We will explicitly calculate the difference between the two Lindbladians on $\ell^{n_\lambda}_\infty$. The difference turns out to be another nearest-neighbor Lindbladian. This calculation depends on the explicit description of $V_\lambda$ as the totally symmetric subspace of $V^{\otimes N}$. 
	
	\begin{proof}
		Let $\pi^j := 1\otimes...\otimes\pi_1\otimes...\otimes 1$ be the $j$-th component in the Lie algebra tensor product representation $\pi_{(1)}^{\otimes N}$, then the tensor product Lindbladian is given by
		\begin{align*}
			\begin{split}
				\mathcal{L}^\beta_\otimes x &:= \sum_{1\leq j\leq N}\pi^j(\mathcal{L}^\beta_{\pi_{(1)}})x = \sum_{1\leq j\leq N}\beta^{1/2}\pi^j(L_{\pi_{(1)}(a)})x + \beta^{-1/2}\pi^j(L_{\pi_{(1)}(a^*)})x 
				\\
				&=
				\sum_{1\leq j \leq N}\beta^{1/2}(2\pi^j(a^*)x\pi^j(a) - \pi^j(a^*a)x - x\pi^j(a^*a)) 
				\\
				&+ \beta^{-1/2}(2 \pi^j(a)x\pi^j(a^*) - \pi^j(aa^*)x - x\pi^j(aa^*))
			\end{split}
		\end{align*}
		By the tensor stability of CLSI constant, we know
		\begin{equation}
			\CLSI(\mathcal{L}^\beta_\otimes)  \geq \CLSI(\mathcal{L}^\beta_{\pi_1})
		\end{equation}
		In particular, this constant is independent of $N$.
		
		Let $E_d$ be the diagonal conditional expectation and let $\widehat{L}_a := \sum_{1\leq j \leq N}\pi^j(L_{\pi_{(1)}(a)})$, then we claim:
		\begin{claim}
			There exists coefficients $(\widehat{\gamma}_k)_{0\leq k \leq n_\lambda - 1}$ and $(\widetilde{\gamma}_k)_{0\leq k \leq n_\lambda - 1}$ such that $\alpha_{n_\lambda - 1, k} = \widehat{\gamma}_k^2 + \widetilde{\gamma}_k^2$ and
			\begin{align}
				\begin{split}
					&E_d \widehat{L}_a E_d = L_{S_{\widehat{\gamma}}}
					\\
					& E_d L_{\pi_\lambda(a)} E_d = L_{S_{\widehat{\gamma}}} + L_{S_{\widetilde{\gamma}}} = E_d\widehat{L}_a E_d + L_{S_{\widetilde{\gamma}}}
				\end{split}
			\end{align}
		\end{claim}
		\begin{claimproof}
			Recall $V_\lambda$ has basis $(\ket{n_\lambda-1, k})_{0\leq k \leq n_\lambda -1}$ where $\ket{n_\lambda - 1, k} = \sum_{\substack{A\subset [|n_\lambda-1|] \\ |A| = k}}\frac{1}{\binom{n_\lambda-1}{k}} v_A$ and $v_A = (\bigotimes_{j\in A} e_1)\otimes(\bigotimes_{j\notin A}e_0)$. $e_0, e_1$ are the weight vectors of the fundamental representation. Subscript $A$ indicates the position where the tensor factor is $e_1$. 
			We calculate $E_d \widehat{L}_a E_d$ on the diagonal basis $\ket{n_\lambda-1,k}\bra{n_\lambda-1,k}$. 
			
			Since $\pi^j(a^*a)v_A = \mathbbm{1}_{j\in A}v_A$, we have 
			\begin{equation*}
				\sum_j \pi^j(a^*a)v_A = |A|v_A
			\end{equation*}Hence
			\begin{equation*}
				\sum_j \pi^j(a^*a)\ket{n_\lambda-1,k}\bra{n_\lambda-1,k} = k\ket{n_\lambda-1,k}\bra{n_\lambda-1,k}
			\end{equation*}In addition since $\pi^j(a^*)v_A = \mathbbm{1}_{j\notin A}v_{A \cup \{1\}}$, we have 
			\begin{equation*}
				\sum_j \pi^j(a^*)\ket{n_\lambda-1,k}\bra{n_\lambda-1,k}\pi^j(a)= (k-1)\ket{n_\lambda-1,k-1}\bra{n_\lambda-1,k-1}
			\end{equation*}
			Therefore $\widehat{\gamma}_k^2 = k$ and $\widetilde{\gamma}_k^2 = \alpha_{n_\lambda-1,k} - \widehat{\gamma}_k^2 = k(n_\lambda-1-k)$. 
		\end{claimproof}

		Back to the main proof. Recall the density matrix $d_N := N_\beta \exp(-\frac{\beta}{2}\pi_N(h))$. In terms of tensor product decomposition, it can be written as
		\begin{equation*}
			d_N  = N_\beta \exp(-\frac{\beta}{2}\pi_{(1)}(h))^{\otimes N} = d_\beta^{\otimes N}
		\end{equation*}
		where $d_\beta$ is the corresponding density for the fundamental representation. In particular this density is diagonal. 
		
		In addition, restricted to the commutative subalgebra $\ell^{n_\lambda}_\infty$, both $\mathcal{L}^\beta_\otimes$ and $\mathcal{L}^\beta_{\pi_\lambda}$ are $d_\beta$-self-adjoint. Hence the remainder Lindbladian $\mathcal{L}^\beta_{\widetilde{\gamma}}$ is self-adjoint with respect to the $d_N$-KMS inner produce restricted to $\ell^{n_\lambda}_\infty$. Moreover, the fixed point algebra of the tensor product Lindbladian $\mathcal{L}^\beta_\otimes$ is the tensor product of fixed point algebra of $\mathcal{L}^\beta_{\pi_{(1)}}$, which is the trivial algebra $\mathbb{C}$.
		
		Therefore for any state $\rho = d_N^{1/2}xd_N^{1/2}$ on $\ell^{n_\lambda}_\infty$ ($x$ is diagonal), by CLSI on the two-point space ($\ell^2_\infty$) we have
		\begin{align*}
			\begin{split}
				D(\rho | E^*_{fix}\rho) \leq  \CLSI(\mathcal{L}^\beta_{\pi_{(1)}})^{-1}\tau_N\big((\mathcal{L}^\beta_\otimes x)\log(x)\big)
			\end{split}
		\end{align*}
		where $\tau_N$ is the \textit{state} associated with the diagonal tensor product density $d_N$. 
		
		Since we have
		\begin{equation*}
			\tau_N\big((\mathcal{L}^\beta_\otimes x) \log(x)\big) = \tau_N\bigg(\big((\mathcal{L}^\beta_{\pi_\lambda} - \mathcal{L}^\beta_{\widetilde{\gamma}})x\big)\log(x)\bigg)
		\end{equation*} and $\mathcal{L}^\beta_{\widetilde{\gamma}} \geq 0$ on $\ell^{n_\lambda}_\infty$, then we have
		\begin{equation*}
			\tau_N\big((\mathcal{L}^\beta_\otimes x)\log(x)\big) \leq \tau_N\big((\mathcal{L}^\beta_{\pi_\lambda}x) \log(x)\big) = EP_{\pi_\lambda(a)}(\rho)  + EP_{\pi_\lambda(a^*)}(\rho)
		\end{equation*}	
		The last equation is true because the density is diagonal. Since the CLSI constant is given by the two-point space CLSI constant, it is independent of $N$. 
		
		This completes the proof of the lemma.
	\end{proof}
	\begin{proof}[Proof of the Proposition]
		For each irreducible representation $V_n := \mathbb{C}^{n+1}$ of $SU(2)$, there exists a tensor product of fundamental representation $V^{\otimes n}$ such that $V_n$ is the irreducible component of $V^{\otimes n}$ with the highest dimension. Therefore we can apply \ref{lemma:irrep} to each irreducible representation.
		
		By the discussion before the lemma, for any state $\rho$ in $\Omega$, we have
		\begin{equation*}
			D(\rho | E^*_{fix}\rho )\leq \max_\lambda\{C(\beta,\lambda)\}\big(EP_{\pi_N(a)}(\rho) + EP_{\pi_N(a^*)}(\rho)\big)
		\end{equation*}
		Since $C(\beta, \lambda)$ are independent of $\lambda$, the CLSI constant $\max_\lambda\{C(\beta,\lambda)\}$ only depends on $\beta$.
	\end{proof}
	Now we can state the main theorem of this section.
	\begin{theorem}
		Let $V^{\otimes N}$ be the tensor product of fundamental representation of $\mathfrak{su}(2)$, and let $\rho_N = d_N^{1/2}x d_N^{1/2}$ be a state on $V^{\otimes N}$ with $x\in\mathbb{B}(V^{\otimes N})$. Consider the $d_N$-KMS self-adjoint Lindbladian $\mathcal{L}^\beta_N$ associated with the tensor product representation $\pi_N$. Let $E_{fix}$ be the conditional expectation onto its fixed point algebra. Then there exists a constant $C(\beta)$ depending only $\beta$ such that
		\begin{equation}
			D(\rho_N | E^*_{fix}\rho_N) \leq C(\beta)N^2\big(EP_{\pi_N(a)}(\rho_N) + EP_{\pi_N(a^*)}(\rho_N)\big)
		\end{equation}
	\end{theorem}
	\begin{proof}
		Using the notation as above, $\Omega = C^*(|\pi_N(a)|)' \cap C^*(|\pi_N(a^*)|)'$. Then by chain rule of relative entropy we have\begin{equation*}
			D(\rho_N | E^*_{fix}\rho_N)  = D(\rho_N | E^*_\Omega\rho_N) + D(E^*_\Omega\rho_N | E^*_{fix} \rho_N)
		\end{equation*}
		The first term is bounded above by $C(\beta)N^2\big(EP_{\pi_N(a)}(\rho_N) + EP_{\pi_N(a^*)}(\rho_N)\big)$, and the second term is bounded above by \begin{equation*}
			\big(EP_{\pi_N(a)}(E^*_\Omega\rho_N) + EP_{\pi_N(a^*)}(E^*_\Omega\rho_N)\big) \leq C(\beta)\big(EP_{\pi_N(a)}(\rho_N) + EP_{\pi_N(a^*)}(\rho_N)\big)
		\end{equation*}
		Therefore we have the claimed upper bound.
	\end{proof}
	\begin{corollary}
		$\CLSI$ constant of $\mathcal{L}^\beta_N$ is bounded below by $\frac{C(\beta)}{N^2}$ for some constant $C(\beta) >0$.
	\end{corollary}
	\begin{proof}
		This is immediate from the theorem. 
	\end{proof}
	In the next section, we will give an upper bound on the spectral gap of $\mathcal{L}^\beta_N$ by an explicit construction. It turns out that the spectral gap $\lambda(\mathcal{L}^\beta_N) \lessapprox\frac{1}{N}$. Combining the results from these two section, we have
	\begin{equation*}
		\frac{1}{N^2}\lessapprox \CLSI(\mathcal{L}^\beta_N) \leq 2\lambda(\mathcal{L}^\beta_N)  \lessapprox \frac{1}{N}
	\end{equation*}
	In particular when $N$ goes to infinity, both the CLSI constant and the spectral gap go to 0.
	
	Using the same framework and lemma \ref{corollary:evaporation}, we obtain the following CLSI estimate.
	\begin{theorem}\label{theorem:evaporate}
		Consider the Lindbladian $\mathcal{L}^{(3)}_{N}$generated by $a_{(3)} := \pi_N(a)\pi_N(a^*)\pi_N(a)$.\begin{equation}
			\mathcal{L}^{(3)}_Nx := e^{\beta / 2}L_{a_{(3)}}x + e^{-\beta / 2}L_{a^*_{(3)}}x
		\end{equation}
		where $L_{a_{(3)}}x := 2a_{(3^*)}xa_{(3)} - a^*_{(3)}a_{(3)}x - xa^*_{(3)}a_{(3)}$. Then there exists a constant $C_3(\beta)$ such that
		\begin{equation}
			\CLSI(\mathcal{L}^{(3)}_N) \geq C_3(\beta)
		\end{equation}
	\end{theorem}
	\begin{proof}
		By lemma \ref{corollary:evaporation}, the uniform spectral gap of $|a_{(3)}|$ is bounded below by a constant of order $O(1)$. In addition, since $\pi_N(a^*)\pi_N(a)$ is diagonal, it is easy to see that the fixed point algebra of $\mathcal{L}^{(3)}_N$ is the same as the fixed point algebra of $\mathcal{L}^\beta_N$. Moreover, the same calculation as lemma \ref{lemma:commutant} shows that the conditional expectations $E_{a_{(3)}}:\mathbb{B}(V^{\otimes N})\rightarrow C^*(|a_{(3)}|)'$ and $E_{a_{(3)}^*}:\mathbb{B}(V^{\otimes N})\rightarrow C^*(|a_{(3)}^*|)'$ commute and their intersection $\Omega^{(3)} := C^*(|a_{(3)}|)'\cap C^*(|a_{(3)}^*|)'$ equals $\Omega$. Therefore repeating the same arguments as in section \ref{section:su(2)} we have the lower bound on $\CLSI(\mathcal{L}^{(3)}_N)$. 
	\end{proof}\section{\large Upper Bound on Spectral Gap}\label{section:gap}
	The main purpose of this section is to construct a state $\rho\in S_1(V^{\otimes N})$ that gives an $O(1/N)$ upper bound for the spectral gap of $\mathcal{L}^\beta_N$. First we will construct a vector $\xi$ (not self-adjoint and not positive) in the Schatten 2-class such that the Dirichlet form on $\xi$ is bounded below by $\gtrapprox\frac{||\xi||^2_2}{N}$. Then the state $\rho$ will be a perturbation around the stationary state by $\xi + \xi^*$. This will conclude the proof of our main theorem \ref{theorem:main}. 
	
	The construction of $\xi$ is based on explicit calculations and perturbations around the reference state. Following the same idea, we will show that the noise model generated by the Lindbladian $\mathcal{L}^\beta_N$ induces multiple generalized dephasing channels. A generalized dephasing channel\footnote{This definition is a generalization of what is usually called (single-qubit) dephasing channels. The generalized dephasing channels can act any finite-dimensional Schatten 1-class with the same qualitative features of the usual dephasing channels. Hence the nomenclature is justified.} is a completely positive trace preserving map (CPTP) where the off-diagonal elements of a density matrix decay to zero while the diagonal elements of a density matrix are preserved \cite{KP}\cite{KG}\cite{ABF}.
	
	Moreover, we will identify a set of normalized states where the entropy decay rate under the collective noise model is of the order $O(1/N)$ when the temperature is sufficiently low. This is an example of an "almost"-decoherence-free subspace \cite{LCW}\cite{KLV}\cite{LB}. The states from this subspace is not invariant but meta-stable against the noise. If the noise model is concrete with all the parameters fixed, our method can be used for numerics to calculate a much tighter bounds on the decay rate. Since this section contains several parts, we will divide this section into three subsections.
	
	\subsection{Upper Bound on Spectral Gap}
	Recall from the proof of lemma \ref{lemma:Dirichlet}, we considered the embedding:
	\begin{equation*}
		\iota_N:\mathbb{B}(V^{\otimes N})\rightarrow S_2(V^{\otimes N}): x\mapsto d^{1/4}_N x d^{1/4}_N\pl.
	\end{equation*}
	And the Dirichlet form determined by the Lindbladian $\mathcal{L}^\beta_N$ is given by:
\begin{align*}
	\begin{split}
		\langle\xi,\xi\rangle_L &= \langle\iota_N(x), \iota_N(x)\rangle_L
		\\
		 &=||d_N^{1/4}(\pi_N(a) x -  x\pi_N(a))d_N^{1/4}||^2_2 + ||d_N^{1/4}(\pi_N(a^*) x -  x\pi_N(a^*) )d_N^{1/4}||^2_2 
	\end{split}
\end{align*}
where $||\cdot||^2_2$ is the Hilbert-Schmidt inner product.

Let $V^{\otimes N} \cong \bigoplus_{\lambda\in \mathcal{P}_2(N)} V_\lambda^{\oplus \dim(W_\lambda)}$ be the irreducible decomposition. We will construct our example in the multiplicity free subspace $V := \bigoplus_{\lambda} V_\lambda$.
Since $(\ket{n,j})_{ 0\leq j\leq n =\dim(V_\lambda) }$ form a basis of $V$, a generic element $\xi$ in $S_2(V)$ can be written as
\begin{equation}
	\xi:= \sum_{n,m}\sum_{0\leq j \leq n, 0\leq k\leq m}\xi_{(n,j),(m,k)} \ket{n,j}\bra{m,k}
\end{equation}

Before stating the main result of this section, we first prove the following lemma.
\begin{lemma}\label{lemma:DirichletCalc}
	Suppose the coefficients of $\xi\in S_2(V)$ satisfy the relations: $\xi_{(N,j),(N-2,j)} = e^{-(j-1)\beta/2}\xi_{(N,1),(N-2,1)}$ and $\xi_{n,j} = 0$ for $n < N-2$, then 
	\begin{equation}\label{equation:bound}
		\langle \xi, \xi \rangle_L \geq \frac{2e^{\beta / 2} ||\xi||^2_2}{N} \pl.
	\end{equation}
	In addition, the symmetrization $\xi^* + \xi$ satisfies the same inequality.
\end{lemma}
\begin{proof}
	By rescaling on both sides, we can take the trace in the Hilbert-Schmidt norm to be the standard trace on $\mathbb{B}(V)$. The proof is based on direct computations of the Dirichlet form. Using the definition of the Dirichlet form, we have:	\begin{align*}
		\begin{split}
			\langle\xi,\xi\rangle_L
			&= ||e^{\beta/4}\pi_N(a)\xi - e^{-\beta/4}\xi\pi_N(a)||^2_2 + ||e^{-\beta/4}\pi_N(a^*)\xi - e^{\beta/4}\xi\pi_N(a^*)||^2_2
			\\
			&= e^{\beta/2}||\pi_N(a)\xi||_2^2 + e^{-\beta/2}||\xi\pi_N(a)||^2_2 + e^{-\beta/2}||\pi_N(a^*)\xi||_2^2 + e^{\beta/2}||\xi\pi_N(a^*)||^2_2
			\\
			&- \langle \pi_N(a)\xi, \xi\pi_N(a)\rangle - \langle \xi\pi_N(a), \pi_N(a)\xi\rangle - \langle \pi_N(a^*)\xi, \xi\pi_N(a^*)\rangle - \langle \xi\pi_N(a^*), \pi_N(a^*)\xi\rangle
			\\
			&= e^{\beta/2}||\pi_N(a)\xi||^2_2 + e^{-\beta/2}||\xi \pi_N(a)||^2_2 + e^{-\beta/2} ||\pi_N(a^*)\xi||^2_2 + e^{\beta/2} ||\xi \pi_N(a^*)||^2_2 
			\\
			&+ 2||[\pi_N(a),\xi]||^2_2 - 2||\pi_N(a)\xi||^2_2 - 2||\xi \pi_N(a)||^2_2
			\\
			&= (e^{\beta/2} + e^{-\beta/2} - 2)(||\pi_N(a)\xi||^2_2 + ||\xi \pi_N(a)||^2_2) + 2||[\pi_N(a),\xi]||^2_2 \\
			&+ e^{\beta/2}\tau(\xi[\pi_N(a^*),\pi_N(a)]\xi^*) + e^{-\beta/2}\tau(\xi^*[\pi_N(a),\pi_N(a^*)]\xi)
		\end{split}
	\end{align*}
	Expanding $\xi$ in terms of the canonical basis, we have: \begin{equation*}
		e^{\beta/2}\tau(\xi[a^*,a]\xi^*) + e^{-\beta/2}\tau(\xi^*[a,a^*]\xi) = \sum |\xi_{(n,j),(m,k)}|^2 (e^{-\beta/2}(2k-m) - e^{\beta/2}(2j-n))
	\end{equation*}
	Summation is over $n,m,j,k$. A direct computation involving this derivative term shows:
	\begin{align*}
		\begin{split}
			2||&[\pi_N(a),\xi]||^2_2 - 2 ||\pi_N(a)\xi||^2_2 - 2||\xi\pi_N(a)||^2_2 =-2\tau(\xi^*a^*\xi a)-2\tau(\xi^*a\xi a^*)
			\\
			&= -2\sum \overline{\xi_{(n,j+1),(m,k+1)}}\xi_{(n,j),(m,k)} \alpha_{n,j+1}\alpha_{m,k+1} - 2\sum \overline{\xi_{(n,j-1), (m,k-1)}}\xi_{(n,j),(m,k)}\alpha_{n,j}\alpha_{m,k}
			\\
			&= -4\sum Re(\overline{\xi_{(n,j+1),(m,k+1)}}\xi_{(n,j),(m,k)})\alpha_{n,j+1}\alpha_{m,k+1}
		\end{split}
	\end{align*}
	In addition we have:
	\begin{align*}
		\begin{split}
			&(e^{\beta/2}+e^{-\beta/2})(||\pi_N(a)\xi||^2_2 + ||\xi\pi_N(a)||^2_2) + e^{\beta/2}\tau(\xi\pi_N([a^*,a])\xi^*) + e^{-\beta/2}\tau(\xi^*\pi_N([a,a^*])\xi)
			\\
			&= (e^{\beta/2}+e^{-\beta/2})\sum |\xi_{(n,j),(m,k)}|^2(\alpha_{n,j}^2+\alpha_{m,k+1}^2) + \sum |\xi_{(n,j),(m,k)}|^2(e^{-\beta/2}(2k-m) - e^{\beta/2}(2j-n))
			\\
			&= \sum |\xi_{(n,j),(m,k)}|^2 (e^{-\beta/2}\alpha_{n,j+1}^2 + e^{-\beta/2}\alpha_{m,k+1}^2 + e^{\beta/2}\alpha_{m,k}^2 + e^{\beta/2}\alpha_{n,j}^2)
			\\
			&=\sum |\xi_{(n,j),(m,k)}|^2(e^{\beta/2}(\alpha_{n,j}^2 + \alpha_{m,k}^2) + e^{-\beta/2}(\alpha_{n,j+1}^2 + \alpha_{m,k+1}^2))
			\\
			&=\sum |\xi_{(n,j+1),(m,k+1)}|^2e^{\beta/2}(\alpha_{n,j+1}^2+\alpha_{m,k+1}^2) + |\xi_{(n,j),(m,k)}|^2e^{-\beta/2}(\alpha_{n,j+1}^2 + \alpha_{m,k+1}^2)
			\\
			&=\sum (e^{\beta/2}|\xi_{(n,j+1),(m,k+1)}|^2 + e^{-\beta/2}|\xi_{(n,j),(m,k)}|^2)(\alpha^2_{n,j+1} + \alpha^2_{m,k+1})
		\end{split}
	\end{align*}
	Combining the two computations, we have:
	\begin{align}\label{equation:directcomp}
		\begin{split}
			\langle \xi,\xi\rangle_L &= \sum (e^{\beta/2}|\xi_{(n,j+1),(m,k+1)}|^2 + e^{-\beta/2}|\xi_{(n,j),(m,k)}|^2)(\alpha^2_{n,j+1} + \alpha^2_{m,k+1}) 
			\\
			&
			-4\sum Re(\overline{\xi_{(n,j+1),(m,k+1)}}\xi_{(n,j),(m,k)})\alpha_{n,j+1}\alpha_{m,k+1}
		\end{split}
	\end{align}
	By the assumptions on the coefficients of $\xi$, we have:
	\begin{align*}
		\begin{split}
			\langle \xi,\xi\rangle_L &= 2e^{\beta/2}\sum_{1\leq j \leq N-2}e^{-(j-1)\beta}|\xi_{(N,1),(N_2,1)}|^2(\alpha_{N,j}-\alpha_{N-2,j})^2
			\\
			& =8e^{3\beta/2}|\xi_{(N,1),(N-2,1)}|^2\sum_{1\leq j \leq N-2}\frac{je^{-j\beta}}{(\sqrt{N-j+1} + \sqrt{N-j-1})^2}
			\\
			& > \frac{2e^{\beta/2}|\xi_{(N,1),(N-2,1)}|^2}{N}\sum_{1\leq j\leq N-2}e^{-(j-1)\beta}
		\end{split}
	\end{align*}
	In the last step, we used $\sqrt{N-j+1} + \sqrt{N-j-1} < 2\sqrt{N}$ for $1\leq j\leq N-2$ and $\sum_{1\leq j \leq N-2} je^{-j\beta} > \sum_{1\leq j\leq N-2}e^{-j\beta}$. Since \begin{equation*}
		||\xi||^2_2 = |\xi_{(N,1),(N-2,1)}|^2\sum_{1\leq j\leq N-2}e^{-(j-1)\beta}
	\end{equation*}
	then we have \begin{equation*}
		\langle\xi,\xi\rangle_L \geq \frac{2e^{\beta/2}\langle\xi,\xi\rangle_N}{N} \pl.
	\end{equation*}
	For the symmetrized version, by equation \ref{equation:directcomp} we have:
	\begin{equation*}
		\langle \xi^* + \xi, \xi^* +\xi\rangle_L = \langle \xi^*, \xi^*\rangle_L + \langle \xi,\xi\rangle_L \pl.
	\end{equation*}
	In addition, since $\xi$ is completely off-diagonal we have:
	\begin{equation*}
		\langle\xi^*+\xi,\xi^*+\xi\rangle = \langle\xi^*,\xi^*\rangle + \langle\xi,\xi\rangle \pl.
	\end{equation*}
	Therefore $\xi^* + \xi$ satisfies equation \ref{equation:bound} with the same constant.
\end{proof}
\begin{proposition}\label{proposition:gapvector}
	There exists a normalized positive matrix $x\in \mathbb{B}(V)$ such that $\iota_N(x)$ is a state and: 
	\begin{equation}
		\langle\iota_N(x),\iota_N(x)\rangle_L \geq \frac{C(\beta)\langle x - E_{fix}x, x - E_{fix}x\rangle_N}{N}
	\end{equation}
	where $\langle\cdot,\cdot\rangle_N$ is the KMS-inner product associated with $d_N$, $C(\beta)$ is a constant depending only on $\beta$ and $E_{fix}$ is the conditional expectation onto the fixed point algebra $\bigoplus_\lambda \mathbb{C}$.
\end{proposition}
\begin{proof}
	We consider the candidate matrix $\widetilde{x}$ such that $\iota_N(\widetilde{x}) = \xi$, where $\xi$ is constructed in Lemma \ref{lemma:DirichletCalc}. Up to now we have left $\xi_{(N,1),(N-2,1)}$ as a free parameter. To fix this parameter, we calculate the operator $\widetilde{x}$ explicitly. Since $\iota_N(\widetilde{x}) = \xi_{(N,1), (N-2, 1)}\sum_{1\leq j \leq N-2}e^{-(j-1)\beta / 2}\ket{N,j}\bra{N-2, j}$, then we have:
	\begin{align*}
		\begin{split}
			\widetilde{x} &= d_N^{-1/4}\xi d_N^{-1/4} = N_\beta^{-1/2}\xi_{(N,1),(N-2,1)}\sum_{1\leq j\leq N-2}e^{-\frac{(j-1)\beta}{2}} e^{\frac{\beta\pi_N(h)}{8}}\ket{N,j}\bra{N-2,j}e^{\frac{\beta\pi_N(h)}{8}}
			\\
			&= N_\beta^{-1/2}\xi_{(N,1),(N-2,1)}\sum_{1\leq j\leq N-2} e^{\frac{\beta N}{4}-\frac{\beta}{4}}e^{-j\beta}\ket{N,j}\bra{N-2,j}
		\end{split}\pl.
	\end{align*}
	We choose $\xi_{(N,1),(N-2,1)} := N_\beta^{1/2}e^{-\beta N/ 4 + 5\beta / 4}$, then $\widetilde{x}^* +  \widetilde{x} = \sum_{1\leq j \leq N-2}e^{-(j-1)\beta}(\ket{N,j}\bra{N-2,j} + \ket{N-2,j}\bra{N,j})$. The spectral radius of $\widetilde{x}^* + \widetilde{x}$ is bounded above by 2. Therefore $x := 1 - \frac{\widetilde{x}^* + \widetilde{x}}{2} \geq 0$ is a positive matrix. In addition, since the fixed point algebra is diagonal, we have:
	\begin{equation*}
		x - E_{fix}(x) = x - 1 = - \frac{\widetilde{x}^* + \widetilde{x}}{2} \pl.
	\end{equation*}
	And since the identity operator is in the fixed point algebra, we have:
	\begin{equation*}
		\langle\iota_N(x),\iota_N(x)\rangle_L = \frac{1}{4}\langle\iota_N(\widetilde{x}^*+\widetilde{x}),\iota_N(\widetilde{x}^*+\widetilde{x})\rangle_L
	\end{equation*}
	Therefore we have:
	\begin{equation*}
		\langle\iota_N(x),\iota_N(x)\rangle_L \geq \frac{e^{\beta /2}\langle x-E_{fix}x,x-E_{fix}x\rangle_N}{2N} \pl. \qedhere
	\end{equation*}
\end{proof}
\begin{corollary}\label{corollary:dimdep}
	For the $\mathfrak{su}(2)$-Lindbladian $\mathcal{L}^\beta_N$, there exist constants $C_1(\beta), C_2(\beta) > 0$ such that
	\begin{equation}
		\frac{C_1(\beta)}{N^2} \leq CLSI(\mathcal{L}^\beta_N) \leq 2\lambda(\mathcal{L}^\beta_N) \leq \frac{C_2(\beta)}{N}\pl.
	\end{equation}
\end{corollary}
This completes the proof of our main theorem \ref{theorem:main}. 
\subsection{Embedded Generalized Dephasing Channels} In this subsection, we will show that the the collective noise model generated by the Lindbladian $\mathcal{L}^\beta_N$ can be decomposed into a family of generalized dephasing channels. Each channel acts independently on a subset of normalized states. We first define a generalized dephasing channel:
\begin{definition}
	A generalized dephasing channel on the Schatten 1-class $S_1(\mathbb{C}^d)$is a completely positive trace preserving map:
	\begin{equation}
		\mathcal{D}: S_1(\mathbb{C}^d)\rightarrow S_1(\mathbb{C}^d): \sum_{ij}\rho_{ij}e_{ij} \mapsto \sum_{i\neq j}\alpha_{ij}\rho_{ij}e_{ij} + \sum_i \rho_{ii}e_{ii}
	\end{equation}
	where $|\alpha_{ij}| < 1$ and $\{e_{ij}\}_{1\leq i,j\leq d}$ is a fixed family of matrix units. 
\end{definition}
These channels generalize the usual single qubit dephasing channel in the sense that the off-diagonal elements strictly decay and the diagonal elements are preserved by the channel. The usual single qubit dephasing channel is a generalized dephasing channel because the single qubit dephasing channel can be written as:
\begin{equation*}
	\mathcal{D}(\sum_{ij}\rho_{ij}e_{ij}) = \begin{pmatrix}
		\rho_{11} & \alpha\rho_{12} \\ \alpha^*\rho_{21} & \rho_{22}
	\end{pmatrix}
\end{equation*}
where $|\alpha| < 1$.

To decompose the Lindbladian, we need to decompose the algebra $\mathbb{B}(V^{\otimes N})$ into a family of mutually orthogonal operator subspaces. We first recall the irreducible decomposition of $V^{\otimes N} = \bigoplus_{\lambda\in\mathcal{P}_2(N)} V_\lambda\otimes W_\lambda$ where $W_\lambda$'s are the multiplicity subspaces. In terms of the canonical basis under the representation of $\mathfrak{su}(2)$, a complete orthonormal basis of $V$ is given by $\{\ket{n,j,\alpha}:=\ket{n,j}\otimes\ket{\alpha}\}$ where $\ket{n,j}\in V_\lambda$ is the canonical basis of the $(n+1)$-dimensional irreducible component and $\ket{\alpha}\in W_\lambda$ is a fixed basis such that in terms of the matrix units $\{\ket{n,j,\alpha}\bra{m,k,\beta}\}$ the reference state $d_N$ is diagonal and has the form:
\begin{equation*}
	d_N = \frac{1}{\cosh(\beta/2)^N}\bigoplus_\lambda \sum_{0\leq j \leq n, \alpha}e^{-\frac{\beta}{2} (n-2j)}\ket{n,j,\alpha}\bra{n,j,\alpha}\pl.
\end{equation*}
Henceforth, we shall always work with the fixed matrix units $\{\ket{n,j,\alpha}\bra{m,k,\beta}\}$. To motivate the decomposition of $\mathbb{B}(V^{\otimes N})$, we perform a simple calculation:
\begin{lemma}\label{lemma:simplecalc}
	The Lindbladian $\mathcal{L}^\beta_N$ acts on the matrix units $\ket{n,j,\alpha}\bra{m,k,\beta}$ by:
	\begin{align}
		\begin{split}
			\mathcal{L}^\beta_N(\ket{n,j,\alpha}\bra{m,k,\beta}) &= -(e^{\frac{\beta}{2}}(\alpha_{n,j}^2 + \alpha_{m,k}^2)+ e^{-\frac{\beta}{2}}(\alpha_{n,j+1}^2 + \alpha_{m,k+1}^2))\ket{n,j,\alpha}\bra{m,k,\beta} \\&+ 2e^{\frac{\beta}{2}}\alpha_{n,j+1}\alpha_{m,k+1}\ket{n,j+1,\alpha}\bra{m,k+1,\beta} \\&+ 2e^{-\frac{\beta}{2}}\alpha_{n,j}\alpha_{m,k}\ket{n,j-1,\alpha}\bra{m,k-1,\beta} \pl.
		\end{split}
	\end{align}
\end{lemma}
\begin{proof}
	Recall that the Lindbladian only acts on $V_\lambda$ components. Then the proof is by a direct calculation following section \ref{section:su(2)}.
\end{proof}
This calculation suggests that we should decompose $\mathbb{B}(V^{\otimes N})$ into the following family of mutually orthogonal operator subspaces:
\begin{definition}
	For fixed parameters $n,m,d$ where $1\leq n \leq m \leq N$ and $-m\leq d\leq n$, the operator space $\mathbb{B}^N_{n,m,d}$ is an operator subspace of $\mathbb{B}(V^{\otimes N})$ generated by the matrix units:
	\begin{equation*}
		\{\ket{n,j,\alpha}\bra{m,j-d, \beta}, \ket{m,j-d,\beta}\bra{n,j,\alpha}: \max(0,d)\leq j\leq \min(n - d, m), \ket{\alpha}\in W_\lambda, \ket{\beta}\in W_{\lambda'}\}
	\end{equation*}
	where $W_\lambda$ is the multiplicity subspace of the $(n+1)$-dimensional irreducible component and $W_{\lambda'}$ is the multiplicity subspace of the $(m+1)$-dimensional irreducible component.
\end{definition}
Lemma \ref{lemma:simplecalc} shows that for all permissible parameters $n,m,d$, we have:
\begin{equation}\label{equation:preservation}
	\mathcal{L}^\beta_N\mathbb{B}^N_{n,m,d} \ssubset \mathbb{B}^N_{n,m,d}\pl.
\end{equation}
Under the KMS inner product on $S_2(V^{\otimes N})$, the subspaces $\{\iota_N(\mathbb{B}^N_{n,m,d})\}$ are mutually orthogonal because the reference density $d_N$ is diagonal and for any two elements $x\in \mathbb{B}^M_{n,m,d}$ and $y\in \mathbb{B}^N_{n',m',d'}$ where the parameters are different $(n,m,d)\neq (n',m',d')$, we have:
\begin{equation*}
	\langle\iota_N(x), \iota_N(y)\rangle = \tau_N(d_N^{1/2}x^*d_N^{1/2}y) = 0\pl.
\end{equation*}
This is due to fact that the diagonal elements of $\iota_N(x^*)\iota_N(y)$ are all zero. Hence we have orthogonal operator subspace decomposition:
\begin{equation}
	\mathbb{B}(V^{\otimes N}) = \bigoplus_{n,m,d}\mathbb{B}^N_{n,m,d}\pl.
\end{equation}
Using this decomposition, the subalgebra $\Omega$ (lemma \ref{lemma:commutant}) can be written as:
\begin{equation*}
	\Omega = \bigoplus_{n}\mathbb{B}^N_{n,n,0}\pl.
\end{equation*}
The quantum Markov semigroup $\{e^{-t\mathcal{L}^\beta_N}\}_{t\geq 0}$ acting on $\Omega$ is generated by a nearest-neighbor Lindbladian (definition \ref{definition:NN}). As discussed in section \ref{section:su(2)}, the Lindbladian $\mathcal{L}^\beta_N$ on this subalgebra is completely classical. In the next proposition, we show that the true quantum component of this collective noise model can be understood as a collection of generalized dephasing channels.
\begin{proposition}
	Let $\mathbb{B}_{quantum} := \Omega_{fix}\bigoplus (\bigoplus_{n\neq m,d}\mathbb{B}^N_{n,m,d})\bigoplus (\bigoplus_{n, d\neq 0}\mathbb{B}^N_{n,n,d})$ be the operator subsystem of $\mathbb{B}(V^{\otimes N})$ that excludes elements in $\Omega - \Omega_{fix}$. Then for each $t\geq 0$, the quantum channel $e^{-t\mathcal{L}^\beta_N}$ restricted to $\mathbb{B}_{quantum}$ is a generalized dephasing channel and it can be decomposed as:
	\begin{equation}
		e^{-t\mathcal{L}^\beta_N} = id \bigoplus (\bigoplus_{n\neq m,d}\Phi^t_{n,m,d})\bigoplus(\bigoplus_{n,d\neq 0}\Phi^t_{n,n,d})
	\end{equation}
	where for each component, $\Phi^t_{n,m,d} := e^{-t\mathcal{L}^\beta_N}|_{\mathbb{B}^N_{n,m,d}}$ and the restricted channel $e^{-t\mathcal{L}^\beta_N}|_{\Omega_{fix}\bigoplus\mathbb{B}^N_{n,m,d}} = id\bigoplus \Phi^t_{n,m,d}$ is a generalized dephasing channel.
\end{proposition}
\begin{proof}
	Since $1\in\Omega_{fix}\ssubset \mathbb{B}_{quantum}$ and $\mathbb{B}_{quantum}$ is closed under conjugation, $\mathbb{B}_{quantum}$ is well-defined as an operator system. For any $x\in\mathbb{B}_{quantum}$, $x$ can be decomposed into its diagonal part and off-diagonal part: $x = x_{diag} + \delta x$. Since $x_{diag}\in\Omega_{fix}$, then it is preserved under the channel $e^{-t\mathcal{L}^\beta_N}$. Since the off-diagonal part $\delta x$ strictly decays to zero under the quantum Markov semigroup and since by lemma \ref{lemma:simplecalc} $e^{-t\mathcal{L}^\beta_N}(\mathbb{B}^N_{n,m,d})\subset\mathbb{B}^N_{n,m,d}$, then the claim is true.
\end{proof}
The restricted channels can be made explicit if we can calculate the spectrum of $\mathcal{L}^\beta_N|_{\mathbb{B}^N_{n,m,d}}$. In the next section we will provide an algorithm to compute the spectrum of $\mathcal{L}^\beta_N|_{\mathbb{B}^N_{n\neq m,0}}$ and use this algorithm to identify some meta-stable states in low temperature.
\subsection{"Almost"-Decoherence-Free Subspace and Slow Entropy Decay}
In this subsection, we construct states that are meta-stable under the collective noise generated by $\mathcal{L}^\beta_N$ at low temperature. Using Pinsker's inequality, we will obtain a lower bound on the relative entropy decay and show that the meta-stable states have $O(1)$ relative entropy decay. The construction of these states depend on an explicit algorithm to calculate the spectrum of $\mathcal{L}^\beta_N|_{\mathbb{B}^N_{n\neq m, 0}}$.

Without loss of generality, assume $ n< m$. Since the Lindbladian $\mathcal{L}^\beta_N$ does not act on the multiplicity subspaces, to calculate the spectrum it suffices to consider $\mathcal{L}^\beta_N$'s action on a typical element $x := \sum_{0\leq j\leq n} x_j\ket{n,j}\bra{m,j}$ where $x_j$'s are complex parameters. By lemma \ref{lemma:simplecalc} and equation \ref{equation:preservation}, there exists complex parameters $y_j$'s such that:
\begin{equation}
	\mathcal{L}^\beta_N x = \sum_j y_j\ket{n,j}\bra{m,j}\pl.
\end{equation} 
For $0 < j < n$, we have:
\begin{align}\label{equation:eigencalc}
	\begin{split}
		y_j &= -x_j(e^{\beta/2}(\alpha_{n,j}^2+\alpha_{m,j}^2) + e^{-\beta/2}(\alpha_{n,j+1}^2+\alpha_{m,j+1}^2)) \\&+ 2x_{j-1}e^{\beta/2}\alpha_{n,j}\alpha_{m,j} +2x_{j+1}e^{-\beta/2}\alpha_{n,j+1}\alpha_{m,j+1}\pl.
	\end{split}
\end{align}
For $j = 0$, we have:
\begin{align}\label{equation:eigencalc2}
	\begin{split}
		y_0 = -x_0e^{-\beta/2}(\alpha_{n,1}^2 + \alpha_{m,1}^2) + 2x_1e^{-\beta/2}\alpha_{n,1}\alpha_{m,1}\pl.
	\end{split}
\end{align}
And for $j = n$, we have:
\begin{align}\label{equation:eigencalc3}
	\begin{split}
		y_n = -x_n(e^{\beta/2}(\alpha_{n,n}^2 + \alpha_{m,n}^2) + e^{-\beta/2}\alpha_{m,n+1}^2) + 2x_{n-1}e^{\beta/2}\alpha_{n,n}\alpha_{m,n}
	\end{split}
\end{align}
By simple induction, we claim:
\begin{lemma}\label{lemma:poly}
	If $x$ is an eigenvector of $\mathcal{L}^\beta_N$, then $x_0 \neq 0$ and for all $1<j\leq n$, we have:
	\begin{equation}
		x_j = \frac{1}{x_0^{j-1}}f_j(x_0,x_1)
	\end{equation}
	where $f_j$ is a homogeneous polynomial with degree $\deg(f_j) = j$. 
	
	In addition, the coefficient of the monomial $x_1^{j}$ in the homogeneous polynomial $f_j$ is given by $(2e^{-\beta/2}\alpha_{n,1}\alpha_{m,1})^j$.
	
	Finally, the eigenvalue corresponding to $x$ is given by $e^{-\beta/2}(-\alpha_{n,1}^2 - \alpha_{m,1}^2 + 2\alpha_{n,1}\alpha_{m,1}(\frac{x_1}{x_0}))$. And the ratio $\frac{x_1}{x_0}$ is given by the solution to a degree-$(n+1)$ polynomial.
\end{lemma}
\begin{proof}
	Let $\gamma$ be the eigenvalue corresponding to $x$, then for all $j$ we have $y_j = \gamma x_j$. In particular for all $1\leq j\leq n$, we have $ y_jx_0 = x_jy_0$. If $x_0 = 0$ and $x$ is an eigenvector, then by equation \ref{equation:eigencalc2} we have $x_1 = 0$. Then a simple induction shows that $x_j = 0$ for all $j$. Thus $x = 0$. Hence for any nontrivial eigenvector $x$, $x_0 \neq 0$.
	
	For simplicity, let $\{\beta_{jj'}\}_{|j'-j| = 1}$ be the matrix entries of $\mathcal{L}^\beta_N$ such that the equations \ref{equation:eigencalc}\ref{equation:eigencalc2}\ref{equation:eigencalc3} can be written as:
	\begin{align*}
		&y_0 = \beta_{00}x_0 + \beta_{01}x_1\\
		&y_j = \beta_{j,j-1}x_{j-1} + \beta_{jj}x_j + \beta_{j,j+1}x_{j+1}\\
		&y_n = \beta_{n,n-1}x_{n-1} + \beta_{nn}x_n\pl.
	\end{align*}
	Therefore for $1\leq j \leq n-1$ we have:
	\begin{equation}
		x_j(\beta_{00}x_0 + \beta_{01}x_1) = x_0(\beta_{j,j-1}x_{j-1} + \beta_{jj}x_j + \beta_{j,j+1}x_{j+1})\pl.
	\end{equation}
	For $j = 1$, we have:
	\begin{equation}
		x_2 = \frac{1}{x_0}f_2(x_0,x_1) = \frac{1}{
			\beta_{12}x_0}(\beta_{00}x_0x_1 + \beta_{01}x_1^2 - \beta_{10}x_0^2 -\beta_{11}x_1x_0)\pl.
	\end{equation}
	It's clear that $f_2$ is a degree-2 homogeneous polynomial and the coefficient of the monomial $x_1^2$ is $\beta_{01} = 2e^{-\beta/2}\alpha_{n,1}\alpha_{m,1}$.
	
	Assume the claims are true upto $j \leq n-1$, then for $j+1$, we have:
	\begin{align}
		\begin{split}
			&x_j(\beta_{00}x_0 + \beta_{01}x_1) = x_0(\beta_{j,j-1}x_{j-1} + \beta_{jj}x_{j} + \beta_{j,j+1}x_{j+1})\\
			&x_{j+1} = \frac{1}{\beta_{j,j+1}x_0}\bigg(\frac{f_j(x_0,x_1)}{x_0^{j-1}}(\beta_{00}x_0 + \beta_{01}x_1 - \beta_{jj}x_0) - \frac{f_{j-1}(x_0,x_1)}{x_0^{j-2}}\beta_{j,j-1}x_0\bigg)\pl.
		\end{split}
	\end{align}
	Note $\beta_{j,j+1} = 2e^{-\beta/2}\alpha_{n,j+1}\alpha_{m,j+1} \neq 0$ for $0\leq j < n$. Therefore we can define:
	\begin{equation*}
		f_{j+1}(x_0,x_1) := \frac{1}{\beta_{j,j+1}}(f_j(x_0,x_1)(\beta_{00}x_0 + \beta_{01}x_1 - \beta_{jj}x_0) - f_{j-1}(x_0,x_1)\beta_{j,j-1}x_0^2)\pl.
	\end{equation*}
	By the inductive hypothesis, the polynomial $f_{j+1}$ is still homogeneous and has degree $j+1$. In addition, the coefficient of the monomial $x_1^{j+1}$ is given by $\beta_{01}^j = (2e^{-\beta/2}\alpha_{n,1}\alpha_{m,1})^j$. 
	
	To complete the proof, we only need to show that the ratio $\frac{x_1}{x_0}$ is the solution to a degree-$(n+1)$ polynomial. This follows from the last constraint equation $y_nx_0 = y_0x_n$. By equation \ref{equation:eigencalc3}, we have:
	\begin{align*}
		\begin{split}
			&x_0(\beta_{n,n-1}\frac{f_{n-1}(x_0,x_1)}{x_0^{n-2}}+\beta_{nn}\frac{f_n(x_0,x_1)}{x_0^{n-1}}) = \frac{f_n(x_0,x_1)}{x_0^{n-1}}(\beta_{00}x_0 + \beta_{01}x_1)
			\\
			& \beta_{n,n-1}x_0^2f_{n-1}(x_0,x_1) + (\beta_{nn} - \beta_{00})x_0f_n(x_0,x_1) -\beta_{01}x_1f_n(x_0,x_1)=0
		\end{split}
	\end{align*}
	where the polynomial on the left hand side of the last equation $F(x_0,x_1)$ is a homogeneous polynomial of degree-$(n+1)$. Hence divide both sides by $x_0^{n+1}$, the left hand side of the last equation becomes a degree-$(n+1)$ polynomial $g(\frac{x_1}{x_0}):=\frac{F(x_0,x_1)}{x_0^{n+1}} = F(1, \frac{x_1}{x_0})$. Therefore the equation shows that the ratio $\frac{x_1}{x_0}$ must be a solution to $g = 0$. Note the coefficient of the monomial $x_1^{n+1}$ is $\beta_{01}^n$.
	
	Finally, the formula for the eigenvalue is a simple consequence of equation \ref{equation:eigencalc2}.
\end{proof}
\begin{corollary}
	The spectrum of $\mathcal{L}^\beta_N|_{\mathbb{B}^N_{n\neq m,0}}$ is given by:
	\begin{equation}\label{equation:set}
		\{e^{-\frac{\beta}{2}}(-n-m+2\sqrt{nm}\gamma): g(\gamma) = 0\}
	\end{equation}
	where $g$ is the degree-$(n+1)$ polynomial defined in the proof of lemma \ref{lemma:poly}.
\end{corollary}
\begin{proof}
	For each $\gamma$ and for each pair of labels $\alpha,\beta$ of the multiplicity subspace, lemma \ref{lemma:poly} constructs an eigenvector $x_{\alpha,\beta} = \sum_j x_j\ket{n,j,\alpha}\bra{m,j,\beta}$ with eigenvalue $e^{-\frac{\beta}{2}}(-n-m+2\sqrt{nm}\gamma)$. Since the polynomial $g$ has real coefficients, then if $\gamma$ is a solution, $\gamma^*$ is also a solution. In fact, since $\mathcal{L}^\beta_N$ is self-adjoint, $\gamma\in\mathbb{R}$. Thus for each root $\gamma$, the corresponding eigenvectors in $\mathbb{B}^N_{n,m,d}$ are $\{x_{\alpha,\beta},x_{\alpha,\beta}^*\}$. The cardinality of this set is $2\dim(W_\lambda)\dim(W_{\lambda'})$ where $W_\lambda$ (respectively $W_{\lambda'}$) is the multiplicity subspace corresponding to the unique $(n+1)$-dimensional irreducible component (respectively the $(m+1)$-dimensional irreducible component). Since $\dim\mathbb{B}^N_{n,m,d} = 2(n+1)\dim(W_\lambda)\dim(W_{\lambda'})$ and since the polynomial $g$ has $n+1$ roots (counting multiplicity), the entire spectrum of $\mathcal{L}^\beta_N|_{\mathbb{B}^N_{n\neq m,0}}$ is given by equation \ref{equation:set} .
\end{proof}
Therefore, lemma \ref{lemma:poly} provides an algorithm to find the spectrum of $\mathcal{L}^\beta_N|_{\mathbb{B}^N_{n\neq m,0}}$. In general, the concrete calculation of this spectrum requires numerical tools. However, without performing the full calculation, we can still gain insights into the existence of low-temperature meta-stable states. This is the content of the next proposition.
\begin{proposition}\label{proposition:slowdecay}
	For fixed parameters $1\leq n < m \leq N$, we have
	\begin{equation}
		\min_\gamma |\gamma| \leq \frac{e^{\beta/2}}{2\sqrt{nm}}
	\end{equation}
	where the minimum is taken over all roots of the polynomial $g$ constructed in the proof of lemma \ref{lemma:poly}.
	
	In addition, for $\beta = O(\log N)$ and $n,m = O(N)$, there exists $x = x^*\in\mathbb{B}^N_{n\neq m,0}$ such that $e^{-t\mathcal{L}^\beta_N}x = e^{-\Gamma t}x$ where the decay rate $\Gamma = O(1/N)$.
\end{proposition}
\begin{proof}
	Recall from the proof of lemma \ref{lemma:poly}, $g(\frac{x1}{x_0}) = F(1, \frac{x_1}{x_0})$ where $F(x_0,x_1)$ is a degree-$(n+1)$ homogeneous polynomial. Thus if $\gamma$ is a solution to $g = 0$, then $\frac{1}{\gamma}$ is a solution to $F(\frac{x_0}{x_1}, 1) = 0$. Again by lemma \ref{lemma:poly}, the coefficient of the monomial $x_1^{n+1}$ in $F(x_0,x_1)$ is $\beta_{01}^n$. Hence the constant term in the polynomial $F(\frac{x_0}{x_1},1)$ is $\beta_{01}^n$. Therefore the product of all the roots of $F(\frac{x_0}{x_1},1) = 0$ is $\beta_{01}^n$. More precisely, we have:
	\begin{equation}
		\prod_\gamma \frac{1}{|\gamma|} = \beta_{01}^n
	\end{equation}
	where the product is taken over all roots $\gamma$ of $g = 0$. Thus we have: $\beta_{01}^n \leq (\frac{1}{\min_\gamma|\gamma|})^{n+1}$. Since the roots $\gamma$ are real and since $\lim_{t\rightarrow\infty}e^{-t\mathcal{L}^\beta_N}x = 0$ for all $x\in \mathbb{B}^N_{n\neq m,0}$, the eigenvalues of $\mathcal{L}^\beta_N$ on $\mathbb{B}^N_{n\neq m, 0}$ must be non-negative. Hence for all $\gamma$, $e^{-\beta/2}(-n-m+2\sqrt{nm}\gamma) \geq 0$. In particular $\gamma \geq 0$. Thus we have:
	\begin{equation*}
		\min \gamma \leq \frac{1}{\beta_{01}} = \frac{e^{\beta/2}}{2\sqrt{nm}}
	\end{equation*}
	where the minimum is taken over all roots of $g$. In addition, the minimum decay rate is given by $\min_\gamma \big(e^{-\beta/2}(-n-m+2\sqrt{nm}\gamma)\big) = e^{-\beta/2}(-n-m+2\sqrt{nm}\min\gamma) \leq e^{-\beta/2}(-n-m+2\sqrt{nm}/\beta_{01})$. When $\beta \geq 2\log (2N)$, the minimum decay rate is bounded above by:
	\begin{equation*}
		e^{-\beta/2}(-n-m + e^{\beta/2}) = 1 - e^{-\beta/2}(n+m) = O(1)\pl.
	\end{equation*}
	In particular, if both $n$ and $m$ are of the order $O(N)$ and the parameter $\beta = O(\log N)$, then we have\footnote{For example, if $n = N - 2$ and $m = N$ and $\beta = 2\log(2N)$, then $1 - e^{-\beta/2}(n+m) = \frac{1}{N}$.}:
	\begin{equation}
		1 - e^{-\beta/2}(n+m) = O(1/N)\pl.
	\end{equation}
	Therefore for $n,m = O(N)$ and $\beta = O(\log N)$, the self-adjoint eigenvector $x$ corresponds to the minimum root $\gamma$ is the meta-stable operator that we are looking for. If $x$ is not self-adjoint, we can always replace $x$ with $\frac{x+x^*}{2}$ since $\frac{x+x^*}{2}$ has the same eigenvalue as $x$. 
\end{proof}
For each pair of $1\leq n\neq m\leq N$, proposition \ref{proposition:slowdecay} gives a self-adjoint element $x_{n,m}\in\mathbb{B}^N_{n\neq m, 0}$ that decays with at most $O(1)$-rate. For each $x_{n,m}$, let $r(x_{n,m})$ be its spectral radius. Then for each $\eta\in[-1,1]$, $\widetilde{x}_{n,m}(\eta):=\mathbbm{1} + \frac{\eta}{r(x_{n,m})}x_{n,m}$ is positive semidefinite and has unit trace. Thus by Pinsker's inequality,we have:
\begin{corollary}\label{corollary:pinsker}
	Let $\rho_{n,m}(\eta) := \iota_N(\widetilde{x}_{n,m}(\eta)) = d_N^{1/2}\widetilde{x}_{n,m}(\eta)d_N^{1/2}$ be the state corresponding to $\widetilde{x}_{n,m}(\eta)$. Then we have the relative entropy lower bound:
	\begin{equation}
		D(e^{-t(\mathcal{L}^\beta_N)^*}\rho_{n,m}(\eta)| E_{fix}^*\rho_{n,m}(\eta)) \geq \frac{e^{-2\Gamma t}}{2}||(id - E_{fix}^*)\rho_{n,m}(\eta)||_1^2
	\end{equation}
	where $\Gamma$ is the eigenvalue of $x_{n,m}$ and it is of the order $O(1)$.
	
	If $\beta = O(\log N)$ and $n,m = O(N)$, then the relative entropy decay rate is of the order $O(1/N)$.
\end{corollary}
\begin{proof}
	The proof is a direct application of Pinsker's inequality. By definition, $E_{fix}^*\rho_{n,m}(\eta) = d_N$ and $e^{-t(\mathcal{L}^\beta_N)^*}\rho_{n,m}(\eta) = d_N + \frac{\eta e^{-\Gamma t}}{r(x_{n,m})}x_{n,m}$. Thus we have:
	\begin{align*}
		\begin{split}
			D(e^{-t(\mathcal{L}^\beta_N)^*}\rho_{n,m}(\eta)| E_{fix}^*\rho_{n,m}(\eta)) \geq \frac{1}{2}||\frac{\eta e^{-\Gamma t}}{r(x_{n,m})}x_{n,m}||_1^2 = \frac{e^{-2\Gamma t}}{2}||(id - E_{fix}^*)\rho_{n,m}(\eta)||_1^2 \pl.\qedhere
		\end{split}
	\end{align*}
\end{proof}
Proposition \ref{proposition:slowdecay} and corollary \ref{corollary:pinsker} show that when the temperature is low enough ($\beta = O(\log N)$), there are plenty meta-stable normalized states. Within a time period of length on the order of $O(N)$, these meta-stable states have only negligible decay under the collective noise generated by the Lindbladian $\mathcal{L}^\beta_N$. Although these states eventually decay to their fixed point states, they are "almost" decoherence free. 
\section{Related Noise Models}\label{section:related}
In this section, we consider other collective noise models. The collective noise model generated by the Lindbladian $\mathcal{L}^\beta_N$ is special because the generators $\pi_N(a), \pi_N(a^*)$ are symmetric under permutation of the qubits. Therefore, this model can only describe symmetric coupling between the system and the environment. As the proof of theorem \ref{theorem:main} has shown, the permutation-invariance is the main reason that the Lindbladian $\mathcal{L}^\beta_N$ has a large fixed point algebra $\Omega_{fix}$. However, perfect permutation invariance may not exist in realistic physical systems. Since $\mathcal{L}^\beta_N$ describes the original Dicke's superradiance model \cite{Dicke}, it is natural for us to consider Lindbladians of generalized Dicke's models. These atomic systems are widely studied in quantum optics \cite{MAG}\cite{SMAG}\cite{ECBT}. When the underlying system is put on a one-dimensional regular lattice, the corresponding Lindbladian has generators $\mathcal{O}_\nu$ of the form:
\begin{equation*}
	\mathcal{O}_\nu = \sum_{1\leq j \leq N}e^{i\theta_j}\pi_{(j)}(a)\pl.
\end{equation*}
where $\theta_j\in [0,2\pi]$ are rotation angles. Each generator $\mathcal{O}_\nu$ describes one collective decay mode of the atomic system. At inverse temperature $\beta$, the corresponding Lindbladian can be written as:
\begin{equation}\label{equation:lattice_lind}
	\mathcal{L}^\beta_\theta:= e^{-\beta / 2}L_{\mathcal{O}_\nu} + e^{\beta / 2}L_{\mathcal{O}_\nu^*}
\end{equation}
where $L_{\mathcal{O}_\nu}x = 2\mathcal{O}_\nu^* x\mathcal{O}_\nu - \mathcal{O}_\nu^*\mathcal{O}_\nu x - x \mathcal{O}_\nu^*\mathcal{O}_\nu$.
We will show that, if an atomic system on a one-dimensional regular lattice has multiple decay modes described by generators of the form $\mathcal{O}_\nu$, then under generic conditions the system decays to a unique equilibrium state. This will be the content of the first subsection.

In the second subsection, we study whether the decay to equilibrium in the generalized Dicke's models satisfies CLSI inequality. We are able to prove the following theorem:
\begin{theorem}\label{theorem:primitive_decay}
	There exists a finite number $m(\beta) = O(1)$ such that for all $N$ there exists a finite set of unitaries $\{U_j\}_{1\leq j\leq m(\beta)}$ such that:
	\begin{enumerate}
		\item Each unitary $U_j$ commutes with the density of the reference state $d_N$;
		\item The following CLSI inequality holds:
		\begin{equation}\label{equation:rot_CLSI}
			D(\rho | E_{\mathbb{C}}^*\rho) \leq C(N)\sum_{1\leq j \leq m}I_{\mathcal{L}^\beta_{U_j}}(\rho)
		\end{equation}
		where the constant $C(N)$ is poly-logarithmic in $N$ and the rotated Lindbladian is given by:
		\begin{equation}\label{equation:rot_lind}
			\mathcal{L}^\beta_{U_j} := e^{-\beta / 2}L_{U_j^*\pi_N(a)U_j} + e^{\beta / 2}L_{U_j^*\pi_N(a^*)U_j}\pl.
		\end{equation}
	\end{enumerate}
\end{theorem}
The Lindbladian in equation \ref{equation:lattice_lind} is an example of a rotated Lindbladian. The unitary associated with the generator $\mathcal{O}_\nu$ is given by:
\begin{equation}\label{equation:rot}
	U_\theta := \otimes_{1\leq i \leq N}\exp(i\frac{\theta_i}{2} \pi_{(i)}(h)) = \exp(i\sum_{1\leq i \leq N}\frac{\theta_i}{2}\pi_{(i)}(h))\pl.
\end{equation}
However, not all unitaries constructed in the proof of theorem \ref{theorem:primitive_decay} is of the form $U(\theta)$. These unitaries come from quantum expanders and are designed to efficiently remove the multiplicity spaces $W_\lambda$. It remains to be seen whether we can use theorem \ref{theorem:primitive_decay} to prove CLSI for more physically relevant models. 

\subsection{Primitivity of Generalized Dicke's Model on One -Dimensional Regular Lattice}
In this subsection, we study generalized Dicke's models on one-dimensional regular lattices. It is well-known that these models can be effectively described by a Lindbladian of the form \cite{MAG}\cite{SMAG}\cite{ECBT}:
\begin{equation}\label{equation:gen_Dicke}
	\mathcal{L}^\beta := \sum_\nu\mathcal{L}^\beta_{\mathcal{O}_\nu}
\end{equation}
where the summation over $\nu$ sums over different decay modes and $\mathcal{L}^\beta_{\mathcal{O}_\nu}$ is given by equation \ref{equation:lattice_lind}. We will show that under generic conditions, the Lindbladian $\mathcal{L}^\beta$ is primitive. 

First we must ensure that the reference state $d_N$ is invariant under all Lindbladians $\mathcal{L}^\beta_{\mathcal{O}_\nu}$:
\begin{lemma}
	The reference state $d_N$ is invariant under $(\mathcal{L}^\beta_{\mathcal{O}_\nu})^*$.
\end{lemma}
\begin{proof}
	For all $j$, we still have:
	\begin{align*}
		&d_N^{it}\pi_{(j)}(a)d_N^{-it} = e^{i\beta t}\pi_{(j)}(a)
		\\& d_N^{it}\pi_{(j)}(a^*)d_N^{-it} = e^{-i\beta t}\pi_{(j)}(a^*)\pl.
	\end{align*}
	Hence we have:
	\begin{align*}
		&d_N^{it}\mathcal{O}_\nu d_N^{-it} = \sum_j e^{i\theta_j} d_N^{it}\pi_{(j)}(a)d_N^{-it}=e^{i\beta t}\mathcal{O}_\nu\\
		&d_N^{it}\mathcal{O}^*_{\nu}d_N^{-it} = \sum_j e^{-i\theta_j}d_N^{it}\pi_{(j)}(a^*)d_N^{-it} = e^{-i\beta t}\mathcal{O}_\nu^*\pl.
	\end{align*}
	Thus by the general theory of Carlen and Maas \cite{CM1}\cite{CM2}, $(\mathcal{L}^\beta_\theta)^*d_N = 0$.
\end{proof}
Now we exploit the special structure of the generators $\mathcal{O}_\nu$ to prove the generic primitivity of the Lindbladian $\mathcal{L}^\beta_N + \mathcal{L}^\beta_{\mathcal{O}_\nu}$:
\begin{lemma}\label{lemma:vandermonde}
	Suppose the parameters $\theta_j$'s are mutually different in the sense that for all $i\neq j$:
	\begin{equation*}
		(\theta_i - \theta_j) \neq 0 \mod 2\pi\pl.
	\end{equation*}
	Then the fixed point algebra of $\mathcal{L}^\beta_N + \mathcal{L}^\beta_{\mathcal{O}_\nu}$ is $\mathbb{C}$.
\end{lemma}
\begin{proof}
	Recall the fixed point algebra of $\mathcal{L}^\beta_N$ is given by:\begin{equation*}
		\Omega_{fix} = \{x\in\mathbb{B}(V^{\otimes N}): [\pi_N(a), x] = [\pi_N(a^*), x] = 0\} \pl.
	\end{equation*}
	And using the unitary $U_\theta$ constructed in equation \ref{equation:rot}, the fixed point algebra of $\mathcal{L}^\beta_{\mathcal{O}_\nu}$ is given by:
	\begin{equation*}
		U^*_\theta \Omega_{fix} U_\theta = \{x\in\mathbb{B}(V^{\otimes N}): [\mathcal{O}_\nu, x] = [\mathcal{O}^*_\nu, x] = 0\}\pl.
	\end{equation*} Thus the fixed point algebra of $\mathcal{L}^\beta_N + \mathcal{L}^\beta_{\mathcal{O}_\nu}$ is given by:
	\begin{equation*}
		\Omega_{fix} \cap U^*_\theta \Omega_{fix} U_\theta = \{x \in\mathbb{B}(V^{\otimes N}): [\pi_N(a), x] = [\mathcal{O}_\nu,x] = 0, [\pi_N(a^*), x] = [\mathcal{O}_\nu^*,x] = 0\}\pl.
	\end{equation*}
	In particular, if $x$ is in this fixed point algebra, then it is annihilated by the generators:
	\begin{equation*}
		[\pi_N(a) + \mathcal{O}_\nu, x] = [\pi_N(a^*) + \mathcal{O}_\nu^*, x] = 0\pl.
	\end{equation*}
	Consider the generator $\pi_z(a) := \pi_N(a) + \mathcal{O}_\nu = \sum_j (1 + e^{i\theta_j})\pi_{(j)}(a)$ where we denote $z_j := 1 + e^{i\theta_j}$. The Lie algebra generated by $\{\pi_z(a), \pi_z(a)^*\}$ contains the following elements:
	\begin{align*}
		&\pi_z(h):=[\pi_z(a), \pi_z(a)^*] = \sum_j |z_j|^2\pi_{(j)}(h)\\
		&\pi^{(1)}_z(a):=\pi_z(h), \pi_z(a)] = 2\sum_j |z_j|^2z_j \pi_{(j)}(a)\\
		&\pi_z^{(1)}(a)^*:=[\pi_z(h), \pi_z(a)^*] = 2\sum_j |z_j|^2z_j^*\pi_{(j)}(a^*)\pl.
	\end{align*}
	More generally, by simple induction we have the following formulae:
	\begin{align*}
		&\pi_z^{(k)}(a):=ad_{\pi_z(h)}^k(\pi_z(a)) = 2^k\sum_j |z_j|^{2k}z_j \pi_{(j)}(a)\\
		&\pi_z^{(k)}(a)^*:=ad_{\pi_z(h)}^k(\pi_z(a)^*) = 2^k\sum_j |z_j|^{2k}z_j^*\pi_{(j)}(a^*)
	\end{align*}
	where $ad_{\pi_z(h)}$ is the adjoint operator: $[\pi_z(h),\cdot]$. 
	
	By simple induction and repeated use of Jacobi identity, $x$ is annihilated by all iterated adjoint operators:  $\{ad_{\pi_z^{(k)}(a)},ad_{\pi_z^{(k)}(a)^*}\}_{k\geq 1}$. In particular, $x$ is annihilated by all linear combinations of these adjoint operators.
	
	Consider the following formulae:
	\begin{align}
		&\begin{pmatrix}
			\pi_z(a)\\\pi_z^{(1)}(a)\\...\\\pi_z^{(N-1)}(a)
		\end{pmatrix} = \begin{pmatrix}
			z_1&z_2&...&z_N\\|z_1|^2z_1&|z_2|^2z_2&...&|z_N|^2z_N\\...\\|z_1|^{2(N-1)}z_1&|z_2|^{2(N-1)}z_2&...&|z_N|^{2(N-1)}z_N
		\end{pmatrix}\begin{pmatrix}
			\pi_{(1)}(a)\\\pi_{(2)}(a)\\...\\\pi_{(N)}(a)
		\end{pmatrix}\\
		&\begin{pmatrix}
			\pi_z(a)^*\\\pi_z^{(1)}(a)^*\\...\\\pi_z^{(N-1)}(a)^*
		\end{pmatrix} = \begin{pmatrix}
			z_1^*&z_2^*&...&z_N^*\\|z_1|^2z_1^*&|z_2|^2z_2^*&...&|z_N|^2z_N^*\\...\\|z_1|^{2(N-1)}z_1^*&|z_2|^{2(N-1)}z_2^*&...&|z_N|^{2(N-1)}z_N^*
		\end{pmatrix}\begin{pmatrix}
			\pi_{(1)}(a^*)\\\pi_{(2)}(a^*)\\...\\\pi_{(N)}(a^*)\end{pmatrix}\pl.
	\end{align}
	By the classical formula of the determinant of the Vandermonde matrices, we have:
	\begin{align*}
		\begin{split}
			&\det\bigg( \begin{pmatrix}
				z_1&z_2&...&z_N\\|z_1|^2z_1&|z_2|^2z_2&...&|z_N|^2z_N\\...\\|z_1|^{2(N-1)}z_1&|z_2|^{2(N-1)}z_2&...&|z_N|^{2(N-1)}z_N
			\end{pmatrix}\bigg) 
			\\&= (\prod_j z_j) \det\bigg( \begin{pmatrix}
				1&1&...&1\\|z_1|^2&|z_2|^2&...&|z_N|^2\\...\\|z_1|^{2(N-1)}&|z_2|^{2(N-1)}&...&|z_N|^{2(N-1)}
			\end{pmatrix}\bigg)
			= (\prod_j z_j)(\prod_{i < j}(|z_i|^2 - |z_j|^2))\pl.
		\end{split}
	\end{align*}
	And similarly we have:
	\begin{align*}
		\det\bigg( \begin{pmatrix}
			z_1^*&z_2^*&...&z_N^*\\|z_1|^2z_1^*&|z_2|^2z_2^*&...&|z_N|^2z_N^*\\...\\|z_1|^{2(N-1)}z_1^*&|z_2|^{2(N-1)}z_2^*&...&|z_N|^{2(N-1)}z_N^*
		\end{pmatrix}\bigg) = (\prod_j z_j^*)(\prod_{i < j}(|z_i|^2 - |z_j|^2))\pl.
	\end{align*}
	Since $|z_i|^2 - |z_j|^2 = 2\cos\theta_i - 2\cos\theta_j$, then by the assumption both determinants are non-zero. Thus linear combinations of $\{ad_{\pi_z^{(k)}(a)},ad_{\pi_z^{(k)}(a)^*}\}_{k\geq 1}$ induce $\{ad_{\pi_{(j)}(a)}, ad_{\pi_{(j)}(a^*)}\}_{1\leq j \leq N}$. Hence for all $j$, we have: $[\pi_{(j)}(a),x] = [\pi_{(j)}(a^*),x] = 0$. Since $\{\pi_{(j)}(a), \pi_{(j)}(a^*)\}$ generate the entire algebra $\mathbb{B}(V^{\otimes N})$, $x$ is in the commutant $\mathbb{B}(V^{\otimes N})' = \mathbb{C}$. Thus the fixed point algebra is the trivial algebra $\mathbb{C}$. The Lindbladian $\mathcal{L}^\beta_N + \mathcal{L}^\beta_{\mathcal{O}_\nu}$ is primitive.
\end{proof}
\begin{corollary}
	In the Lindbladian $\mathcal{L}^\beta$ (c.f. equation \ref{equation:gen_Dicke}), if there exists a pair of summands $\mathcal{L}^\beta_{\mathcal{O}_\nu}$ and $\mathcal{L}^\beta_{\mathcal{O}_\mu}$ whose rotation parameters $\{\theta_j\}_{1\leq j \leq N}$ and $\{\varphi_j\}_{1\leq j \leq N}$ satisfy the condition:
	\begin{equation*}
		(\theta_i - \theta_j) - (\varphi_i - \varphi_j) \neq 0 \mod 2\pi
	\end{equation*}
	for $i\neq j$, then the Lindbladian $\mathcal{L}^\beta$ is primitive.
\end{corollary}
\begin{proof}
	Consider $\mathcal{L}^\beta_{\mathcal{O}_\nu} + \mathcal{L}^\beta_{\mathcal{O}_\mu}$ and let $U_\varphi$ be the unitary matrix that implements the conjugation:\begin{equation*}
		\mathcal{O}_\mu = U_\varphi^* \pi_N(a)U_\varphi\pl.
	\end{equation*}
	Then we have:
	\begin{equation*}
		\mathcal{L}^\beta_{\mathcal{O}_\nu} + \mathcal{L}^\beta_{\mathcal{O}_\mu} = Ad^*_{U_\varphi}\circ(\mathcal{L}^\beta_{\theta-\varphi} + \mathcal{L}^\beta_N)\circ Ad_{U_\varphi}
	\end{equation*}
	where $\mathcal{L}^\beta_{\theta-\varphi}$ is the Lindbladian generated by $U^*_{\theta-\varphi}\pi_N(a)U_{\theta_\varphi}$ and $Ad_{U_\varphi}$ is the adjoint action by the unitary $U_{\varphi}$ on $\mathbb{B}(V^{\otimes N})$. By lemma \ref{lemma:vandermonde}, the Lindbladian $\mathcal{L}^\beta_{\theta - \varphi} + \mathcal{L}^\beta_N$ is primitive. Thus $\mathcal{L}^\beta_{\mathcal{O}_\nu} + \mathcal{L}^\beta_{\mathcal{O}_\mu}$ is primitive as well. Since the fixed point algebra of $\mathcal{L}^\beta$ is a subset of the fixed point algebra of $\mathcal{L}^\beta_{\mathcal{O}_\nu} + \mathcal{L}^\beta_{\mathcal{O}_\mu}$, $\mathcal{L}^\beta$ is primitive.
\end{proof}
\subsection{Efficient Decay to Equilibrium in a Generalized Dicke's Model}
In this subsection, we prove theorem \ref{theorem:primitive_decay}. The main motivation is to study the efficient decay to a unique equilibrium starting with a non-primitive Lindbladian $\mathcal{L}^\beta_N$. Before discussing the main idea of the proof, we fix some notations for the algebras used in the proof. Let $V^{\otimes N} = \bigoplus_\lambda V_\lambda\otimes W_\lambda$ be the Schur-Weyl decomposition of $V^{\otimes N}$, the algebra $\Omega = \bigoplus_\lambda \ell^{n_\lambda}_{\infty}\otimes\mathbb{B}(W_\lambda)$ is the commutant of $C^*(|\pi_N(a)|)\vee C^*(|\pi_N(a^*)|)$ and $\Omega_{fix} = \bigoplus \mathbb{C}\otimes\mathbb{B}(W_\lambda)$ is the fixed point algebra of the original Lindbladian $\mathcal{L}^\beta_N$. Define the algebra $\Omega_{diag} := \bigoplus_\lambda \ell^{n_\lambda}_\infty\otimes \mathbb{C}\ssubset\Omega$. This is the multiplicity-free subalgebra of $\Omega$. Since $\Omega_{diag}$ is invariant under the modular automorphism $d_N^{it}\cdot d_N^{-it}$, there exists a normal conditional expectation $E_{diag}: \mathbb{B}(V^{\otimes N})\rightarrow \Omega_{diag}$. 
\begin{lemma}\label{lemma:diag_fix_commute}
	Restricted on the subalgebra $\Omega$, the two conditional expectations $E_{diag}$ and $E_{fix}$ commute. And the product $E_{diag}E_{fix}$ is a conditional expectation onto the intersection: $\Omega_\oplus = \bigoplus_\lambda \mathbb{C}$.
\end{lemma}
\begin{proof}
	The conditional expectation $E_{diag}$ only acts on the multiplicity subspaces: $E_{diag} = \bigoplus_\lambda id\otimes E^\lambda_{\mathbb{C}}$ where $E^\lambda_{\mathbb{C}}:\mathbb{B}(W_\lambda)\rightarrow \mathbb{C}$ is the conditional expectation onto the center. On the other hand, the conditional expectation $E_{fix} = \bigoplus_\lambda E^\lambda_{fix}\otimes id$ where $E^\lambda_{fix}:\ell^{n_\lambda}_\infty\rightarrow \mathbb{C}$ is the conditional expectation onto the constant. Then it is clear that:
	\begin{equation*}
		[E_{diag}, E_{fix}] = \bigoplus_\lambda [id\otimes E^\lambda_{\mathbb{C}}, E^\lambda_{fix}\otimes id] = 0
	\end{equation*}  

	The product of these two conditional expectations is given by:
	\begin{equation}\label{equation:oplusce}
		E_\oplus := E_{diag}E_{fix} = \bigoplus_\lambda E^\lambda_{fix}\otimes E^\lambda_{\mathbb{C}} = \bigoplus_\lambda E_{\mathbb{C}}\pl.\qedhere
	\end{equation}
\end{proof}
In other words, we have the commuting square:
\begin{center}
	\begin{tikzcd}
		\Omega_{fix} \arrow[hookrightarrow]{r} &\Omega
		\\
		\Omega_{\oplus}\arrow[hookrightarrow]{u}\arrow[hookrightarrow]{r}&\Omega_{diag}\ar[hookrightarrow]{u}
	\end{tikzcd}
\end{center}

Later, we also need to average over the eigen-subspaces of the operator $\pi_N(h)$. We introduce the subalgebra: $M:=\bigoplus_{0\leq j \leq N}\mathbb{B}(H_j)$ where $H_j := \{\xi\in V^{\otimes N}: \pi_N(h)\xi = (N-2j)\xi\}$. The center of $M$ will be denoted as $Z(M) := \bigoplus_j \mathbb{C}$. The maximal Abelian subalgebra of $M$ will be denoted as $\Omega_\sigma:= \bigoplus_j \ell^{h_j}_\infty$ where $h_j :=\dim H_j$. 
\begin{lemma}\label{lemma:inclusion}
	The fixed point algebra $\Omega_{fix}\ssubset M$ and $\Omega_{diag} \ssubset \Omega_\sigma\ssubset M$.
\end{lemma}
\begin{proof}
	The first statement follows from the following decomposition:
	\begin{equation*}
		\Omega_{fix} = \bigoplus_\lambda \mathbb{C}\otimes\mathbb{B}(W_\lambda) = \bigoplus_{0\leq j \leq N}\bigoplus_{n \geq |N-2j|}\mathbb{C}\otimes\mathbb{B}(W_n) \ssubset \bigoplus_j \mathbb{B}(H_j)
	\end{equation*}
	where $W_n$ is the multiplicity space associated with the irreducible representation $V_n$. 
	
	To see the second claim, we use a similar decomposition:
	\begin{equation}
		 \Omega_{diag} = \bigoplus_j\bigoplus_{n\geq |N-2j|}\ell_\infty^{k_j}\otimes \mathbb{C} \ssubset \bigoplus_j \ell^{h_j}_\infty =\Omega_\sigma
	\end{equation}
	where $k_j = |\{0\leq n \leq N: n\geq |N-2j|\}|$. 
\end{proof}
Since $H_j$ is the eigen-subspace of $\pi_N(h)$, the density of the reference state on each $H_j$ is proportional to identity. Therefore there exists a conditional expectation $E_\sigma: M\rightarrow \Omega_\sigma$.

We are now ready to state the idea of the proof. The original Lindbladian $\mathcal{L}^\beta_N$ has a nontrivial fixed point algebra $\Omega_{fix}$ which includes matrix algebras over multiplicity subspaces. To efficiently get rid of $\mathbb{B}(W_\lambda)$'s, we use quantum expanders. More specifically, we need the following fact about quantum expanders:
\begin{theorem}\cite{Has}\cite{BST}\cite{Pisier}\label{theorem:expander}
	On any finite dimensional type I factor $\mathbb{B}(H)$ ($\dim H < \infty$), there exists a constant $c\in (0,1)$ (independent of the dimension) and a finite family of unitaries $\{W_1,...,W_m\}$ (m is independent of the dimension but depends on the constant $c$) such that for all $x\in\mathbb{B}(H)$ we have:
	\begin{equation}
		||\frac{1}{m}\sum_{1\leq i \leq m}W_i^*xW_i - \tau_H(x)||_2 \leq c||x||_2
	\end{equation}
	where $\tau_H(\cdot)$ is the normalized trace on $\mathbb{B}(H)$ and $||\cdot||_2$ is the Hilbert-Schmidt norm on $\mathbb{B}(H)$.
\end{theorem}
For our construction, the set of unitaries need to be closed under conjugation. Therefore for each finite dimensional type I factor $\mathbb{B}(H)$, we fix a quantum expander generated by $\{W_j(H)\}_{1\leq j \leq m}$. This set of unitaries is closed under conjugation. There are several choices of quantum expanders one can use. We will elaborate on the constructions in the Appendix \ref{appendix:expander}. The important point here is that the number of unitaries we need is independent of $\dim H$ and the number is of order $O(1)$.

The quantum expanders in $\mathbb{B}(W_\lambda)$'s help us to implement the conditional expectation: $E_{diag}:\Omega\rightarrow \Omega_{diag}$. However, to control the entropy decay $D(E_{\Omega}^*\rho | E_{diag}^*\rho)$ by entropy production of rotated Lindbladians, on each irreducible component we need to use the following automorphism:
\begin{align}\label{equation:expander_channel}
	\begin{split}
		\Phi_W^\lambda:& \ell^{n_\lambda}_\infty \otimes\mathbb{B}(W_\lambda)\rightarrow \ell^{n_\lambda}_\infty\otimes\mathbb{B}(W_\lambda)
		\\
		& x_j \mapsto W_{j\mod m}^*x_jW_{j\mod m} \text{  ,  if $j \leq n - n\mod m$}\\
		&x_j\mapsto x_j \text{  ,  if $j > n - n\mod m$}\pl.
	\end{split}
\end{align}
Using this automorphism, we will be able to approximate the conditional expectation $E_{diag}$ by a finite iteration of the following channel: $E_{fix}\circ\Phi_W^*\circ E_{fix}\circ\Phi_W\circ E_{fix}$. Then using the standard argument of quasi-factorization, we are able to obtain an upper bound by a finite sum of rotated Lindbladians. 

Unfortunately, this method only works for irreducible component with sufficiently large dimension ($\dim V_\lambda \geq m$). For the remaining irreducible components, we have to apply additional permutation transformations to implement $E_{diag}$. More precisely, consider the direct sum of two irreducible components: $\ell^{n_\lambda}_\infty\otimes\mathbb{B}(W_\lambda)\bigoplus\ell^{n_\eta}_\infty\otimes\mathbb{B}(W_\eta)$, where $n_\lambda \geq m$ and $n_\eta < m$. Hence we can only define automorphism $\Phi_W^\lambda$ on the component $\ell^{n_\lambda}_\infty\otimes\mathbb{B}(W_\lambda)$. By Schur-Weyl duality, we know that $\dim W_\eta > \dim W_\lambda$. Therefore we partition $W_\eta$ into $\dim W_\lambda$-corners and permute each piece to $\mathbb{B}(W_\lambda)$. Once the corner is moved to $\mathbb{B}(W_\lambda)$, we can apply the automorphism $\Phi_W^\lambda$ to obtain an upper bound on the entropy decay by entropy product of rotated Lindbladians.

Combining these ideas, we are able to prove the following:
\begin{proposition}\label{proposition:step1}
	Let $m$ be the number of unitaries used to generate a quantum expander. Then there exists $O(1)$ many unitaries $\{U_j\}$ and a constant $C > 0$ depending only on the number of qubits $N$ such that:
	\begin{equation}
		D(\rho | E^*_{diag}\rho) \leq C\big(\sum_j I_{\mathcal{L}^\beta_{U_j}}(\rho)\big)
	\end{equation}
	where $\mathcal{L}^\beta_{U_j}$ is a rotated Lindbladian generated by $U_j^*\pi_N(a)U_j$.
\end{proposition}
This finishes the first part of the proof of theorem \ref{theorem:primitive_decay}. The second part of the proof deals with the decay from the subalgebra $\Omega_{diag}$ to the fixed point $\mathbb{C}$. Instead of controlling the relative entropy $D(E^*_{diag}\rho | E_{\mathbb{C}}^*\rho)$, we first study the relative entropy $D(E^*\rho | E_\sigma^*\rho)$ where the upper bound by rotatted entropy productions can be easily obtained from general machineries. Then using a particular structure of $\mathfrak{su}(2)$-representations, we are able to relate $D(E^*_{diag}\rho | E_{\mathbb{C}}^*\rho)$ to $D(E^*_{diag}\rho | E_\sigma^*\rho)$ with a constant overhead that only depends on $\beta$. Combining these two parts, we are finally able to prove theorem \ref{theorem:primitive_decay}.

Beforewe dive into the details, there are three facts that we use repeatedly in the proof. The first is an observation about quantum Fourier transform and commutative subalgebras in side finite dimensional type I factors.
\begin{lemma}\label{lemma:quantumFourier}
	Let $\mathbb{B}(H)$ be a finite dimensional type I factor and $N\ssubset \mathbb{B}(H)$ be a commutative subalgebra. Let $\tau$ be the canonical trace on $\mathbb{B}(H)$. In addition, let $\mathcal{F}:H\rightarrow H$ be the quantum Fourier transform:
	\begin{equation*}
		\mathcal{F}(e_\alpha) = \frac{1}{\sqrt{N}}\sum_{1\leq \beta \leq N}e^{-\frac{2\pi i}{N}\alpha\beta}e_\beta
	\end{equation*}
	where $\{e_\alpha\}$ is a fixed orthonormal basis of $H$ and $N = \dim H$. Then we have:
	\begin{equation}
		D(\rho | E_\tau(\rho)) \leq D(\rho | E_N(\rho)) + D(\rho | E_{Ad_\mathcal{F}(N)}(\rho))
	\end{equation}
	where $E_{Ad_\mathcal{F}(N)}$ is the conditional expectation onto $Ad_\mathcal{F}(N)$.
\end{lemma}
\begin{proof}
	Since N is commutative, there exists a maximal Abelian algebra $\ell^n_\infty$ containing $N$ where $n = \dim H$. Let $E_\infty: \mathbb{B}(H)\rightarrow \ell^n_\infty$ be the conditional expectation onto the maximal Abelian subalgebra. Then for all $f = (f_\alpha)_{1\leq \alpha \leq n}\in \ell^n_\infty$, we have:
	\begin{align*}
		\begin{split}
			E_\infty\circ Ad_\mathcal{F}(f) &= E_\infty(\frac{1}{N}\sum_\alpha f_\alpha \sum_{1\leq \beta, \gamma\leq n}exp(-\frac{2\pi i}{N}\alpha(\beta - \gamma))e_{\beta\gamma})
			\\
			&=\frac{1}{N}\sum_\alpha f_\alpha \sum_{1\leq \beta\leq n}e_{\beta\beta} = \tau(f)id
		\end{split}
	\end{align*}
	where $e_{\beta\gamma}$ is the matrix unit associated with the orthonormal basis $\{e_\alpha\}$. From this identity, we have:
	\begin{equation*}
		\Phi^*\Phi:=E_\infty \circ Ad_{\mathcal{F}}^*\circ E_\infty \circ Ad_{\mathcal{F}}\circ E_\infty = E_\infty \circ E^\mathcal{F}_\infty\circ E_\infty = E_\mathbb{C}
	\end{equation*}
	where $\Phi:= E^\mathcal{F}_\infty \circ E_\infty$ and $E^\mathcal{F}_\infty = E_{Ad_\mathcal{F}\ell^n_\infty}$. Then by an iteration procedure \cite{GJLL}\cite{NL}, for all $\rho > 0$ we have:
	\begin{align}
		\begin{split}
			D(\rho | E_\tau\rho)  &= D(\rho | \Phi^*\Phi\rho) \leq H_\tau(\rho) - H_\tau(\Phi\rho)
			\\
			&= \tau(\rho \log \rho) - \tau(E^\mathcal{F}_\infty\rho\log (E^\mathcal{F}_\infty\rho)) + \tau(E^\mathcal{F}_\infty\rho \log(E^\mathcal{F}_\infty\rho)) - \tau(\Phi\rho \log (\Phi \rho))
			\\
			&= D(\rho | E^\mathcal{F}_\infty\rho) + D(E^\mathcal{F}_\infty\rho | \Phi \rho) \leq D(\rho | E^\mathcal{F}_\infty\rho) + D(\rho | E_\infty\rho)
		\end{split}
	\end{align}
	where the last inequality follows from the data processing inequality. Since $N\ssubset \ell^n_\infty$ and $Ad_\mathcal{F}(N)\ssubset Ad_\mathcal{F}\ell^n_\infty$, by the variational principle of relative entropy, we have:
	\begin{align*}
		\begin{split}
			& D(\rho | E_\infty\rho)\leq D(\rho | E_N\rho)\\
			&D(\rho | E^\mathcal{F}_\infty\rho) \leq D(\rho | E_{Ad_\mathcal{F}(N)}\rho)\pl.
		\end{split}
	\end{align*}
	Combining these inequalities, we have the desired result.
\end{proof}
The second fact that we will use repeatedly in the proof of theorem \ref{theorem:primitive_decay} is an iteration procedure \cite{GJLL}\cite{NL}. In fact, we have already used this procedure in the proof of lemma \ref{lemma:quantumFourier}. 
\begin{lemma}[\cite{GJLL}\cite{NL}]\label{lemma:iterationorder}
	Let $\Phi:L_1(N,\tau)\rightarrow L_1(N,\tau)$ be a unital quantum channel, then for any states $\rho, \sigma > 0$ we have:
	\begin{equation}\label{equation:entropydiff}
		D(\rho | \Phi^*\Phi\rho) \leq H_\tau(\rho) - H_\tau(\Phi(\rho)) + D(\rho | \sigma)
	\end{equation}
	where $\tau$ is the trace on $N$ and the (negative) entropy $H_\tau(\rho) :=\tau(\rho \log\rho)$. In addition, if $\Phi^* = \Phi$ on $N$, we have:
	\begin{equation}
		D(\rho | \Phi^{2k}\rho) \leq kD(\rho| \Phi\rho)\pl.
	\end{equation}
	Specifically, for sandwiched conditional expectations: $\Phi := E_2E_1E_2$ where $E_1^* = E_1$ and $E_2^* = E_2$, we have:
	\begin{equation}
		D(\rho | \Phi^{2k}\rho) \leq k(D(\rho | E_1\rho) + D(\rho|E_2\rho))
	\end{equation}
\end{lemma}
\begin{proof}
	The proof is contained in \cite{GJLL} and \cite{NL}. We include it here for completeness. 
	
	For the first claim, we have:
	\begin{align*}
		\begin{split}
			D(\rho | \Phi^*\Phi\rho) &= \tau(\rho \log \rho) - \tau(\rho \log \Phi^*\Phi\sigma) \\&= \tau(\rho \log \rho) - \tau(\Phi\rho \log(\Phi\rho)) + \tau(\Phi\rho \log(\Phi\rho)) - \tau(\rho \log\Phi^*\Phi\sigma)
			\\
			&\leq H_\tau(\rho) - H_\tau(\Phi\rho) + \tau(\Phi\rho \log(\Phi\rho)) - \tau(\Phi\rho\log\Phi\sigma)
			\\&=H_\tau(\rho) - H_\tau(\Phi\rho) + D(\Phi\rho | \Phi\sigma) \leq H_\tau(\rho) - H_\tau(\Phi\rho) + D(\rho | \sigma)
		\end{split}
	\end{align*}
	where we used the operator concavity of logarithm that for any $x > 0$ we have $\Phi^*\log(x) \leq \log\Phi^*(x)$.
	
	If $\Phi^* = \Phi$ on $N$, then using equation \ref{equation:entropydiff} and a simple induction, we have:
	\begin{align*}
		\begin{split}
			D(\rho | \Phi^{2k}\rho)&\leq H_\tau(\rho) - H_\tau(\Phi\rho) + D(\rho | \Phi^{2k-2}\rho) = D(\rho | \Phi\rho) + D(\rho | \Phi^{2k-2}\rho)
			\\
			& \leq kD(\rho | \Phi\rho)
			\pl.
		\end{split}
	\end{align*}
	Lastly, for the sandwiched conditional expectations, we apply equation \ref{equation:entropydiff} on $E_1E_2$:
	\begin{equation*}
		D(\rho | \Phi\rho) = D(\rho | (E_1E_2)^*(E_1E_2)\rho) \leq D(\rho | E_1E_2\rho)\pl.
	\end{equation*}
	Since $E_1^* = E_1$ and $E_2^* = E_2$, we use equation \ref{equation:entropydiff} to $E_1$ and $E_2$ separately:
	\begin{align*}
		\begin{split}
			D(\rho | E_1E_2\rho) &\leq H_\tau(\rho) - H_\tau(E_1\rho) + D(\rho|E_2\rho) = D(\rho|E_1\rho) + D(\rho | E_2\rho)\pl.
		\end{split}
	\end{align*}
	Therefore, we have: $D(\rho | \Phi^{2k}\rho) \leq k(D(\rho|E_1\rho) + D(\rho | E_2\rho))$.
\end{proof}
The last fact that we will use repeatedly in the proof of theorem \ref{theorem:primitive_decay} is an application of Lieb concavity theorem on entropy production. Recall we have shown in lemma \ref{lemma:concavitydouble} and corollary \ref{corollary:omegaFisherDecay} that the entropy production decays under the conditional expectation $E_\Omega$: $I_{\mathcal{L}^\beta_N}(E^*_\Omega\rho)\leq I_{\mathcal{L}^\beta_N}(\rho)$. We now push this result one step further:
\begin{corollary}\label{corollary:concavity}
	Using the same notation as lemma \ref{lemma:concavitydouble}. Given any state $\rho$ on $\mathbb{B}(V^{\otimes N})$, we have:
	\begin{equation}
		I_{\mathcal{L}^\beta_N}(E_{diag}\rho) \leq I_{\mathcal{L}^\beta_N}(\rho)\pl.
	\end{equation}
\end{corollary}
\begin{proof}
	Since the density of the reference state $\rho_N$ is in the subalgebra $\Omega_{diag}$, the claim follows directly from lemma \ref{lemma:concavitydouble}.
\end{proof}
\begin{remark}
	Since the density $d_N = N_\beta \exp(-\frac{\beta}{2}\pi_N(h)) = \bigoplus_\lambda N_\beta \exp(-\frac{\beta}{2}\pi_\lambda(h))\otimes id$, the subalgebra $\Omega_{diag}$ is the smallest subalgebra where we can apply lemma \ref{lemma:concavitydouble}.
\end{remark}
\begin{lemma}\label{lemma:lieb}
	Let $U$ be a unitary that commutes with $\pi_N(h)$, then we have:
	\begin{align}
		\begin{split}
			&I_{\mathcal{L}^\beta_U}(E^U_\Omega\rho) \leq I_{\mathcal{L}^\beta_U}(\rho)
			\\
			&I_{\mathcal{L}^\beta_U}(E^U_{diag}\rho) \leq I_{\mathcal{L}^\beta_U}(\rho)
		\end{split}
	\end{align}
	where $\mathcal{L}^\beta_U = Ad_U^*\mathcal{L}^\beta_N Ad_U$ is the rotated Lindbladian.
\end{lemma}
\begin{proof}
	Since $[U,\pi_N(h)] = 0$, we have $[Ad_U, Ad_{\pi_N(h)}] = 0$ and $Ad_U(d_N) = d_N$ where $d_N$ is the density of the reference state. Hence we have:
	\begin{equation}
		I_{\mathcal{L}^\beta_U}(E^U_\Omega\rho) = I_{\mathcal{L}^\beta_N}(E_\Omega Ad_U\rho) \leq I_{\mathcal{L{^\beta_N}}}(Ad_U\rho) = I_{\mathcal{L}^\beta_U}(\rho)
	\end{equation}
	where the inequality follows from lemma \ref{lemma:concavitydouble}. Similar calculation applies to $E_{diag}$ using corollary \ref{corollary:concavity}.
\end{proof}
\subsection{Proof of Theorem \ref{theorem:primitive_decay}}
\subsubsection{Multiplicity Subspaces and Quantum Expanders}
In this subsection, we prove proposition \ref{proposition:step1}. Throughout the section, we fix a quantum expander of $m$ unitaries closed under adjoint on each finite dimensional Hilbert space. Recall our goal is to reduce $\Omega = \bigoplus_\lambda \ell^{n_\lambda}_\infty\otimes\mathbb{B}(W_\lambda)$ to $\Omega_{diag} =\bigoplus_\lambda \ell^{n_\lambda}_\infty\otimes\mathbb{C}$. We first consider the case where $n_\lambda \geq m$.
\begin{lemma}\label{lemma:expandertransference}
	Let $d = mL + r$ where $0 \leq r < m$ and $L\geq 1$. On $\ell^d_\infty\otimes\mathbb{B}(H)$ consider the automorphism $\Phi^d_W$ defined in equation \ref{equation:expander_channel}. In addition, let $\mu$ be a fixed probability measure on $\ell^d_\infty$ such that $\mu$ is constant $\mu_l / m$ on each interval $[(l-1)m, lm -1]$ for $1\leq l \leq L$. Then we have:
	\begin{enumerate}
		\item Let $\eta:= \mu(\{0,..., Lm-1\}) = \sum_{1\leq l \leq L}\mu_l$. For all $x\in \mathbb{B}(H)$, we have:
		\begin{equation}
			\mathbb{E}_\mu(\Phi^d_W(1\otimes x)) =  \eta\frac{1}{m}\sum_{1\leq j\leq m}W_j^*xW_j + (1-\eta) x\pl.
		\end{equation}
		\item Let $E_\mu^W:= (\Phi^d_W)^* \circ E_\mu \circ \Phi^d_W$ where $E_\mu := \iota\circ\mathbb{E}_\mu$ and $\iota:\mathbb{B}(H)\rightarrow \ell^d_\infty\otimes\mathbb{B}(H):x\mapsto 1\otimes x$ is the inclusion map. If $\eta \geq \frac{1}{2}$, then there exists a constant $C > 0$ such that for all $f:= 1\otimes x\in \ell^d_\infty\otimes\mathbb{B}(H)$:
		\begin{equation}
			D(f | E_\tau(f)) \leq 2C\log\dim(H)D(f | E_\mu E_\mu^W E_\mu f)
		\end{equation}
		where the constant $C$ depends on $\eta$ and the spectral gap of the quantum expander $0 < \theta < 1$. As usual, $E_\tau$ is the conditional expectation $E_\tau:\ell^d_\infty\otimes\mathbb{B}(H)\rightarrow \ell^d_\infty \otimes\mathbb{C}$.
	\end{enumerate}
\end{lemma}
\begin{proof}
	The first claim follows from a direct calculation:
	\begin{align*}
		\begin{split}
			\mathbb{E}_\mu(\Phi^d_W(1\otimes x)) &= \mathbb{E}_\mu\big((W_{j\mod m}^*xW_{j \mod m})_{0\leq j \leq mL - 1}\big) + \mathbb{E}_\mu\big((x)_{mL\leq j < d}\big)\\
			& = (\sum_{1\leq l\leq L}\mu_l)\frac{1}{m}\sum_{1\leq j \leq m}W_j^*xW_j + (1 - \sum_{1\leq l \leq L}\mu_l)x
			\\
			& = \eta\frac{1}{m}\sum_{1\leq j \leq m}W_j^*xW_j + (1 - \eta)x\pl.
		\end{split}
	\end{align*}
	Let $\Psi_W(x):= \frac{1}{m}\sum_{1\leq j \leq m}W_j^*xW_j$. Then for the second claim, we first calculate:
	\begin{align*}
		\begin{split}
			E_\mu E^W_\mu E_\mu (f) & = E_\mu (\Phi^d_W)^*E_\mu \Phi^d_W(1\otimes x) = E_\mu(\Phi^d_W)^*(\eta\Psi_W(x) + (1-\eta)x)
			\\
			&= \eta^2\Psi_W^2(x) + 2\eta(1-\eta)\Psi_W(x) + (1-\eta)^2x
		\end{split}
	\end{align*}
	where we used the fact that the set of unitaries $\{W_j\}$ is closed under conjugation.
	
	Using this formula, we are able to calculate the spectral gap of this composite channel:
	\begin{align}
		\begin{split}
			||E_\mu E^W_\mu E_\mu - E_\tau||_{2\rightarrow 2} &\leq \eta^2||\Psi_W^2 - E_\tau||_{2\rightarrow 2} + 2\eta(1-\eta)||\Psi_W-E_\tau||_{2\rightarrow 2} + (1-\eta)^2||id - E_\tau||_{2\rightarrow 2}\\
			&\leq \eta^2\theta^2 +2\eta(1-\eta)\theta + (1-\eta)^2 = (\eta\theta +1 - \eta)^2 := C(\theta,\eta)\pl.
		\end{split}
	\end{align}
	For $\eta \geq \frac{1}{2}$, the constant $C(\theta, \eta) < 1$. Therefore for $k = O(\frac{\log\dim (H)}{C(\theta, \eta)})$, we have \cite{GJLL}:
	\begin{equation}
		0.9E_\tau \leq_{c.p.} (E_\mu E^W_\mu E_\mu)^k \leq_{c.p.} 1.1 E_\tau\pl.
	\end{equation}
	Hence by the order inequality \cite{GJLL} and lemma \ref{lemma:iterationorder}, we have:
	\begin{equation}
		D(f | E_\tau f) \leq 2D(f | (E_\mu E^W_\mu E_\mu)^kf) \leq 2kD(f | E_\mu E^W_\mu E_\mu f)\pl.
	\end{equation}
	The constant $k$ has the form $C\log\dim(H)$ where $C$ depends only on $\eta$ and the spectral gap $\theta$.
\end{proof}
Lemma \ref{lemma:expandertransference} requires the measure $\mu$ to be piecewise constant. In the application, the measure is given by the reference state $d_N$ and it is not piecewise constant. Therefore we need an additional change of measure:
\begin{corollary}\label{corollary:changeofmeasure}
	Using the same notation as lemma \ref{lemma:expandertransference}. If there exists a constant $C' > 0$ and another probability measure $\mu'$ such that $\mu \leq C'\mu'$ then we have:
	\begin{equation}
		D(f|E_\tau f) \leq 2CC'^2\log\dim(H)(D(f | E_{\mu'}f) + D(f | E^W_{\mu'}f))\pl.
	\end{equation} 
\end{corollary}
\begin{proof}
	By lemma \ref{lemma:iterationorder}, we have:
	\begin{equation*}
		D(f | E_\mu E^W_\mu E_\mu f) \leq D(f | E_\mu f) + D(f | E^W_\mu f)\pl.
	\end{equation*}
	By a change of measure, we have:
	\begin{align*}
		\begin{split}
			&D(f | E_\mu f) \leq C'^2 D(f | E_{\mu'}f)
			\\
			&D(f | E^W_\mu f) = D(\Phi^d_W f | E_\mu \Phi^d_W f) \leq C'^2 D(\Phi^d_W f | E_{\mu'}\Phi^d_W f) = C'^2D(f | E^W_{\mu'}f)\pl.
		\end{split}
	\end{align*}
	Then combine these inequalities, we have the desired result.
\end{proof}
The transference technique in lemma \ref{lemma:expandertransference} does not apply to cases where $d < m$. In our application to $\Omega$, $n_\lambda < m$ happens when the irreducible representation has small dimension. In this case, the corresponding multiplicity subspace $W_\lambda$ has large dimension. It is well-known that the dimension of $W_\lambda$ is given by \cite{Fulton}:
\begin{equation*}
	\dim W_\lambda = \frac{N+1 - 2j}{N+1}\binom{N+1}{j}
\end{equation*}
where the dimension of the irreducible representation is given by $\dim V_\lambda = N+1 - 2j$ for $0\leq j \leq [\frac{N}{2}]$. Hence, when $j \sim [\frac{N}{2}]$, $n_\lambda < m$ and $\dim W_\lambda$ is large. The idea is to partition $W_\lambda$ into corners that have the same dimension as $W_\eta$ where $n_\eta \geq m$. Then by exchanging $W_\eta$ with these corners, we can apply $W_\eta$'s quantum expander channel to reduce these corners to $\mathbb{C}$. A repetition of this procedure would eventually approximate the conditional expectation $E_\tau$ on $W_\lambda$. This observation leads to the following technical lemma:
\begin{lemma}\label{lemma:exchange}
	Let $H_1$ and $H_2$ be two finite dimensional Hilbert spaces such that there exists constants $k > 0$ and $0\leq r < \dim H_2$:
	\begin{equation*}
		\dim H_1 = k\dim H_2+ r\pl. 
	\end{equation*}
	Let $E_b:\mathbb{B}(H_1\oplus H_2) \rightarrow \mathbb{B}(H_1)\oplus \mathbb{B}(H_2)$ be the conditional expectation onto the block-diagonal matrices and let $E_2:\mathbb{B}(H_1\oplus H_2)\rightarrow\mathbb{B}(H_1)\oplus \mathbb{C}: \begin{pmatrix}
		x_{11} & x_{12} \\ x_{21} & x_{22}
	\end{pmatrix}\mapsto \begin{pmatrix}
	x_{11} & 0 \\ 0 & \tau_2(x_{22})
\end{pmatrix}$ be the conditional expectation where $\tau_2$ is the trace on $H_2$. Then there exists $O(k)$-many unitaries $\{U_j\}$ such that for all block-diagonal $\rho \in S_1(H_1)\oplus S_1(H_2)$:
\begin{equation}\label{equation:rotationupper}
	D(\rho | E_\tau\rho) \leq \sum_j D(\rho | E^{U_j}_2\rho)
\end{equation}
where $E^{U_j}_2$ is the rotated conditional expectation $Ad_{U_j}^*\circ E_2 \circ Ad_{U_j}$ and $E_\tau:\mathbb{B}(H_1)\oplus\mathbb{B}(H_2)\rightarrow \mathbb{C}\oplus\mathbb{C}$.  
\end{lemma}
\begin{proof}
	Fix an isomorphism: $H_1\oplus H_2 \cong \ell^{k+1}_2 \otimes H_2 \oplus \ell^r_2$ where $H_1\cong \ell^k_2\otimes H_2 \oplus \ell^r_2$ is identified with the first $k$-copies of $H_2$. Let $S_{k+1}$ be the permutation group acting on $\ell^{k+1}_2$ and let $\tau_j := (j,k+1)$ be the transposition that exchanges $j$ and $k+1$ where $1\leq j \leq k$. Then for all $\rho:=(\rho_1, \rho_2) \in S_1(H_1)\oplus S_1(H_2)$ we have:
	\begin{equation}\label{equation:partition}
		E_j(\rho):= E_2 Ad_{\tau_j}^* E_2 Ad_{\tau_j} E_2 (\rho) = E_2E^{\tau_j}_2 E_2(\rho) = (\rho^j_1, \tau_2(\rho_2))
	\end{equation}
	where in block-diagonal form $\rho_1 = \begin{pmatrix}
		\rho_{uu} & \rho_{uj} & \rho_{ud} \\ \rho_{ju} & \rho_{jj} & \rho_{jd} \\ \rho_{du} & \rho_{dj} & \rho_{dd}
	\end{pmatrix}$ with $\rho_{uu}\in S_1(\ell^{j-1}_2\otimes H_2), \rho_{jj}\in S_1(H_2), \rho_{dd}\in S_1(\ell^{k - j}_2\otimes H_2 \oplus \ell^r_{2})$ and $\rho^j_1  = \begin{pmatrix}
	\rho_{uu} & 0 & \rho_{ud} \\ 0 & \tau_2(\rho_{jj}) & 0 \\ \rho_{du} & 0 & \rho_{dd}
\end{pmatrix}$. By lemma \ref{lemma:iterationorder}, we have:
\begin{equation}
	D(\rho | \prod_{1\leq j \leq k}E_j \rho) \leq \sum_{1\leq j \leq k}D(\rho | E_j\rho)\pl.
\end{equation}
By the formula \ref{equation:partition}, all the conditional expectations $\{E_j\}_{1\leq j \leq k}$ commute with each other. Hence the product is also a conditional expectation:
\begin{equation*}
	\widetilde{E}:= \prod_{1\leq j \leq k}E_j: \mathbb{B}(H_1)\oplus \mathbb{B}(H_2)\rightarrow \ell^k_\infty\oplus\mathbb{B}(\ell^r_2)\oplus\ell^{\dim H_2}_\infty\pl.
\end{equation*}
To take care of the remainder part $\mathbb{B}(\ell^r_2)$, we need an additional unitary on $\ell^r_2\oplus H_2$. Fix an isomorphism: $\ell^r_2\oplus H_2\cong \ell^r_2\oplus \ell^{\dim H_2 - r}_2\oplus\ell^r_2$. Then let $\tau$ be the unitary that exchanges the first component $\ell^r_2$ with the last component $\ell^r_2$. Then we have:
\begin{equation}
	E =\widetilde{E} \widetilde{E^\tau} \widetilde{E}:\mathbb{B}(H_1)\oplus\mathbb{B}(H_2)\rightarrow \ell^{k+1}_\infty\oplus \mathbb{C}\cong \{\sum_{1\leq j \leq k}\lambda_j\ket{j}\bra{j}\otimes id_{H_2} + \lambda_r id_{\ell^r_2} + \lambda_2 id_{H_2}\}
\end{equation}
where $\widetilde{E^\tau} = Ad^*_\tau \widetilde{E}Ad_\tau$. Since the subalgebra $\ell^{k+1}_\infty\ssubset\mathbb{B}(H_1)$ is a commutative subalgebra in a finite dimensional type I factor, we can apply lemma \ref{lemma:quantumFourier} and obtain the following formula:
\begin{equation}\label{equation:finalce}
	E_\tau = E  E^\mathcal{F} E: \mathbb{B}(H_1)\oplus\mathbb{B}(H_2)\rightarrow \mathbb{C}\oplus\mathbb{C}
\end{equation}
where $\mathcal{F}$ is the quantum Fourier transform on $H_1$. Then by repeated use of lemma \ref{lemma:iterationorder}, for all $\rho\in S_1(H_1)\oplus S_1(H_2)$ we have:
\begin{align}
	\begin{split}
		D(\rho | E_\tau\rho) &\leq D(\rho | E\rho) + D(\rho | E^\mathcal{F}\rho) \\&\leq D(\rho | \widetilde{E}\rho) + D(\rho | \widetilde{E^\tau}\rho) + D(Ad_\mathcal{F}\rho | \widetilde{E}Ad_\mathcal{F}\rho) + D(Ad_\mathcal{F}\rho | \widetilde{E^\tau}Ad_\mathcal{F}\rho)
		\\
		&\leq \sum_{1\leq j \leq k}\big(D(\rho | E_j\rho) + D(Ad_\tau\rho | E_jAd_\tau\rho) + D(Ad_\mathcal{F}\rho | E_jAd_\mathcal{F}\rho) + D(Ad_{\tau\mathcal{F}}\rho| E_jAd_{\tau\mathcal{F}}\rho)\big)
	\end{split}
\end{align}
Since $D(\rho | E_j\rho) \leq D(\rho | E_2\rho) + D(\rho | E^{\tau_j}_2\rho)$, for any unitary $Ad_U$ we have:
\begin{equation*}
	D(Ad_U\rho | E_j Ad_U\rho) \leq D(Ad_U\rho | E_2Ad_U\rho) + D(Ad_U\rho | E^{\tau_j}_2Ad_U\rho) = D(\rho | E^U_2\rho) + D(\rho|E^{\tau_j\circ U}_2\rho)
\end{equation*}
where $\tau_j\circ U$ is the composition of the two unitaries. Combining these inequalities, we have:
\begin{align}
	\begin{split}
		D(\rho | E_\tau\rho)\leq  \sum_{1\leq j \leq k}&\big(D(\rho | E_2\rho) + D(\rho | E_2^{\tau_j}\rho) + D(\rho | E^\tau_2\rho) + D(\rho | E_2^{\tau_j\tau}\rho)\\
		&+D(\rho | E^{\mathcal{F}}_2\rho) + D(\rho | E^{\tau_j\mathcal{F}}_2\rho) + D(\rho | E^{\tau\mathcal{F}}_2\rho) + D(\rho | E^{\tau_j\tau\mathcal{F}}_2\rho)\big)\pl.
	\end{split}
\end{align}
Therefore we need $8k$ unitaries to implement the upper bound in equation \ref{equation:rotationupper}.
\end{proof}
\begin{remark}
	The constant $8k$ in the previous lemma is the only place in the proof of theorem \ref{theorem:primitive_decay} where the number of unitaries depends on the dimension. 
\end{remark}
We are now ready to prove proposition \ref{proposition:step1}.
\begin{proof}[Proof of Proposition \ref{proposition:step1}]
Using the same notation as in lemma \ref{lemma:exchange}\ref{lemma:expandertransference}, we first decompose the relative entropy:
\begin{equation*}
	D(\rho | E_{diag}^*\rho)  = D(\rho | E_\Omega^*\rho) + D(E_\Omega^*\rho | E_{diag}^*\rho)\pl.
\end{equation*}
To apply lemma \ref{lemma:expandertransference} and corollary \ref{corollary:changeofmeasure}, we need to study the measure induced by the reference state $d_N$ on the commutative components in $\Omega = \bigoplus_\lambda \ell^{n_\lambda}_\infty\otimes\mathbb{B}(W_\lambda)$. On $\ell^{n_\lambda}_\infty$, the reference state $d_N$ induces a measure: 
\begin{equation*}
	\mu^\lambda_j := C(\lambda)\exp(\frac{\beta}{2}(n_\lambda - 2j))
\end{equation*}
where $0\leq j \leq n_\lambda$ and $C(\lambda)$ is the normalization constant depending only on $\lambda$. Here the labeling is in reverse order as the one we have been using in representation theory. More precisely, here $j = 0$ corresponds to the lowest weight vector in each irreducible component. For $n_\lambda = mL + r \geq m$ ($0 \leq r < m$), we need a piecewise constant probability measure that is close to $\mu^\lambda$. To achieve this, we consider the following measure $\nu$:
\begin{align*}
		&\nu^\lambda_j = \mu^\lambda_{m([\frac{j}{m}]+1)}\text{  ,  for $0\leq j<mL$}\\
		&\nu^\lambda_j = \mu^\lambda_j \text{  ,  for $j \geq mL$}\pl.
\end{align*}
After normalization, the measure $\nu$ differs from $\mu$ by a constant:
\begin{align*}
	\begin{split}
		\exp(-\frac{\beta m}{2})\nu \leq \mu \leq \exp(\frac{\beta m}{2})\nu
	\end{split}
\end{align*}
The bound on the Radon-Nikodym derivative only depends on $\beta$ and $m$. 

In addition, we need to check that the measure $\nu$ on the subspace $\{0,...,mL-1\}$ is larger than $\frac{1}{2}$. To see this,  we observe that for $L\geq 1$ and $r < m$:
\begin{align*}
	\sum_{0\leq j \leq mL - 1}\exp(\frac{\beta}{2}(n_\lambda - 2j)) \geq \sum_{j\geq mL}\exp(\frac{\beta}{2}(n_\lambda - 2j))\pl.
\end{align*}
Since $\nu^\lambda_j \geq \mu^\lambda_j$, then the weigth of $\nu$ on $\{0,...,mL-1\}$ is larger than $\frac{1}{2}$. 

Now on the direct sum of the components with $n_\lambda \geq m$, we apply the direct sum channel:
\begin{equation*}
	E_> := \bigoplus_{n_\lambda\geq m} E_{\nu^\lambda}E^W_{\nu^\lambda}E_{\nu^\lambda}:\Omega\rightarrow \big(\bigoplus_{n_\lambda\geq m}\ell^{n_\lambda}_\infty\otimes\mathbb{C}\big)\bigoplus\big(\bigoplus_{n_\lambda < m} \ell^{n_\lambda}_\infty \otimes\mathbb{B}(W_\lambda)\big)\pl.
\end{equation*}
Since $E_\mu|_\Omega = E_{fix}$ and $\Phi_W\Omega = \Omega$, by lemma \ref{lemma:expandertransference} and corollary \ref{corollary:changeofmeasure}, we can control the relative entropy decay by a sum of $m$ rotated entropy production terms:
\begin{align}
	\begin{split}
		D(E_\Omega^*\rho | E_> \rho)  &\leq C\log(N)\big(D(E_\Omega^*\rho | E_{fix}^*\rho) + D(\Phi_W E_\Omega^*\rho | E_{fix}^*\Phi_W E^*_\Omega\rho)\big)
		\\
		& = C\log(N)\big(D(E_\Omega^*\rho | E^*_{fix}\rho) + D(E^*_{\Phi_W\Omega}\Phi_WE^*_\Omega\rho | E^*_{fix}E^*_{\Phi_W\Omega}\Phi_WE_\Omega\rho)\big)
		\\&\leq C\log(N)\big(I_{\mathcal{L}^\beta_N}(E^*_\Omega\rho) + I_{\mathcal{L}^\beta_N}(E^*_{\Phi_W\Omega}\Phi_WE^*_\Omega\rho) \big)
	\end{split}
\end{align}
where $\Phi_W := \bigoplus_{n_\lambda \geq m}\Phi^\lambda_W$. In the second line, we have used the identity: $\Phi_WE^*_\Omega = \Phi_WE^*_\Omega E^*_\Omega = E^*_{\Phi_W\Omega}\Phi_WE^*_\Omega$. Since $\Phi_W\Omega = \Omega$, we have:
\begin{equation*}
	I_{\mathcal{L}^\beta_N}(E^*_{\Phi_W\Omega}\Phi_WE^*_\Omega\rho) = I_{\mathcal{L}^\beta_N}(E^*_\Omega\Phi_WE^*_\Omega\rho) \leq I_{\mathcal{L}^\beta_N}(\Phi_WE^*_\Omega\rho)
\end{equation*} In addition, since the automorphism $\Phi_W$ preserves the trace and preserves the reference state, we have:
\begin{equation*}
	I_{\mathcal{L}^\beta_N}(\Phi_WE^*_\Omega\rho) = I_{\mathcal{L}^\beta_N}(E^*_\Omega\rho) \leq I_{\mathcal{L}^\beta_N}(\rho)\pl.
\end{equation*}
Combining these inequalities, we have:
\begin{equation}\label{equation:large}
	D(E^*_\Omega\rho | E_>\rho) \leq C\log(N)I_{\mathcal{L}^\beta_N}(\rho)
\end{equation}
Here we used the fact that $\max_\lambda\dim(V_\lambda) = N+1$. Hence the dimension factor in corollary \ref{corollary:changeofmeasure} is given by $\log(N)$.

For the remaining components, $n_\lambda < m$. Let $\eta$ be the component such that $n_\eta = \min\{n_\lambda : n_\lambda\geq m\}$. Fix an embedding: $\iota_\lambda:\ell^{n_\lambda}_\infty\rightarrow \ell^{n_\eta}_\infty$ ($n_\eta \geq m \geq n_\lambda$). On the subalgebra: $\ell^{n_\lambda}\otimes\mathbb{B}(W_\lambda)\bigoplus \iota_\lambda(\ell^{n_\lambda}_\infty)\otimes\mathbb{B}(W_\eta)\cong\ell^{n_\lambda}_\infty\otimes (\mathbb{B}(W_\lambda)\oplus\mathbb{B}(W_\eta))$, we apply the channel: $id\otimes E_\tau$ where $E^\lambda_\tau:\mathbb{B}(W_\lambda)\oplus\mathbb{B}(W_\eta)\rightarrow \mathbb{C}\oplus\mathbb{C}$ is the channel constructed in lemma \ref{lemma:exchange}. By construction, for different components $\lambda$ and $\lambda'$ the corresponding conditional expectations commute:
\begin{equation*}
	[E^\lambda_\tau , E^{\lambda'}_\tau] = 0\pl.
\end{equation*}

In lemma \ref{lemma:exchange} we need to use the conditional expectation $E_2$. In our current application, $E_2$ can be implemented by $E_>$. Combining these observations, we consider the product conditional expectations:
\begin{equation*}
	E_< := \prod_{n_\lambda < m}E^\lambda_\tau
\end{equation*}
($E_<$ is self-adjoint). By lemma \ref{lemma:exchange}, for each $n_\lambda < m$ there exists a finite family of unitaries $\{U_j\}$ such that $Ad_{U_j}\Omega = \Omega$. Using this and lemma \ref{lemma:lieb} we have:
\begin{align}
	\begin{split}
		D(E^*_\Omega\rho | E_<\rho) &\leq \sum_{n_\lambda < m}D(E^*_\Omega\rho | E^\lambda_\tau\rho) \leq \sum_{n_\lambda < m}\sum_{1\leq j \leq K}D(E^*_\Omega\rho | E^{U_j}_>\rho)\\
		&\leq C\log(N)\sum_{n_\lambda < m}\sum_{1\leq j \leq K}I_{\mathcal{L}^\beta_{U_j}}(E^*_\Omega\rho)
		\\
		&\leq C\log(N)\sum_{n_\lambda < m}\sum_{1\leq j \leq K}I_{\mathcal{L}^\beta_{U_j}}(\rho)
	\end{split}
\end{align}
where the second line follow from the first part of the proof.

The number of unitaries $K$ can be estimated from the formula of multiplicity \cite{Fulton}:
\begin{equation*}
	K = 8\times O(\frac{\frac{N+1-2j}{N+1}\binom{N+1}{j}}{\frac{N+1-2k}{N+1}\binom{N+1}{k}})
\end{equation*}
where $k \sim [\frac{N}{2}] - m$ and $k < j \leq [\frac{N}{2}]$. Since $m = O(1)$ (we can take it to be $4$), the ratio $K$ is of the order $O(1)$. 

Combining all the inequalities and using the fact that $E_>E_< = E_<E_> = E_{diag}$, we have:
\begin{align}
	\begin{split}
		D(\rho| E_{diag}^*\rho) &\leq D(\rho | E_{\Omega}^*\rho) + D(E^*_\Omega\rho | E_> \rho) + D(E^*_\Omega\rho | E_<\rho) \\&\leq C\log(N)\sum_j I_{\mathcal{L}^\beta_{U_j}}(\rho)
	\end{split} 
\end{align}
where the number of unitaries need is of the order $O(1)$. 
\end{proof}
\subsubsection{Spectral Subspaces and Quantum Fourier Transform}
In this subsection, we use spectral decomposition of the reference state to prove the following proposition:
\begin{proposition}\label{proposition:step2}
	Let $\rho \in L_1(\Omega_{diag})$ be any state on the commutative algebra $\Omega_{diag}$. Then for $\beta \geq O(\log N)$, there exists a constant $C > 0$ such that:
	\begin{equation}
		D(\rho | E_\mathbb{C}\rho) \leq C k (D(\rho | E^\mathcal{U}_{diag}\rho) + D(\rho| E_{fix}\rho))
	\end{equation}
	where $k$ is a constant of poly-logarithmic order in $N$ and $U$ is a unitary in the spectral subalgebra $M = \bigoplus_{0\leq j \leq N}\mathbb{B}(H_j)$ (c.f. lemma \ref{lemma:inclusion}).
\end{proposition}
It turns out that the unitary $U$ is given by a direct sum of quantum Fourier transforms. To prove proposition \ref{proposition:step2}, we first make the following calculation:
\begin{lemma}\label{lemma:miraculouscalc}
	There exists a unitary $U:= \bigoplus_wU_w \in \mathcal{U}(M)$ such that on the subalgebra $\Omega_{diag}$, we have:
	\begin{equation}
		E_{diag}\circ Ad_{U}= E_\sigma
	\end{equation}
	where $E_\sigma$ is the conditional expectation onto the subalgebra $\Omega_\sigma$.
\end{lemma}
\begin{proof}
	The unitary will be a direct sum of quantum Fourier transforms. Let $w := n-2j$ be the weight of the vector $\ket{n,j, \alpha}$. Since the eigen-subspace $H_w$ has a canonical basis containing vectors of the form $\ket{m,k,\beta}$ where $w = m-2k$ and $\beta$ is the multiplicity label, we fix an ordering of this basis. Under this ordering, $\ket{m,k,\beta}$ will be the $\omega(k,\beta)$-th vector in the ordered basis. Then we can define the quantum Fourier transform:
	\begin{equation*}
		\mathcal{F}_w\ket{n,j,\alpha} = \frac{1}{\sqrt{h_w}}\sum_{k,\beta}e^{-\frac{2\pi i}{h_w}\omega(j,\alpha)\omega(k,\beta)} \ket{m,k,\beta}\pl.
	\end{equation*}
	The unitary $U$ is a direct sum of these quantum Fourier transforms.
	
	By linearity, we only need to show the formula holds for $x = \ket{n,j}\bra{n,j}\otimes\sum_{\alpha}\ket{\alpha}\bra{\alpha}$ where $\alpha$ is the multiplicity label.
	\begin{equation}
		Ad_Ux = \frac{1}{h_w}\sum_{\alpha,k_1,\beta_1,k_2,\beta_2}e^{-\frac{2\pi i}{h_w}\omega(j,\alpha)(\omega(k_1,\beta_1) - \omega(k_2,\beta_2))}\ket{m_1,k_1,\beta_1}\bra{m_2,k_2,\beta_2}
	\end{equation} 
	where $w = n-2j$ and $m_i = w + 2k_i$ ($i = 1,2$). Applying the conditional expectation $E_{diag}$, we have:
	\begin{align}
		\begin{split}
			E_{diag}\circ Ad_Ux &= \frac{1}{h_w}\sum e^{-\frac{2\pi i}{h_w}\omega(j,\alpha)(\omega(k_1,\beta_1) - \omega(k_2,\beta_2))}E_{diag}(\ket{m_1,k_1,\beta_1}\bra{m_2,k_2,\beta_2})\\
			&=\frac{\dim W_{\eta(n)}}{h_w}\sum_{k, m = w + 2k}\big(\ket{m,k}\bra{m,k}\otimes\sum_{\beta\in W_\eta(m)}\ket{\beta}\bra{\beta}\big)
		\end{split}
	\end{align}
	where $W_{\eta(n)}$ is the multiplicity subspace of the unique $(n+1)$-dimensional irreducible representation. The numerator $\dim W_{\eta(n)}$ comes from summation over the multiplicity label $\alpha$.
	
	On the other hand, since $x\in \mathbb{B}(H_w)$, we have:
	\begin{align}
		\begin{split}
			E_\sigma x &= tr_w(x)\sum_{k, m = w+2k}\ket{m,k}\bra{m,k}\otimes\sum_{\beta\in W_{\eta(m)}}\ket{\beta}\bra{\beta} \\&= \frac{\dim W_{\eta(n)}}{h_w}\sum_{k, m = w+2k}\ket{m,k}\bra{m,k}\otimes\sum_{\beta\in W_{\eta(m)}}\ket{\beta}\bra{\beta}
		\end{split}
	\end{align}
	where $tr_w$ is the normalized trace on the spectral subspace $H_w$. 
\end{proof}
\begin{corollary}\label{corollary:Fourier}
	Using the same notation as lemma \ref{lemma:miraculouscalc}, for any state $\rho$ on $V^{\otimes N}$ we have:
	\begin{equation}
		D(E_{diag}\rho | E_\sigma\rho) \leq D(\rho | E_{diag}\rho)  + D(\rho | E^\mathcal{F}_{diag}\rho)
	\end{equation}
	where $\mathcal{F}$ is the direct sum of quantum Fourier transforms constructed in lemma \ref{lemma:miraculouscalc}. And $E^\mathcal{F}_{diag} = Ad^*_\mathcal{F}\circ E_{diag}\circ Ad_{\mathcal{F}}$ is the rotated conditional expectation.
\end{corollary}
\begin{proof}
	 Using the spectral decomposition of $\pi_N(h)$, the conditional expectaion $E_{\sigma}$ can be written as:
	\begin{equation}\label{equation:decompositionofce}
		E_{\sigma} = \sum_w \mu_w E_{w}
	\end{equation}
	where $\mu_w$ is the weight of the spectral subspace $H_w$ under the reference state $\rho_N$. And $E_w$ is the tracial conditional expectation from $\mathbb{B}(H_w)$ to its center. Since $E_w^* = E_w$, $E_\sigma$ is self-adjoint. Thus we have: $E_\sigma = E_\sigma E_\sigma = E_\sigma^*E_\sigma$ where the first equation comes from the definition of conditional expectation. By lemma \ref{lemma:miraculouscalc}, we have $E_\sigma = E_{diag}\circ Ad_\mathcal{F}\circ E_{diag}$. Thus we have:
	\begin{equation}
		E_\sigma = E_{diag}\circ Ad_\mathcal{F}^*\circ E_{diag}\circ Ad_{\mathcal{F}}\circ E_{diag} = E_{diag}E^\mathcal{F}_{diag}E_{diag}\pl.
	\end{equation}
	By data processing inequality and lemma \ref{lemma:iterationorder}, we have:
	\begin{align}
		\begin{split}
			D(E_{diag}\rho | E_\sigma \rho) &= D(E_{diag}\rho | E_{diag}E^\mathcal{F}_{diag}E_{diag}\rho) 
			\\
			&\leq D(\rho | E^\mathcal{F}_{diag}E_{diag}\rho) \\
			&\leq D(\rho | E^\mathcal{F}_{diag}\rho)  + D(\rho | E_{diag}\rho)\pl.\qedhere
		\end{split}
	\end{align}
\end{proof}

\begin{proposition}\label{proposition:saloffcoste}
	Recall $E_\oplus$ (c.f. Equation \ref{equation:oplusce}) is the conditional expectation onto $\Omega_{fix} = E_{fix}(\Omega_{diag}) = \bigoplus_\lambda \mathbb{C}$. Consider the self-adjoint channel: $\Phi:= E_{\oplus} E^\mathcal{F}_{diag}E_{\oplus}: \Omega_\oplus\rightarrow \Omega_\oplus$. Then we have:
	\begin{equation}
		0.9 E_{fix} \leq \Phi^{2k} \leq 1.1E_{fix}
	\end{equation}
	where $k$ is a constant of poly-logarithmic order in $N$. 
\end{proposition}
The proof of Proposition \ref{proposition:saloffcoste} is given in the appendix. Combining these results, we are now ready to prove proposition \ref{proposition:step2}:
\begin{proof}[Proof of Proposition \ref{proposition:step2}]
	By lemma \ref{lemma:miraculouscalc}, the channel $\Phi$ can be written as: $\Phi = E_\oplus E^\mathcal{F}_{diag}E_\oplus = E_{fix}E_{diag}E^\mathcal{F}_{diag}E_{diag}E_{fix} = E_{fix}E_\sigma E_{fix}$. By Proposition \ref{proposition:saloffcoste} and Lemma \ref{lemma:iterationorder}, for state $\rho$ on $\Omega_{diag}$ we have:
	\begin{equation*}
		D(\rho | E_\mathbb{C}\rho) \leq 2Ck D(\rho | \Phi\rho)\leq 4Ck(D(\rho | E_{fix}\rho) + D(\rho | E^\mathcal{F}_{diag}\rho))\pl.\qedhere
	\end{equation*}
\end{proof}
We now present the proof of theorem \ref{theorem:primitive_decay}.
\begin{proof}[Proof of Theorem \ref{theorem:primitive_decay}]
	By proposition \ref{proposition:step1}, there exists $O(1)$-many unitaries $\{U_j\}$ that commutes with the reference state such that:\begin{equation*}
		D(\rho | E^*_{diag}\rho) \leq C\log(N)\sum_j I_{\mathcal{L}^\beta_{U_j}}(\rho)\pl.
	\end{equation*}
	By the chain rule of relative entropy: $D(\rho | E_{\mathbb{C}}^*\rho) = D(\rho | E_{diag}^*\rho) + D(E^*_{diag}\rho | E^*_\mathbb{C}\rho)$. By proposition \ref{proposition:step2}, we have:
	\begin{align}
		\begin{split}
			D(\rho | E_\mathbb{C}^*\rho)& = D(\rho | E_{diag}^*\rho)  + D(E^*_{diag}\rho | E_\mathbb{C}^*\rho)\\
			&\leq C\log(N)\sum_j I_{\mathcal{L}^\beta_{U_j}}(\rho) + C k (D(\rho | E^\mathcal{F}_{diag}\rho) + D(\rho | E_{fix}\rho))
			\\
			& \leq C\log(N)\sum_j I_{\mathcal{L}^\beta_{U_j}}(\rho) + Ck I_{\mathcal{L}^\beta_\mathcal{F}}(\rho) + Ck I_{\mathcal{L}^\beta_N}(E_{diag}\rho) \\
			&\leq C\log(N)\sum_j I_{\mathcal{L}^\beta_{U_j}}(\rho) + Ck I_{\mathcal{L}^\beta_\mathcal{F}}(\rho) + Ck I_{\mathcal{L}^\beta_N}(\rho)\pl.\qedhere
		\end{split}
	\end{align}
\end{proof}

\section{Clustering and Spectral Gap}\label{section:Davies}

As mentioned in the introduction, Corollary \ref{corollary:dimdep} emphasizes the difference between primitive quantum Markov semigroups and non-primitive ones \cite{KB}\cite{BCR}. In this section, we elaborate on this apparent contradiction. Our main example has the form of a Davies generator. Recall a Davies generator associated with a Hamiltonian $\mathcal{H}$ is defined as \cite{D83}:
\begin{definition}
	Consider an open system $S$ coupled with a heat bath $B$, and the total system Hamiltonian is given by $\mathcal{H}_S + \mathcal{H}_B + \sum_\alpha S_\alpha\otimes B_\alpha$ where the interaction term is written generically as a sum of tensor product coupling. Assume the heat bath is in thermal equilibrium and assume the system-bath coupling is weak, then by taking a Born-Markov approximation, the effective dissipative dynamics on system $S$ is given by the Davies generator of the form:\begin{equation*}
		\mathcal{L}^Dx = \sum_{\omega, \alpha} \chi_{\omega, \alpha}(2S_{\omega, \alpha}^* x S_{\omega, \alpha} - S_{\omega, \alpha}^*S_{\omega, \alpha}x - xS_{\omega, \alpha}^*S_{\omega, \alpha})
	\end{equation*}
	where $\omega$'s are the Bohr frequencies, $\chi_{\omega,\alpha}$'s come from the Fourier coefficients of the two-point correlation function of the bath, and $S_{\omega, \alpha}$'s are the Fourier modes of the coupling operator $S_\alpha$:
	\begin{equation*}
		e^{-it\mathcal{H}_S}S_\alpha e^{it\mathcal{H}_S} = \sum_{\omega} S_{\omega, \alpha} e^{it\omega}
	\end{equation*}
\end{definition}
In particular, let $\rho_S$ be the Gibbs state associated with $\mathcal{H}_S$, then the generators $S_{\omega, \alpha}$ satisfy:
\begin{equation*}
	\rho_S^{it} S_{\omega, \alpha}\rho_S^{-it} = e^{it\beta\omega}S_{\omega, \alpha}
\end{equation*}
where $\beta$ is the inverse temperature of the heat bath. And by the KMS condition, coefficients $\chi_{\omega, \alpha}$ satisfy:
\begin{equation*}
	\chi_{-\omega, \alpha} = e^{-\beta\omega}\chi_{\omega,\alpha}
\end{equation*}
\begin{lemma}
	$\mathcal{L}^\beta_N$ is a Davies generator if we take $\pi_N(h)$ to be the system Hamiltonian.
\end{lemma}
\begin{proof}
	Since $\pi_N(h)$ is taken to be the system Hamiltonian, the corresponding Gibbs state is exactly $d_N$ with inverse temperature $\beta$. In addition, since $d_N^{it}\pi_N(a)d_N^{-it} = e^{i\beta t}\pi_N(a)$, $\pi_N(a)$ is the interaction operator corresponding to Bohr frequency $\omega = 1$. Therefore $\mathcal{L}^\beta_N$ is a Davies generator.
\end{proof}
Recall the definitions of minimal conditional expectation and local projection in \cite{KB}\cite{BCR}. Here we directly apply the definitions to the Gibbs state $d_N$.
\begin{definition}
	Let $A\ssubset [|N|]$ and fix the Gibbs state $d_N$ as before. The minimal conditional expectation of $d_N$ on $A$ is given by:
	\begin{equation*}
		\mathcal{E}_A(x) := \tau_A(d_{A}^{1/2}xd_{A}^{1/2})
	\end{equation*}
	where $\tau_A := \bigotimes_{i\in A}\tau \otimes \bigotimes_{i\notin A} id$ is a composition of partial trace on $V^{\otimes A}\ssubset V^{\otimes N}$ and tensoring with the identity matrix, and $d_{A} := \bigotimes_{i\in A} d_\beta\otimes \bigotimes_{i\notin A}id$. 
	
	The local projection of the Davies generator $\mathcal{L}^\beta_N$ is given by
	\begin{equation}
		E^\mathcal{L}_A := \lim_{k\rightarrow\infty}\mathcal{E}^k_A
	\end{equation}
\end{definition}
Compared with the original definitions, the definitions here are simplified because $d_N$ is a factor state.
\begin{lemma}
	Let $A,B\ssubset[|N|]$ be two subsets with nontrivial overlap $A\cap B\neq \emptyset$. Then\begin{equation}
		\mathcal{E}_A\circ\mathcal{E}_B = \mathcal{E}_B\circ\mathcal{E}_A = \mathcal{E}_{A\cup B}
	\end{equation}
	In addition, the local projections also commute:
	\begin{equation}
		E^\mathcal{L}_A\circ E^\mathcal{L}_B = E^\mathcal{L}_B\circ E^\mathcal{L}_A
	\end{equation}
\end{lemma}
\begin{proof}
	Since $d_N$ is a factor state, we have
	\begin{equation*}
		d_A = d_{A\cap B} \otimes d_{A - B}, d_B = d_{A\cap B}\otimes d_{B - A}, d_{A\cup B} = d_A \otimes d_{B - A} = d_{A-B} \otimes d_B
	\end{equation*}
	Hence we have
	\begin{align*}
		\begin{split}
			\mathcal{E}_A\circ\mathcal{E}_B(x) &= \tau_A(d_A^{1/2}\tau_B(d_B^{1/2}xd_B^{1/2})d_A^{1/2}) = \tau_{A\cap B}\otimes\tau_{A- B}(d_{A\cap B}^{1/2}\otimes d_{A-B}^{1/2}\tau_B(d_B^{1/2}xd_B^{1/2})d_{A\cap B}^{1/2}\otimes d_{A - B}^{1/2})
			\\
			&=\tau_{A\cap B}\otimes \tau_{A - B}(d^{1/2}_{A\cap B}\tau_B(d_{A-B}^{1/2}\otimes d_B^{1/2}xd_B^{1/2}d_{A-B}^{1/2})d_{A\cap B}^{1/2})
			\\
			&=\tau_{A\cap B}(d_{A\cap B})\tau_{A-B}\otimes \tau_B(d_{A\cup B}^{1/2}x d_{A\cup B}^{1/2}) = \tau_{A\cup B}(d_{A\cup B}^{1/2}x d_{A\cup B}^{1/2}) = \mathcal{E}_{A\cup B}(x)
		\end{split}
	\end{align*}
	Here in the third equation, we used the bimodule property of minimal conditional expectation.
	By the same calculation, we also have $\mathcal{E}_B\circ\mathcal{E}_A = \mathcal{E}_{A\cup B}$.
	
	For the local projections, we have
	\begin{align*}
		\begin{split}
			E^\mathcal{L}_A\circ E^{\mathcal{L}}_B = \lim_{m\rightarrow \infty}\lim_{n\rightarrow \infty}\mathcal{E}_A^m \mathcal{E}_B^n = \lim_{m\rightarrow\infty}\lim_{n\rightarrow\infty}\mathcal{E}_B^n\mathcal{E}_A^m = E^\mathcal{L}_B\circ E^\mathcal{L}_A
		\end{split}\qedhere
	\end{align*}
\end{proof}
Recall the definition of conditional covariance \cite{KB}\cite{BCR}.
\begin{definition}
	Fix a faithful state $\rho$ and fix $A\ssubset[|N|]$ then conditional covariance (with respect to $E^\mathcal{L}_A$) on $A$ is given by 
	\begin{equation*}
		Cov^\rho_A(x, y) := \langle x - E^\mathcal{L}_A(x), y - E^\mathcal{L}_A(y)\rangle_\rho
	\end{equation*}
\end{definition}
One can also define conditional covariance with respect to the local projections by replacing $E^\mathcal{L}_A$ with $\mathcal{E}_A$. It turns out that the two definitions are equivalent for our purpose \cite{KB}.
\begin{corollary}\label{corollary:cov}
	Let $A, B\ssubset[|N|]$ be two subsets, and fix the Gibbs state $d_N$. Then\begin{equation}
		Cov^N_{A\cup B}(E^\mathcal{L}_A(x), E^\mathcal{L}_B(x)) = 0
	\end{equation}
\end{corollary}
\begin{proof}
	\begin{align*}
		\begin{split}
			Cov^N_{A\cup B}(E^\mathcal{L}_A(x), E^\mathcal{L}_B(x)) &=\langle E^\mathcal{L}_A(x)-E^\mathcal{L}_{A\cup B}E^\mathcal{L}_A(x), E^\mathcal{L}_B(x) - E^\mathcal{L}_{A\cup B}E^\mathcal{L}_B(x)\rangle_N
			\\
			&=\langle E^\mathcal{L}_A(x - E^\mathcal{L}_{B}(x)), E^\mathcal{L}_B(x - E^\mathcal{L}_A(x))\rangle_N
			\\
			&=\langle x - E^\mathcal{L}_B(x), E^\mathcal{L}_A E^\mathcal{L}_B(x - E^\mathcal{L}_A(x))\rangle_N
			\\
			&=\langle x - E^\mathcal{L}_B(x), E^\mathcal{L}_{A\cup B}(x - E^\mathcal{L}_A(x))\rangle_N = 0\qedhere\pl.
		\end{split}
	\end{align*}
\end{proof}
Recall the definition of \textit{strong clustering}\cite{KB}.
\begin{definition}
	Let $A, B\ssubset [|N|]$ be two subsets such that $A\cap B \neq \emptyset$ and let $d_N$ be the Gibbs state. Then $d_N$ satisfies strong clustering (with respect to the minimal conditional expectation) if there exist constants $c,\xi > 0$ such that for any $x\in \mathbb{B}(V^{\otimes N})$
	\begin{equation*}
		Cov^N_{A\cup B}(E^\mathcal{L}_A(x), E^\mathcal{L}_B(x)) \leq c\langle x, x\rangle_Ne^{-d(B - A, A - B)/ \xi}
	\end{equation*}
	where $d(X,Y):=\min_{\substack{x\in X, y\in Y}}|x - y|$ is the distance between the two subsets $X, Y \ssubset[|N|]$.
\end{definition}
From Corollary \ref{corollary:cov}, the factor state $d_N$ satisfies strong clustering. By \cite[Thm 23]{KB}, we should expect that the spectral gap of $\mathcal{L}^\beta_N$ is independent of the system size. Compare with the corollary \ref{corollary:dimdep}, we see that the non-primitivity essentially closes the spectral gap.
\section{\large Discussion}\label{section:discussion}
In this paper, we studied Lindbladians in Davies form which are not necessarily \textit{geometrically local}. These Lindbladians are local in the sense of \cite{KBGKE11} and can be efficiently simulated by quantum circuits. When these Lindbladians are generated by $\mathfrak{su}(2)$-representations, we gave a general framework to estimate their CLSI constants from the uniform spectral gap of the absolute values of the $\mathfrak{su}(2)$-generators. These new examples of Lindbladians in Davies form are not primitive for any system size larger than 1. Therefore strong clustering of the equilibrium Gibbs state does not imply that the spectral gap is independent of system size. Although we focused on $\mathfrak{su}(2)$-representations, a careful analysis using representation theory of Lie algebras can generalize our results to more general simple Lie algebras. But this is beyond the scope of the current work.

Using this general framework, we also found non-local Lindbladians that have dimension-independent CLSI constant. These Lindbladians are generated by $\pi_N(a)(\pi_N(a^*)\pi_N(a))^{\gamma}$\ref{theorem:evaporate}. To see how these Lindbladians may arise in physics, consider the following system-bath coupling:
\begin{equation*}
	\pi_N(a^*)\pi_N(a)\pi_N(a^*)\otimes \pi_{M-N}(a) + h.c.
\end{equation*}
where the system is embedded in a larger 1-dimensional lattice: $[|N|]\ssubset[|M|]$, and $h.c.$ stands for Hermitian conjugate. Then if we take $\pi_N(h)$ to be the system Hamiltonian, following the derivation of \cite{D83} the Davies generator is given by $\mathcal{L}^{(3)}_N$. This system-bath interaction can describe the scattering of a pair of particles in system $N$ with one of the particles escaping the system. Our analysis shows that even in the infinite size limit, regardless of the initial state, the system eventually decays to a stationary state. This final state depends on the initial state, but the decay rate is of order 1. Since (ultra-weak closure of) the limit $\lim_{N\rightarrow\infty}\mathbb{M}_2^{\otimes N}$ with respect to the reference state $d_N$ is a type $III_\lambda$ factor where $\lambda := e^{-\beta}$, it is interesting to see how our result can be used to study entropy decay in quantum field theory (for example, black hole evaporation).

\appendix
\section{Proof of Proposition \ref{proposition:saloffcoste}}
The proof of Proposition \ref{proposition:saloffcoste} has two parts. First we recall the main object of interest. $\Omega_\oplus = \bigoplus_\lambda\mathbb{C}$ is the commutative algebra that we will focus on in this appendix. The reference state $\rho_N$ defines a probability measure on $\Omega_N$:
\begin{equation}
	\mu_\lambda := \rho_N|_{\mathbb{C}_\lambda} = \sum_{0\leq j \leq n_\lambda}\frac{\exp(-\frac{\beta}{2}(n_\lambda-2j))}{(2\cosh(\beta/2))^N}\dim W_\lambda
\end{equation}
where $\mathbb{C}_\lambda$ denotes the $\lambda$-th summand in $\Omega_\oplus$, $n_\lambda + 1 = \dim V_\lambda$ and $W_\lambda$ is the multiplicity subspace corresponding to the irreducible representation $V_\lambda$. We shall need the explicit formula for $\dim W_\lambda$ \cite{Fulton}:
\begin{equation}
	\dim W_\lambda = \frac{n_\lambda+1}{N+1}\binom{N+1}{\frac{N - n_\lambda}{2}}\pl.
\end{equation}

We want to study the entropy decay property of the self-adjoint channel (c.f. Proposition \ref{proposition:saloffcoste}):
\begin{equation}
	\Phi = E_\oplus E^\mathcal{F}_{diag} E_\oplus: \Omega_\oplus \rightarrow \Omega_{\oplus}\pl.
\end{equation}
In this appendix, an alternative formula for $\Phi$ is more useful:
\begin{lemma}\label{lemma:alternative}
	By Lemma \ref{lemma:miraculouscalc}, we have the following formula:
	\begin{equation}
		\Phi = E_{fix}E_\sigma E_{fix}
	\end{equation}
	where $E_{fix}$ is the conditional expectation onto the fixed point algebra and $E_\sigma$ is the conditional expectation onto the center of the spectral subalgebra $M$ (c.f. See the discussion before Lemma \ref{lemma:inclusion}). 
\end{lemma}
\begin{proof}
	Following the same notation as Lemma \ref{lemma:miraculouscalc}, we have:\begin{equation}
		\Phi = E_\oplus E^\mathcal{F}_{diag} E_\oplus = E_{fix}E_{diag}E^\mathcal{F}_{diag} E_{diag}E_{fix} = E_{fix}E_\sigma E_{fix}\pl.
	\end{equation}
\end{proof}

Since the channel $\Phi$ is nothing but a Markov chain, we can calculate its matrix representation. The explicit matrix representation will be important for the return time calculation. Let $\Omega_N := \{n: 0\leq n \leq N, n = N \mod 2\}$. Then $\Omega_\oplus = \ell_\infty(\Omega_N)$. Let $e_n$ be the $n$-th basis in $\ell_\infty(\Omega_N)$. Then under this basis, the matrix representation of $\Phi$ is given by:
\begin{lemma}\label{lemma:commutativechannelcalc}
	Under the basis $\{e_n\}_{n\in\Omega_N}$, the matrix entries of $\Phi$ are given by:
	\begin{equation}
		\Phi_{m,n} = \sum_{0\leq j \leq n} \frac{\dim W_n}{\dim H_{w(j)}}\frac{\exp(-\beta w(j)/2)}{\sum_{0\leq k \leq m} \exp(-\beta(m - 2k)/2)}
	\end{equation}
	where $w(j):= n-2j$ is the weight of the canonical basis $\ket{n,j}$, $W_n$ is the multiplicity subspace of the unique $(n+1)$-dimensional irreducible representation, and $H_{w(j)} =\{ \xi \in V^{\otimes N}: \pi_N(h)\xi = w(j)\xi\}$ is the eigen-subspace of $\pi_N(h)$.
\end{lemma}
\begin{proof}
	In terms of the canonical basis of the tensor product representation, $e_n$ can be written as:
	\begin{equation}
		e_n = \sum_{0\leq j \leq n}\ket{n,j}\bra{n,j}\otimes id_{W_n} := \sum_{0\leq j \leq n}x_j
	\end{equation}
	where for each j, we define $x_j := \ket{n,j}\bra{n,j}\otimes id_{W_n} \in \Omega_{diag}$. By Lemma \ref{lemma:alternative}, since $\Omega_{\oplus} = E_{fix}E_{diag}$, $\Phi$ acting on $\Omega_{\oplus}$ equals to $E_{fix}E_{\sigma}$. The action of $E_\sigma$ on $x_j$ is given by:
	\begin{equation}
		E_\sigma x_j = \frac{\dim W_n}{\dim H_{w(j)}}id_{H_{w(j)}} = \frac{\dim W_n}{\dim H_{w(j)}}\sum_{\substack{m,k\\ m-2k = w(j)}}\ket{m,k}\bra{m,k}\otimes id_{W_m}\pl.
	\end{equation}
	The action of $E_{fix}$ on $id_{H_w}$ is given by:
	\begin{equation}
		E_{fix}id_{H_w} = \sum_{m \geq w} \frac{\exp(-\beta w/2)}{\sum_{0\leq k \leq m}\exp(-\beta(m-2k)/2)}e_m\pl.
	\end{equation}
	Therefore we have:
	\begin{equation}
		\Phi(e_n) = E_{fix}E_\sigma e_n = \sum_{0\leq j \leq n}\frac{\dim W_n}{\dim H_{w(j)}}\sum_{m \geq w(j)}\tau_{m,k}e_m
	\end{equation}
	where we define $\tau_{m,k} := \frac{\exp(-\beta(m-2k)/2)}{\sum_{0\leq k'\leq m}\exp(-\beta(m - 2k')/2)}$ and $m - 2k = w(j)$. Then the claim is clear. 
\end{proof}
The main idea of the proof is to apply Gaussian comparison technique \cite{Horm}\cite{GJL} to estimate the logarithmic Sobolev constant of $\Phi$ and use the Diaconis-Saloff-Coste return time estimates \cite{DSC}. Then Proposition \ref{proposition:saloffcoste} follows from Lemma \ref{lemma:iterationorder}. 

In the remainder of this appendix, we frequently use the following notation:
\begin{definition}
	If there exists a constant $c> 0$ such that two variables $x,y$ satisfy: $c^{-1}x \leq y \leq cx$, then we will denote this equivalence relation: $x\sim_c y$. If there exists a constant $c > 0$ such that two measures $\mu, \nu$ on a finite space satisfy: $c^{-1}\mu \leq \nu \leq c\mu$, then we will denote this equivalence relation: $\mu \sim_c\nu$.
\end{definition}
In order to compare the reference measure $\mu$ with a Gaussian measure, we first change the measure $\mu$ to a simpler equivalent measure $\hat{\mu}_n := \frac{\exp(\beta n/2)}{(2\cosh(\beta / 2))^N}\dim W_n$.
\begin{lemma}
	There exists a constant $c > 0$ such that $\hat{\mu} \sim_c \mu $.
\end{lemma}
\begin{proof}
	This is a simple consequence of the summation of the geometric series:\begin{equation}
		\mu_n = \frac{\sum_{0\leq j \leq n}\exp(-\frac{\beta}{2}(n - 2j))}{(2\cosh(\beta / 2))^N}\dim W_n = \frac{\sinh(\beta(n+1)/2)}{\sinh(\beta / 2)(2\cosh(\beta / 2))^N}\dim W_n\pl.
	\end{equation}
	Since $\sinh(\beta(n+1)/2)/\sinh(\beta / 2)\sim_c \exp(\beta n /2)$, then we have the claim.
\end{proof}
We collect a few facts about the measure $\hat{\mu}$ in the following lemma.
\begin{lemma}\label{lemma:measurefacts}
	Define $s_n$ such that $e^{-s_n^2} = \hat{\mu}_n$ and define $\delta_n = s_{n+2}^2 - s^2_n$. Then we have:
	\begin{enumerate}
		\item $(\delta_n)$ is monotonely increasing;
		\item For all $n\in \Omega_N$, $s^2_n = -\log \hat{\mu}_n \sim_c C(\beta)N$ where $C(\beta)$ is a constant depending on $\beta$;
		\item The second-order difference sequence $(\delta_{n+2} - \delta_n)$ is monotonly increasing. And for all $n\in \Omega_N$, $\delta_{n+2} - \delta_n \geq \frac{C(\beta)}{N}$ where $C(\beta)$ is a constant depending on $\beta$.
	\end{enumerate}
\end{lemma}
\begin{proof}
	By definition, $\delta_n = \log(\frac{\hat{\mu}_n}{\hat{\mu}_{n+2}})  = -\beta + \log\frac{\dim W_n}{\dim W_{n+2}}$. Since $\frac{\dim W_n}{\dim W_{n+2}} = \frac{n+1}{n+3}\frac{N+n+4}{N-n}$, we have:\begin{equation*}
		\delta_n = -\beta+ \log\frac{n+1}{n+3}  + \log\frac{N+n+4}{N-n} = -\beta + \log(1 - \frac{2}{n+3}) + \log(-1 + \frac{2N+4}{N-n})\pl.
	\end{equation*}
	Thus as $n$ increases, $\delta_n$ increases.
	
	Since $s_n^2 = -\log \hat{\mu}_n$, we have $s_n^2 =-\frac{\beta}{2} n - \log\dim W_n + M\log(2\cosh(\beta / 2))$. To estimate the order of $s_n$, we use the following approximation of the factorial: $\sqrt{2\pi n}(\frac{n}{e})^n e^{\frac{1}{12n + 1}} < n! < \sqrt{2\pi n}(\frac{n}{e})^n e^{\frac{1}{12n}}$
	\begin{align*}
		\begin{split}
			\log\dim W_n &= \log(n+1) - \log(N+1) + (N+1)\log(N+1) - (N+1) + \frac{1}{2}\log(N+1) \\&- \frac{N-n}{2}\log\frac{N-n}{2} + \frac{N-n}{2} + \frac{1}{2}\log\frac{N-n}{2} - \frac{N+n+2}{2}\log\frac{N+n+2}{2} \\&+\frac{1}{2}\log\frac{N+n+2}{2} + O(\frac{1}{N})\\&
			=\log(n+1) + \frac{1}{2}\log(N+1) + \frac{n+1}{2}\log\frac{N-n}{N+n+2} + N\log\frac{N+1}{\sqrt{\frac{N-n}{2}}\sqrt{\frac{N+n+2}{2}}} \\& +O(\frac{1}{N})\pl.
		\end{split}
	\end{align*}
	It is clear that the leading order term is linear in $N$. Hence we can find a $\beta$-dependent constant $C(\beta) > 0$ such that: $s_n^2 \sim_c C(\beta)N$.
	
	Finally, we study the second-order difference:
	\begin{align*}
		\begin{split}
			\delta_{n+2} - \delta_n&= \log\frac{(n+1)(n+5)}{(n+3)^2} + \log\frac{(N+n+4)(N+n+6)}{(N-n-2)(N-n)} \\
			&=\log(1 - \frac{4}{(n+3)^2}) +  \log\frac{(N+n+4)(N+n+6)}{(N-n-2)(N-n)}\pl.
		\end{split}
	\end{align*}
	Hence it is clear that $\delta_{n+2} - \delta_n$ is monotonely increasing with $n$. Therefore for all $n \in \Omega_N$ we have: $\delta_{n+2} - \delta_n \geq \delta_{3} - \delta_1$ when $N$ is odd and $\delta_{n+2} - \delta_n \geq \delta_2 - \delta_0$ when $N$ is even.
	
	We can calculate the lower bound explicitly:
	\begin{align*}
		\begin{split}
			& \delta_3 - \delta_1 = \log\frac{12}{25} + \log(1 + \frac{8}{N-3}) + \log(1 + \frac{8}{N-1})\\
			&\delta_2 - \delta_0 = \log\frac{5}{9} + \log(1 + \frac{6}{N-2}) + \log(1 + \frac{6}{N})\pl.
		\end{split}
	\end{align*}
	Therefore there exists a constant $C(\beta) > 0$ such that for all $n\in \Omega_N$, $\delta_{n+2} - \delta_n \geq \frac{C(\beta)}{N}$.
\end{proof}
To apply Gaussian comparison technique, the key observation is the following fact about Gaussian integral:
\begin{align}\label{equation:localGaussian}
	\begin{split}
		\int_t^{t+\delta} e^{-x^2}dx & \sim_c \delta \exp(-t^2) \text{ , if } \delta t \leq 1 
		\\
		& \sim_c \frac{\exp(-t^2)}{t} \text{ , if } \delta t \geq 1\pl.
	\end{split}
\end{align}
Therefore, for each $n\in \Omega_N$, we must construct the right coefficient $t_n$ such that the measure $\hat{\mu}_n \sim_c (t_{n+1} - t_n)\exp(-t_n^2)$ if $(t_{n+1} - t_n)t_n \leq 1$ or $\hat{\mu}_n \sim_c (t_{n+1} - t_n)\frac{\exp(-t_n^2)}{t_n}$ if $(t_{n+1} - t_n)t_n \geq 1$. The following theorem gives this key construction.
\begin{theorem}\label{theorem:LSIcoefficient}
	Let $(\Omega = \{0,1,...,N\}, \mu)$ be a finite discrete measure space where the measure $\mu_n$ is monotonely decreasing with $n$. Assume $\min_n\log\frac{\mu_n}{\mu_{n+1}} < 2$. Let $s_n$ be such that $e^{-s_n^2} = \mu_n$. Assume there exists a constant $n_0$ such that for all $n\geq n_0$ we have $s_n \geq 1$. In addition assume there exists a constant $\sigma > 0$ such that $(s_{n+1} - s_n)s_n\geq \sigma$. Then the following logarithmic Sobolev inequality (LSI) holds:
	\begin{equation}
		D(f^2 | E_\mu f^2) \leq \frac{C}{\sigma}\sum_{n\in \Omega}\frac{\mu_n}{(s_{n+1} - s_n)^2}(f(n+1) - f(n))^2 
	\end{equation}
	where $C>0$ is an absolute constant and $f \in \ell_\infty(\Omega)$.
\end{theorem}
\begin{proof}
	We consider the following ansatz for the coefficients $t_n$:\begin{equation}
		t_n^2 := s_n^2 - \log s_n\pl.
	\end{equation}
	Since for $x > 0$ the function $x^2 - \log x > 0$, the coefficients $(t_n)$ are well-defined. The key observation is the following identity:
	\begin{equation}
		\frac{\exp(-t^2_n)}{t_n} = \frac{s_n}{t_n}\mu_n = \mu_n \frac{1}{\sqrt{1 - \frac{\log s_n}{s_n}}}\pl.
	\end{equation}
	Consider the function $f(x):=\frac{\log x}{x}$. If $x\geq 1$, then $f(x)$ has a unique maximum at $x = e$ with value $1/e$. Since there exists $n_0$ such that for all $n\geq n_0$ we have $s_n \geq 1$, then there exists a constant $c>0$ such that $\frac{\exp(-t_n^2)}{t_n} \leq c\mu_n$ for all $n\in \Omega$. On the other hand, since for all $n \geq n_0$ we have $t_n^2 \leq s_n^2$. Then there exists a constant $c' > 0$ such that $c' \mu_n \leq \frac{\exp(-t_n^2)}{t_n}$ for all $n\in \Omega$. Hence $\mu_n\sim_c \frac{\exp(-t_n^2)}{t_n}$ for all $n\in \Omega$.
	
	In addition, for all $n\geq n_0$,  we have $\frac{1}{2}s_n^2 \leq t_n^2 \leq s_n^2$. Here we used the fact that the function $x^2/2 - \log(x)$ is strictly bounded below by $1/2$ when $x\geq 1$. Hence $t_n \sim_c s_n$.
	
	We also need to show $(t_{n+1} - t_n) \sim_c (s_{n+1} - s_n)$. Since $\mu_n$ is monotonely decreasing, then $s_n$ is monotonely increasing. For the sequence $(t_n)$ we have:
	\begin{equation}
		t_{n+1} - t_n = \frac{t_{n+1}^2 - t_n^2}{t_{n+1} + t_n} = \frac{s^2_{n+1} - s^2_n - \log\frac{s_{n+1}}{s_n}}{t_{n+1} + t_n}\pl.
	\end{equation}
	Since $s_{n+1} \geq s_n$ and $t_n \sim_c s_n$, we have:
	\begin{equation}
		t_{n+1} - t_n \leq \frac{s^2_{n+1} - s^2_n}{t_{n+1} + t_n} \sim_c s_{n+1} - s_n\pl.
	\end{equation}
	On the other hand, for all $n\geq n_0$, we have:
	\begin{equation}
		\frac{s_{n+1}}{s_n}  = 1 + \frac{s_{n+1} - s_n}{s_n}  = 1+ \frac{\log(\frac{\mu_n}{\mu_{n+1}})}{s_n(s_{n+1}+s_n)} \leq 1 + \frac{1}{2}\log\frac{\mu_n}{\mu_{n+1}}\pl.
	\end{equation}
	Thus for all $n\geq n_0$ we have:
	\begin{equation}
		t_{n+1 } - t_n \geq \frac{s^2_{n+1} - s^2_n - \log(1 + \frac{1}{2}\log\frac{\mu_n}{\mu_{n+1}})}{t_{n+1} + t_n} = \frac{\log\frac{\mu_n}{\mu_{n+1}} - \log(1 + \frac{1}{2}\log\frac{\mu_n}{\mu_{n+1}})}{t_{n+1} + t_n} \pl.
	\end{equation}
	For all $x > 0$ the function $f(x):= x - \log(1+ x/2)$ is monotonely increasing and $f(x) \geq \frac{x}{2}$ for $0 < x < 2$. Since $\min_n \log\frac{\mu_n}{\mu_{n+1}} < 2$, then we have:
	\begin{equation}
		t_{n+1} -t_n \geq \frac{s^2_{n+1} - s^2_n}{2(t_{n+1} + t_n)} \sim_c \frac{1}{2}(s_{n+1} - s_n)\pl.
	\end{equation}
	Therefore $(t_{n+1} - t_n) \sim_c (s_{n+1} - s_n)$. Combining these estimates, we have:
	\begin{equation}
		t_n(t_{n+1} - t_n)\sim_c s_n(s_{n+1} - s_n) \geq \sigma\pl.
	\end{equation}
	In particular, if $t_n(t_{n+1} - t_n)\geq 1$, then $\int_{t_n}^{t_{n+1}}e^{-x^2}dx \sim_c\frac{e^{-t_n^2}}{t_n}\sim_c\mu_n$. If $t_n(t_{n+1}-t_n) \leq 1$, then $\int_{t_n}^{t_{n+1}}e^{-x^2}dx \sim(t_{n+1} - t_n)e^{-t_n^2}\gtrapprox \sigma \frac{e^{-t_n^2}}{t_n}\sim\sigma\mu_n$. On the other hand, $\int_{t_n}^{t_{n+1}}e^{-x^2}dx \leq (t_{n+1} - t_n)e^{-t_n^2}\leq \frac{e^{-t_n^2}}{t_n}\sim_c\mu_n$. Therefore there exists an absolute constant $c > 0$ such that:
	\begin{equation}
		c^{-1}\sigma\mu_n \leq \gamma([t_n, t_{n+1}])\leq c\sigma^{-1}\mu_n
	\end{equation}
	where $\gamma$ is the standard Gaussian measure on $\mathbb{R}$.

	Now following the Gaussian comparison technique developed in \cite{Horm}\cite{GJL}, we consider the transference map: $\pi:\ell_\infty(\Omega)\rightarrow L^\infty(\mathbb{R})$. We partition each interval $[t_n, t_{n+1}] = I_n\cup J_n$ where the length of $J_n$ is $(t_{n+1} - t_n)/2$. The function $\pi(f)(t)$ is piecewise-linear. $\pi(f)(t) = f_n$ for $t\in I_n$ and $\pi(f)(t) = \frac{2(f_{n+1} - f_n)}{t_{n+1} - t_n}(t - t_n) + f_n$ for $t\in J_n$. $\pi(f)(t) = f_0$ for $t \leq t_0$ and $\pi(f)(t) = f_N$ for $t\geq t_N$. By the chain rule of relative entropy and non-negativity of relative entropy \cite{JLR}, we have:
	\begin{align}
		\begin{split}
			D(f^2 | E_\mu f^2) &\leq D(\pi(f)^2 | E_\gamma \pi(f)^2) \leq C\int_{-\infty}^\infty |\pi(f)'(t)|^2 d\gamma(t) \\&
			= C\sum_n \frac{\gamma(J_n)}{(t_{n+1} - t_n)^2}(f_{n+1} - f_n)^2 \\&\leq C\sigma^{-1}\sum_n\frac{\mu_n}{(s_{n+1} - s_n)^2}(f_{n+1} - f_n)^2
		\end{split}
	\end{align}
	where $\gamma$ is the standard Gaussian measure on $\mathbb{R}$ and the second equation follows from the classical logarithmic Sobolev inequality of the Gaussian measure on $\mathbb{R}$. The last inequality follows from the equivalence $(t_{n+1}-t_n)\sim (s_{n+1}-s_n)$.
\end{proof}
In our application, the measure $\hat{\mu}$ is Bell-shaped. An immediate corollary of Theorem \ref{theorem:LSIcoefficient} is the following observation:
\begin{corollary}\label{corollary:bellshaped}
	Using the same notation as Theorem \ref{theorem:LSIcoefficient}. If the measure $\mu$ is Bell-shaped with the unique maximum at $n_0\in \Omega$. Assume the weight at $n_0$ is bounded above: $\mu_{n_0} \leq 1$. In addition, assume $\min_{n > n_0}\log\frac{\mu_n}{\mu_{n+1}} < 2$ and $\min_{n < n_0}\log \frac{\mu_{n-1}}{\mu_n} < 2$. And all the other assumptions in Theorem \ref{theorem:LSIcoefficient} hold without modification. Then the following logarithmic Sobolev inequality holds:
	\begin{equation}
		D(f^2 | E_\mu f^2) \leq \frac{C}{\sigma}\sum_{n\in \Omega}\mu_n \big(\frac{(f_{n+1} - f_n)^2}{(s_{n+1} - s_n)^2} + \frac{(f_n - f_{n-1})^2}{(s_n - s_{n-1})^2}\big)
	\end{equation}
	where $C > 0$ is an absolute constant and $f\in \ell_\infty(\Omega)$.
\end{corollary}
\begin{proof}
	The proof is almost the same as Theorem \ref{theorem:LSIcoefficient}. For each $\mu_n$, we define $s_n$ and $t_n$ as in the proof of Theorem \ref{theorem:LSIcoefficient}. Then for $n < n_0$, reversing the construction in Theorem \ref{theorem:LSIcoefficient}, we can show that there exists a constant $c > 0$ such that:
	\begin{equation}
		c^{-1}\sigma \mu_n \leq \gamma([-t_{n-1}, -t_n]) \leq c\sigma^{-1} \mu_n\pl.
	\end{equation}
	Note for $n < n_0$, $(t_n)$ is monotonely decreasing.	For $n > n_0$, Theorem \ref{theorem:LSIcoefficient} applies directly and we have:
	\begin{equation}
		c^{-1}\sigma \mu_n \leq \gamma([t_n, t_{n+1}]) \leq c\sigma^{-1}\mu_n \pl.
	\end{equation}
	Here $\gamma$ is the standard Gaussian measure on $\mathbb{R}$.
	
	For each $f\in \ell_\infty(\Omega)$, we define the following transfered function $\pi(f)\in L^\infty(\mathbb{R})$. For $n < n_0$, we partition each interval $[-t_{n-1} + \mu_{n_0}, -t_n + \mu_{n_0}] = I_n \cup J_n$ where the length of $J_n$ is $(t_{n-1} - t_{n})/2$. For $t\in I_n$, $\pi(f)(t) = f_{n-1}$ and for $t\in J_n$, $\pi(f)(t) = \frac{2(f_n - f_{n-1})}{t_{n-1} - t_n}(t + t_{n-1} - \mu_{n_0}) + f_{n-1}$. For $n > n_0$, we partition each interval $[t_n - \mu_{n_0}, t_{n+1} - \mu_{n_0}] = I_n \cup J_n$ where the length of $J_n$ is $(t_{n+1} - t_n)/2$. For each $t\in I_n$, $\pi(f)(t) = f_n$ and for $t\in J_n$, $\pi(f)(t) = \frac{2(f_{n+1} - f_n)}{t_{n+1} - t_n}(t - t_n + \mu_{n_0}) + f_n$. This finishes the construction of $\pi(f)$ outside the interval $[-\mu_{n_0}, \mu_{n_0}]$. Inside this interval, we do the same partition as before and define $\pi(f)$ by linear interpolation: $\pi(f)(t)$ linearly interpolates between $f_{n_0 - 1}$ and $f_{n_0}$ for $-\mu_{n_0} \leq t \leq 0$. And $\pi(f)$ linearly interpolates between $f_{n_0}$ and $f_{n_0 +1}$ for $0\leq t \leq \mu_{n_0}$.
	
	Then by the chain rule of relative entropy and non-negativity of relative entropy \cite{JLR}, we have:
	\begin{align}\label{equation:longsplit}
		\begin{split}
			D(f^2 | E_\mu f^2)  &\leq D(\pi(f)^2 | E_\gamma \pi(f)^2) \leq C\int_{-\infty}^\infty |\pi(f)'(t)|^2d\gamma(t) \\
			&= C\sum_{n}\frac{\gamma(J_n)}{(t_{n+1} - t_n)^2}(f_{n+1} - f_n)^2\\
			& \leq \frac{C}{\sigma}\bigg(\sum_{n < n_0}\frac{\mu_n}{(s_{n} -s _{n-1})^2}(f_n - f_{n-1})^2 + \sum_{n > n_0} \frac{\mu_n}{(s_{n+1} -s_n)^2}(f_{n+1} - f_n)^2\\& + \max\{\frac{\gamma([-\mu_{n_0}, 0])}{\mu_{n_0}}, \frac{\gamma([0, \mu_{n_0}])}{\mu_{n_0}}\}\mu_{n_0}(\frac{(f_{n_0} - f_{n_0 -1})^2}{(s_{n_0} - s_{n_0-1})^2}+\frac{(f_{n_0 + 1} - f_{n_0})^2}{(s_{n_0+1}-s_{n_0})^2})\bigg)\pl.
		\end{split}
	\end{align}
	Since the Gaussian measure is equivalent to the uniform measure in $[-1,1]$, we have:
	\begin{align}
		\begin{split}
			D(f^2 | E_\mu f^2) \leq \frac{C}{\sigma}\sum_{n\in \Omega}\mu_n \big(\frac{(f_{n+1} - f_n)^2}{(s_{n+1} - s_n)^2} + \frac{(f_n - f_{n-1})^2}{(s_n - s_{n-1})^2}\big)
		\end{split}
	\end{align}
	where we used the fact that each summand is non-negative and hence each summation in Equation \ref{equation:longsplit} is bounded above by $\sum_{n\in \Omega}\mu_n \big(\frac{(f_{n+1} - f_n)^2}{(s_{n+1} - s_n)^2} + \frac{(f_n - f_{n-1})^2}{(s_n - s_{n-1})^2}\big)$.
\end{proof}
We are now ready to apply the result to our example:
\begin{corollary}\label{corollary:equivalentGauss}
	Using the same notation as Theorem \ref{theorem:LSIcoefficient} and Lemma \ref{lemma:measurefacts}, for each $f\in \Omega_\oplus$ we have:
	\begin{equation}
		D(f^2 | E_\mu f^2) \leq C(\beta)N^4 \sum_{n}\mu_n (f_n - f_{n+2})^2
	\end{equation}
	where $C(\beta) > 0$ is a constant depending on $\beta$. 
\end{corollary}
\begin{proof}
	We work with the equivalent measure $\hat{\mu}$. Let $\hat{\mu}$ achieves maximum at $n_0\in \Omega_N$. Recall $\hat{\mu}_n = e^{-s_n^2}$ and $\delta_n = s^2_{n+2} -s^2_n$. Then $(s_{n+2} - s_n)s_n = \frac{s^2_{n+2} - s^2_n}{s_{n+2} + s_n}s_n$. By Lemma \ref{lemma:measurefacts}, for all $n$ we have $s_n^2 \sim_c C(\beta)N$. Therefore these exists a $\beta-$dependent constant $c(\beta) > 0$ such that:\begin{equation}
		c(\beta)^{-1}\delta_n \leq s_n(s_{n+2} - s_n)\leq c(\beta)\delta_n\pl.
	\end{equation}
	By Lemma \ref{lemma:measurefacts}, around the point $n_0$ we have:
	\begin{align}
		\delta_{n_0} \geq 0 \text{ , } \delta_{n_0 - 2} \leq 0 \text{ , } \delta_{n_0} - \delta_{n_0 - 2} \geq \frac{C(\beta)}{N}\pl.
	\end{align}
	There are two cases. If both $|\delta_{n_0}|$ and $|\delta_{n_0-2}|$ has a lower bound of the order $O(\frac{1}{N})$, then for either side of $n_0$ we can apply the same argument as in Theorem \ref{theorem:LSIcoefficient}. The lower bound is given by $\frac{C(\beta)}{N}$. Therefore we get the LSI constant $C(\beta)N$. If one of the coefficients $|\delta_{n_0}|$ or $|\delta_{n_0 - 2}|$ does not have a lower bound of the order $O(\frac{1}{N})$, then it is necessary that the other coefficient has a lower bound of the order $O(\frac{1}{N})$. Without loss of generality, assume $|\delta_{n_0-2}| \gtrapprox \frac{C(\beta)}{N}$. Then we apply the constructino in Corollary \ref{corollary:bellshaped} by placing a length $2\mu_{n_0}$ interval around the origin. The condition $|\delta_n| \gtrapprox \frac{C(\beta)}{N}$ is satisfied by $n < n_0$ and $n > n_0$. Then the construction in Corollary \ref{corollary:bellshaped} shows that the LSI constant is given by $C(\beta)N$. Therefore in both cases, we have:
	\begin{equation}\label{equation:twosided}
		D(f^2 | E_{\hat{\mu}}f^2) \leq C(\beta)N \sum_{n}\hat{\mu}_n(\frac{(f_{n+2}-f_n)^2}{(s_{n+2} - s_n)^2} + \frac{(f_n - f_{n-2})^2}{(s_{n-2} -s_n)^2})\pl.
	\end{equation}
	Since $s_{n+2} -s_n = \frac{\delta_n}{s_{n+2}+s_n}$, by Lemma \ref{lemma:measurefacts} we have $s_{n+2} -s_n \gtrapprox \frac{1}{N^{3/2}}$. Therefore we have:
	\begin{equation}
		D(f^2 |E_{\hat{\mu}} f^2) \leq C(\beta)N^4\sum_n \hat{\mu}_n(f_{n+2}-f_n)^2
	\end{equation}
	where we have used the fact that the two sums in Equation \ref{equation:twosided} are the same.
	
	Now by the chain rule of relative entropy and the equivalence of measure $\mu\sim_c\hat{\mu}$, we have the desired logarithmic Sobolev inequality.
\end{proof}

We now apply the result of Corollary \ref{corollary:equivalentGauss} to estimate the return time of the channel $\Phi$.
\begin{lemma}\label{lemma:return}
	Using the same notation as in Lemma \ref{lemma:commutativechannelcalc}. If there exists constants $C,\alpha > 0$ such that $\min_{n\in \Omega_N}\Phi_{n+2,n} \geq \frac{C}{N^\alpha}$, then the return time of $\Phi$ is of the order:
	\begin{equation}
		t(\Phi) = O(N^{\alpha + 4}(1+\log\log \min_n \frac{1}{\mu_n}))\pl.
	\end{equation}
\end{lemma}
\begin{proof}
	The proof follows from the classical result in \cite{DSC}. Consider the Lindbladian $id - \Phi:\Omega_\oplus \rightarrow \Omega_\oplus$. The energy form of $id - \Phi$ is bounded below by the energy form of the nearest-neighbor interactions:
	\begin{equation}
		\mathcal{E}_{id - \Phi}(f) = \sum_{n,m}\mu_n\Phi_{n,m}(f_n - f_m)^2 \geq \sum_n \mu_n \Phi_{n+2,n}(f_{n+2} - f_{n})^2\pl.
	\end{equation}
	By assumption, the nearest-neighbor coefficients $\Phi_{n+2,n}$ have a uniform lower bound. Then we have:
	\begin{equation}
		\mathcal{E}_{id - \Phi}(f) \geq \frac{C}{N^\alpha}\sum_n \mu_n (f_{n+2} - f_n)^2\pl.
	\end{equation}
	Combining with Corollary \ref{corollary:equivalentGauss}, we have:
	\begin{equation}
		D(f^2 | E_\mu f^2) \leq C(\beta) N^4\sum_n \mu_n (f_{n+2} -f_n)^2 \leq C(\beta)N^{4+\alpha}\mathcal{E}_{id - \Phi}(f)\pl.
	\end{equation}
	Therefore $LSI^{-1}(id -\Phi) \leq CN^{\alpha + 4}$. By Theorem 3.7 in \cite{DSC}, the semigroup $\{T_t\}_{t\geq 0}$ generated by $id - \Phi$ satisfies the return time estimate:
	\begin{equation*}
		||T_t - E_\mu: L_1(\Omega_N)\rightarrow L_\infty(\Omega_N)|| \leq e^{2-2c}
	\end{equation*}
	for $t  = \frac{c}{\lambda_2(id - \Phi)} + LSI^{-1}(id - \Phi)\log\log\min_n \frac{1}{\mu_n}$ where $c > 0 $ is a constant and $\lambda_2(id-\Phi)$ is the spectral gap. Since $\Phi$ is $\mu$-preserving and positive, by Corollary 2.2 in \cite{DSC} we have $||\Phi^N - E_\mu||_{1\rightarrow \infty} \leq ||T_{2N} - E_\mu||_{1\rightarrow \infty}$. Therefore the return time of $\Phi$ is of the order $O(N^{\alpha+ 4}(1 + \log\log\min_n \frac{1}{\mu_n}))$.
\end{proof}
Now we are ready to prove Proposition \ref{proposition:saloffcoste}.
\begin{proof}[Proof of Proposition \ref{proposition:saloffcoste}]
	To apply the return time estimate of Lemma \ref{lemma:return}, we need to verify that $\Phi_{n+2, n}$ indeed has a uniform lower bound. By lemma \ref{lemma:commutativechannelcalc}, we have:
	\begin{align}
		\begin{split}
			\Phi_{n+2,n} &= \sum_{0\leq j \leq n}\frac{\dim W_{n}}{\dim H_{w(j)}}\tau_{n,(n-w(j))/2} \geq \frac{\dim W_{n}}{\dim H_{-n}}\tau_{n,n} \geq e^{-\beta /2}\frac{(n+1)\binom{N+1}{(N-n)/2}}{(N+1)\binom{N}{(N-n)/2}}
			\\
			&=\frac{C(\beta)(n+1)}{N+n+2} \geq\frac{C(\beta)}{N}\pl.
		\end{split}
	\end{align}
	In addition, we have: $\min_n \log\log\frac{1}{\mu_n} = \min_n \log(-\frac{\beta n}{2}+ N \log(2\cosh(\frac{\beta}{2})) - \log \dim W_n) \gtrapprox \log(N)$. Then by Lemma \ref{lemma:return}, the return time of $\Phi$ is of the order $O(N^5(1 + \log N))$. By the results of \cite{GJLL}, after $O(N^5(1 + \log N))$-iterations, we have:
	\begin{equation}
		0.9 E_\mu \leq \Phi^{2k} \leq 1.1E_{\mu}\pl.\qedhere
	\end{equation}
\end{proof}
\section{Quantum Expanders}\label{appendix:expander}
In order to get rid of the nontrivial fixed point of $\mathcal{L}^\beta_N$, we used quantum expanders to efficiently reduce the size of this algebra. In this part of the appendix, we provide two proofs of Theorem \ref{theorem:expander}. 

First we can use the construction of quantum expanders in \cite{BST}. This construction is based on classical Ramanujan graphs. 
\begin{proof}[Proof of Theorem \ref{theorem:expander}]
	Given a constant $c\in (0,1)$, it is shown in \cite{BST} that for each finite dimensional Hilbert space $H$, there exists a set of unitaries $\{U_j\}_{1\leq j \leq m}$ where $m = O(\frac{1}{c^4})$ such that for all $x\in \mathbb{B}(H)$:
	\begin{equation}
		||\frac{1}{m}\sum_{1\leq j \leq m}U_j^*xU_j - \tau_H(x)||_2 \leq c||x||_2\pl.
	\end{equation}
	For spectral gap $c = 1- \epsilon$, the number of unitaries required is given by: $m = O(\frac{1}{c^4}) = O(1 + 4\epsilon) = O(1)$.
\end{proof}
Since the construction in \cite{BST} is not explicit, we give another proof of Theorem \ref{theorem:expander} using the random matrix construction due to Hastings \cite{Has}.
\begin{proof}
	Fix a small constant $\epsilon \in (0, \frac{\sqrt{5}}{3})$. Let $\nu_N$ be the normalized Haar measure on $SU(N)$. It was shown in \cite{Has} that there exists an $N_0 \in \mathbb{N}$ such that for all $N \geq N_0$, we have:
	\begin{equation}
		\nu_N^3(\{U_i \in SU(N): ||\frac{1}{6}\sum_{1\leq j\leq 6}U_j^*x U_j - \tau(x)||_2 \leq (1 - \frac{\sqrt{5}}{3} + \epsilon)||x||_2 \text{ , } 1\leq i \leq 3\}) \geq \frac{1}{2}
	\end{equation}	
	where the set of three unitaries is sampled from the normalized Haar measure on $SU(N)$ and the summation is over the set of six unitaries: $\{U_i, U_i^*\}_{1\leq i \leq 3}$. Therefore for $N \geq N_0$, we can sample three unitaries from $SU(N)$ and with high probability we will have a quantum expander with uniform spectral gap $\frac{\sqrt{5}}{3} - \epsilon$. 
	
	For $N < N_0$, we use the result of Bourgain and Gamburd \cite{BG1}\cite{BG2}. For each $N$, we consider a unitary representation of the free group $\mathbb{F}_3$:
	\begin{equation}
		\varphi_N: \mathbb{F}_3\rightarrow SU(N)
	\end{equation}
	such that for each generator $s_j$, the image $\varphi_N(s_j)$ is a unitary matrix with algebraic entries. Then on the fundamental representation $\mathcal{H}_N$ of $SU(N)$, there exists a constant $\lambda^{max}_N \in (0,1)$ such that for all $x\in \mathbb{B}(\mathcal{H}_N)$:
	\begin{equation}
		||\frac{1}{6}\sum_{1\leq j \leq 6}\varphi_N(s_j)^*x\varphi_N(s_j) - \tau(x)||_2 \leq \lambda^{max}_N ||x||_2
	\end{equation}
	where the Hecke operator is formed out of $\{\varphi_N(s_j), \varphi_N(s_j^{-1})\}_{1\leq j \leq 3}$. In principle the spectral gap $1 - \lambda^{max}_N$ depends on the dimension $N$. But since there are only finitely many $N\leq N_0$, there exists a constant $\eta:= \min\{\frac{\sqrt{5}}{3} - \epsilon, 1 - \lambda^{max}_N\}_{N < N_0}$ such that for all $N\in \mathbb{N}$ we can find a set of six unitaries closed under conjugation and they form a quantum expander with spectral gap $\eta$.
\end{proof}
The second proof has the advantage that the number of unitaries is independent of the spectral gap and the dimension. And the construction is more explicit than the one given in \cite{BST}.
\nocite{*}
\bibliographystyle{alpha}
\bibliography{spectral}
\end{document}